\documentclass[11pt,a4paper]{article}
\usepackage{amssymb,amsmath, amsfonts}
\usepackage{graphicx,graphics}
\usepackage{mathtools}
\usepackage{bbold}
\usepackage[english]{babel}
\usepackage[utf8]{inputenc} 

\usepackage{epsfig,url}
\usepackage{bbm,theorem}
\usepackage{a4wide}
\usepackage{color}
\usepackage{enumerate}
\usepackage{calrsfs}
\DeclareMathAlphabet{\pazocal}{OMS}{zplm}{m}{n}

\newtheorem{theorem}{Theorem}[section]
\newtheorem{definition}[theorem]{Definition}

\newtheorem{lemma}[theorem]{Lemma}
\newtheorem{proposition}[theorem]{Proposition}
\newtheorem{assumption}[theorem]{Assumption}

{\theorembodyfont{\upshape}
\newtheorem{remark}[theorem]{Remark}
\newtheorem{example}[theorem]{Example}
}
\numberwithin{equation}{section}
\numberwithin{theorem}{section}

\newcommand{\qed}{\hfill$\Box$}
\newenvironment{proof}{\begin{trivlist}\item[]{\em Proof:}\/}{\qed\end{trivlist}}
\newenvironment{proofof}[1][Proof]{\noindent \textit{#1.} }{\ \qed} %\rule{0.5em} {0.5em}}

\newcommand{\R}{{\mathbb R}}

\newcommand{\N}{{\mathbb N}}
\newcommand{\T}{{\mathcal T}}

\newcommand{\PP}{\mathbb{P}}
\newcommand{\EE}{\mathbb{E}}
\newcommand{\ep}{\varepsilon}

\newcommand{\B}{\mathcal{B}}

\newcommand{\mbf}{\mathbf}

\DeclareMathOperator*{\supp}{supp}

\DeclarePairedDelimiter\ceil{\lceil}{\rceil}

\newcounter{jlisti}

\title{On the growth of a %moving 
particle coalescing in a Poisson distribution of obstacles}

\author{Alessia Nota \thanks{\emailalessia} , Juan J. L. Vel\'azquez \thanks{\emailjuan} \\[1em]
$\,^\ddag$\UBaddress}
\date{\today}
\newcommand{\email}[1]{E-mail: \tt #1}
\newcommand{\emailjuan}{\email{velazquez@iam.uni-bonn.de}}
\newcommand{\emailalessia}{\email{nota@iam.uni-bonn.de}}
\newcommand{\UBaddress}{\em University of Bonn, Institute for Applied Mathematics\\
\em Endenicher Allee 60, D-53115 Bonn, Germany}

\date{\today}

\begin{document} \selectlanguage{english}
\maketitle

\begin{abstract}
In this paper we consider the \textit{coalescence} dynamics of a \textit{tagged particle} moving in a random distribution of particles with volumes independently distributed according to a probability distribution (\textit{CTP model}). 
%The particle coalesces with the particles colliding with it and increases its size correspondingly.
We provide a rigorous derivation of a kinetic equation for the probability density for the size and position of the tagged particle in the kinetic limit where the volume fraction $\phi$ filled by the background of particles tends to zero. Moreover, we prove that the particle system, i.e. CTP model, is well posed for a small but positive volume fraction with probability one as long as the the distribution of the particle sizes is compactly supported. %This result is a first contribution for the analysis of coagulation processes in laminar shear flows which lead to the derivation of the Smoluchowski equation with kernel $K\left(v,w\right)= \frac{4}{3} S (v^\frac{1}{3}+w^\frac{1}{3})^3$. 
\end{abstract}

\tableofcontents

\section{Introduction}\label{sec:1}
This paper is concerned with the growth of a tagged particle which moves in
three dimensions in a background of particles randomly distributed. It will be
assumed that the moving particle coalesces with the particles colliding with it
and increases its size correspondingly. In particular we will derive a kinetic
equation which gives the probability density for the size and position of the
tagged particle if the volume fraction $\phi$ filled by the background of
particles tends to zero.

\bigskip

The motivation to study this problem is the analysis of coagulation processes
in shear flows. This problem was already considered by Smoluchowski in his
seminal analysis of coagulation of particles moving in fluids, see \cite{Sm}.  
 The result of this analysis is the derivation of a kinetic equation for
the particle sizes having the form: 
\begin{equation} \label{SmolKin}
\partial_{t}f\left(  v,t\right)  =\frac{1}{2}\int_{0}^{v}K\left(  v-w,w\right)
f\left(  v-w,t\right)  f\left(  w,t\right)  dw- \int_{0}^{\infty}K\left(  v,w\right)
f\left(  v,t\right)  f\left(  w,t\right)  dw %
\end{equation}
where $f$ is the particle distribution in the space of volumes, and where the
coagulation kernel is given by $K\left(v,w\right)= \frac{4}{3} \tilde S (v^\frac{1}{3}+w^\frac{1}{3})^3$. Cf. \cite{Sm}, \cite{F}. Actually the equation derived in those papers corresponds to the case in which the particle sizes are discrete and equation \eqref{SmolKin} is the standard generalization to the continuum case. 

The derivation of (\ref{SmolKin}) is based on the following assumptions. It is
assumed that a set of spherical particles move in the three-dimensional space
with a horizontal velocity along the direction of the coordinate axis which is
proportional to the vertical component $x_{3},$ i.e. the speed $v=v\left(
x\right)  $ of a given particle at the point $x=\left(  x_{1},x_{2}%
,x_{3}\right)  $ is
\begin{equation}
v\left(  x\right)  =\left(  v_{1}\left(  x\right)  ,v_{2}\left(  x\right)
,v_{3}\left(  x\right)  \right)  =\left( \tilde{S} x_{3},0,0\right) \label{VelShear}%
\end{equation}
where $\tilde{S}= \frac{\partial v_1}{\partial x_3}$ is the coefficient associated to the laminar shear. 

In the derivation of (\ref{SmolKin}) it is assumed that the
particles do not affect the motion of the fluid flow. If the particles are initially
randomly distributed there would be particle collisions between pairs of
particles with different values of $x_{3}$ given that they approach to
each other since their relative velocity is nonzero in the velocity field
(\ref{VelShear}). It is assumed that no correlations develop between the
particles, something that might be expected in the kinetic regime in which the
volume fraction of the particles $\phi$ tends to zero. In order to obtain a non trivial dynamics for the particle sizes in times of order one we rescale the value of the shear $\tilde{S}$ and the particle density and the volume fraction $\phi$ in such a way that in average there is one collision for unit of time, namely $\tilde{S}\phi \sim 1$. This rescaling is the analogous of the Boltzmann-Grad limit in gas dynamics. The kinetic equation \eqref{SmolKin} will then be obtained in the limit $\phi\to 0$.

\bigskip
In this paper we will consider a simpler setting, namely a tagged particle with volume $\tilde{V}$
moving among a random distribution of particles at positions $\left\{
x_{j}\right\}  _{j\in N}$ and with
volumes $\left\{  \tilde{v}_{j}\right\}  _{j\in N}$, where $N$ is a countable set of indexes. Every time that the
tagged particle collides with any of the obstacles they merge. The resulting
new particle has a volume which is the sum of $\tilde{V}$ with the volumes of
the colliding particles, and the position of the new particle becomes the
center of mass of the group of particles involved in the merging. Between
collisions the tagged particle moves freely with a speed $\tilde{U}$ along
the direction of the vector $e_1=\left(  1,0,0\right) .$ %$\vec{e}_1=\left(  1,0,0\right)  .$ 
This is similar to the
assumption made in the analysis of coagulation in shear flows (cf.
\cite{Sm}).  We will assume that the average number of particles for
unit of volume is $1$ and that the volume fraction filled by the background
particles is order $\phi>0.$ Concerning the tagged particle, we will assume
that its volume is also of order $\phi.$ This does not suppose a large loss of
generality because smaller particles would become of order $\phi$ after the
first collision. If the tagged particle has a volume much larger than $\phi$
we should describe its evolution by means of a limit case of the equation
derived in this paper. It is then natural to introduce the rescaled variables
$V=\frac{\tilde{V}}{\phi}$ and $v_{j}=\frac{\tilde{v}_{j}}{\phi}.$ Notice that
in the collisions the position of the tagged particle is shifted. Assume that
the tagged particle is initially at the origin of coordinates $x=0.$ We will
denote its position at a given time $t$ as $x=X\left(  t\right)  =\tilde
{U}te_1+\phi^{\frac{1}{3}}Y,$ where $Y=Y\left(  t\right)  .$ The scaling factor
$\phi^{\frac{1}{3}}$ is of the order of the characteristic radius of the particles. From now on we will denote this system as the \textit{coalescing tagged particle model} or by shortness CTP model.
The relation between the Smoluchowski coagulation dynamics in shear flows described above and the reduced model for a single tagged particle we are considering here, is the same relation as the one between a system of interacting particles used to derive the non linear Boltzmann equation and the random Lorentz Gas (see \cite{Lo}).   The Lorentz Gas is a Newtonian dynamical system in which a
tagged particle  interacts with a fixed background of particles through a suitable pair interaction potential. Note that the only randomness in this system is in the distribution of the scatterers.   

The main result that we will obtain in this paper is the following. Suppose
that the probability of finding the tagged particle at the position $\left(
Y,Y+dY\right)  $ and with volume $\left(  V,V+dV\right)  $  at time $t$ is
$f_{\phi}\left(  Y,V,t\right)  dYdV.$ We will assume also that the speed of
the tagged particle between collisions scales as $\tilde{U}=U\phi^{-\frac
{2}{3}}$ where $U$ is of order one and that the volumes $v_j$ of the particles in the background
are independently distributed according to the probability distribution $G\left(  v\right)
.$ Then, if we assume that%
\begin{equation}
\int_{0}^{\infty}G\left(  v\right)  v^{\gamma}dv<\infty\ \ \text{with\ }%
\gamma>2\label{Ggamma}%
\end{equation}
we obtain that with probability one $f_{\phi}\rightarrow f$ where $f$ solves%
\begin{equation}%
\begin{split}
\partial_{t}f(Y,V,t) &  =\left(\frac{3}{4\pi}\right)
^{\frac{2}{3}}U\int_{0}^{\frac{\pi}{2}}d\theta\int_{0}%
^{2\pi}d\varphi \,\sin\theta\cos\theta\,\\
& \quad \Big[ \int_{0}^{V}dv \, G(v) ((V-v)^{\frac{1}{3}}+v^{\frac{1}{3}})^{2}f(Y-\frac{v}{V-v}R{n}%
(\theta,\varphi),V-v,t)\\&\quad -\int_{0}^{\infty} dv \, G(v)(V^{\frac{1}{3}}+v^{\frac{1}{3}})^{2}f(Y,V,t)\Big]\equiv \mathcal{Q}[f](Y,V,t)
\end{split}
\label{eq:FKScaled}%
\end{equation}
where $R=\left(  \frac{3V}{4\pi}\right)  ^{\frac{1}{3}}$ and %
$
n\left(  \theta,\varphi\right)  =\left(  \cos\theta,\sin\theta\cos\varphi
,\sin\theta\sin\varphi\right).
$

The scaling $\tilde{U}=U\phi^{-\frac
{2}{3}}$ ensures that a tagged
particle with volume of order $\phi$ has an average number of collisions for
unit of time of order one. It is interesting to compare this scaling with the scaling for the shear velocity discussed for the Smoluchowski coagulation model. Given that the particle sizes are of order $\phi^{\frac 1 3 }$ it follows that the relative velocity for colliding particles, placed at different heights, is of order $\phi^{-\frac {2 }{3}}$ that is of the same order of magnitude we are considering for this reduced picture. 

 We remark that in
spite of the fact that the problem is linear, it captures some relevant
information taking place in the nonlinear situation as for instance the scaling of
the coagulation kernel with the size of the particle. Indeed, largest particles are
more effective capturing surrounding particles. The difference between the linear and non linear model is that they provide a different  scaling for the particle sizes. This is because the
growth is due to attachment with particles with a given fixed size in the case
considered in this paper, and with particles having a comparable size in the
nonlinear coagulation case. 

Condition (\ref{Ggamma}) guarantees that the number of big particles in
the background is not too large. The point where this assumption is used
suggest that the condition $\gamma>2$ might probably be relaxed to $\gamma>1,$
but to obtain the result under this assumption seemingly needs to use
cumbersome technical arguments which we have preferred to avoid. 

It is relevant to mention that the convergence or divergence of the integrals
with the form $\int_{0}^{\infty}G\left(  v\right)  v^{\gamma}dv$ plays a
crucial role in the theory of Continuum Percolation (cf. \cite{G},
\cite{MR}). For instance, it was proved in \cite{H} (cf. also
\cite{MR}) that the assumption $\int_{0}^{\infty}G\left(  v\right)
v^{\frac{5}{3}}dv$ implies the existence of a critical volume fraction
$\phi_{\ast}>0$ such that for $0<\phi<\phi_{\ast}$ all the clusters of
particles distributed according to the Poisson distribution described above
and with distribution of sizes $G\left(  v\right)  $ are finite with
probability one. 

Actually, the results obtained in this paper can be thought as some kind of
dynamic percolation theory. Indeed, one of the main difficulties that we need
to address proving the results of this paper is the fact that coalescing
particles could trigger sequences of coagulation events, because after one
merging, the tagged particle becomes bigger and its position is shifted and it
could coalesce with additional particles, eventually creating an infinite
cluster. A similar difficulty that arises studying the dynamics of coalescing
particles is that the free flights between coagulation events become shorter
due to the increase of the volume of the tagged particle and this could yield
eventually to a runaway growth of the tagged particle in finite time. We will
prove in this paper that such finite time blow-up of particle sizes does not
take place with probability one under suitable assumptions on the distribution
$G\left(  v\right)  .$  Actually the main source of technical difficulties in this paper is precisely this change of size of the tagged particle in time. 
A rather important feature of the CTP model is that after coalescence the new center of the tagged particle changes to the
center of mass of the particles involved in the merging. This implies that the displacement of the position of the center of the tagged particle is not too large as the size of the tagged particle increases. The method used in this paper works also for coalescing models different from CTP for which the distance of the new centre of the particle from the previous one is not very large (cf. Remark \ref{rk:postagpc}) if the size of the tagged particle is large. Analogously, under this assumption, it seems possible to provide a rigorous derivation of a kinetic equation of the form \eqref{eq:FKScaled} where the collision operator must be modified according to the dynamics we are considering. On the contrary if the jumps of the centers are comparable with the size of the tagged particle, blow up in finite or zero time, with probability one, can  be expected because the tagged particle in its new position will meet obstacles from the background with probability close to one.

\bigskip

To the best of our knowledge the rigorous derivation of Smoluchowski equations
taking as starting point mechanical particle systems has been only considered %mechanistic 
in the papers \cite{HR1, HR2}, \cite{LN}, \cite{N}, \cite{YRH}.  
In all these papers they consider the case
of Brownian coagulation (i.e. the coalescing particles are moving according to Brownian motion with a suitable diffusion coefficient), and the size of the particles does 
not increase in
the coagulation events, although in the problem considered in \cite{HR1, HR2} the
particles might have different masses which aggregate in the coagulation
events and influence the diffusion coefficients of the particles.   As we indicated above, the increasing of the sizes of the particles is the source of many of the 
technical difficulties that we have to address in this paper due to the fact that we have to rule out possible finite time blow ups in the particle sizes.

\bigskip
On the contrary, there are several results in which the Smoluchowski equation has been derived starting from a stochastic process for a system of many particles in the 
same spirit of Kac-models for the derivation of Boltzmann equation. See \cite{DF, EP, FG}. 
The main difference with the approach followed in this paper is that the stochasticity is not only in the initial distribution of particles. Here the evolution of the system is 
purely deterministic while, on the contrary, in the case of the derivation starting from a stochastic process the evolution is probabilistic. 
This stochastic process for coagulating particles is usually referred to as Markov-Lushnikov process, see \cite{Lu1, Lu2, M}.

\bigskip

\bigskip

We remark that there are many analogies between the
derivation of the Smoluchowski coagulation equation for Brownian particles in \cite{LN}  and the derivation of
the classical Boltzmann equation for hard spheres (cf. \cite{La}) and short
range potentials (cf.  \cite{GS-RT, PSS}). In both cases we consider 
the derivation of a kinetic equation starting from a mechanical evolution for a many particle system. Note, however
that the dynamics of the particle system is time-reversible in the case of
hard spheres evolving by means of Hamilton's equations of motion and is an irreversible
evolution in the case of coalescing particles described above. In both cases, in a low density regime, a kinetic equation for the
one-particle distribution function (first marginal) is derived.
As indicated above, in the description of the particle system, the linear Smoluchowski coagulation
equation (\ref{eq:FKScaled})  has the same relation with the Smoluchowski coagulation
equation in a linear shear flow \eqref{SmolKin}
as the linear Boltzmann equation, which gives the mesoscopic description of the Lorenz model in the kinetic regime, with the classical nonlinear
Boltzmann equation.  

The derivation of suitable kinetic equations and diffusive equations starting from the Lorentz model has been extensively studied.
We refer to \cite{BNP, BNPP, BBS, DP, DR, G, No, Sp}. We stress that in the derivation of the kinetic equation \eqref{eq:FKScaled} we resort to the key idea introduced by Gallavotti, in \cite{G}, in order to study the kinetic limit of the Lorentz Gas. However, some modifications are required mostly to avoid the explosive growth of the particle sizes. 

\bigskip

The plan of the paper is the following. In Section \ref{sec:PM} we define precisely the particle model under consideration and we describe the main results obtained in this paper. In Section 3 we provide the rigorous derivation of the equation \eqref{eq:FKScaled} in the limit when the volume fraction $\phi$ tends to zero. In order to do this we define the set of good configurations for the particle system which are those for which the only relevant collisions are the binary collisions. Indeed, we prove that the probability of this set of good configuration tends to one when  $\phi$ tends to zero. %(section 3.2 and 3.3)
In Section 4 we analyze the limit kinetic equation \eqref{eq:FKScaled}. In particular, we prove well-posedness and look at the long time asymptotics of the solution  when $G\left(
v\right) $ decreases fast enough. This allows to show that the volume of the tagged particle increases like $t^3$ as $t$ goes to infinity.
In Section 5, which is the most technical section of the paper, we prove that the CTP model  is well posed for a small but positive volume fraction with probability one as long as the distribution of the particle sizes $G(v)$ is compactly supported. Note that the kinetic equation \eqref{eq:FKScaled} has been obtained for arbitrary large times in the limit as $\phi$ tends to zero. This does not rule out the possibility of blow up in finite time for small but positive volume fraction. Actually is not difficult to construct examples of obstacles configurations exhibiting blow up in finite time, cf. Example \ref{ex:blowup}. At the end of the paper we collected in the Appendixes some technical estimates which in particular control the displacement of the center of mass for the coalescing particle as well as a probability lemma using the proof of global well-posedness for positive volume fraction $\phi$ in the particular case in which all the obstacles are identical (i.e. $G(v)=\delta(v-1)$). In this case many of the technicalities used in Section 5 can be avoided and the main 
ideas can be more easily grasped. 
\bigskip

\section{The particle model and main results}\label{sec:PM} 

\subsection{The CTP model%The particle model 
}\label{ssec:PM} 
We consider the configuration space $\Omega$ which is defined as follows 
\begin{equation}
\Omega:=\{\omega=\{x_k,v_k\}_{ k\in I}, \; I \subset \N \;\; \text{s.t.} \,\{x_k\}_{ k\in I} \;\, \text{is locally finite and } \;v_k>0, \; x_k\in\R^3  \}
\end{equation}
and we identify two elements in $\Omega$ if one is obtained from the other by means of a permutation of the indexes. 
We can endow the space $\Omega$ with a structure of measure space in a suitable sigma algebra $\Sigma$ and define a probability $\mu_\phi$ such that
the centers $\{x_k\}$ and the volumes $\{v_k\}$ are independently distributed variables, with the positions given by the Poisson distribution in $\R^3$ of intensity one and the volumes distributed according to $\frac{1}{\phi}G(\frac{v}{\phi})$. Here $G$ is a probability distribution in $[0,\infty)$ and $\phi>0$ is the volume fraction occupied by the particles.  Note  that we drop the dependence of the measure $\mu_\phi$ on $G$. This probability measure is the one associated to the Boolean model in the theory of continuum percolation, the detailed construction can be found in \cite{MR}, Section 1.4. 

 In what follows we construct the stochastic process for the tagged particle with unscaled variables $\tilde{Y},\tilde{V}$. From now on we will use capital letters to denote the dynamical variables for the tagged particle while small letters will be used for the obstacles. Our purpose is to show that the particle system obtained is well defined. We consider the  initial configuration for the tagged particle  $(\tilde{Y_0},\tilde{V_0})$ and we are interested in defining the dynamical system for the volume and the displacement of the position of the tagged particle with respect to the constant speed motion, namely $\tilde{Y}\left(  t\right)  =\tilde{X}(t)-\tilde
{U}t e_1$. Note that this allows to consider a Galilean reference frame in which the background of obstacles moves against the original position of the tagged particle. Thus we define the evolution flow  
$\tilde{T}^t(\tilde{Y_0},\tilde{V_0};\tilde \omega_0)=(\tilde{Y}(t;\tilde{Y_0},\tilde{V_0},\tilde \omega_0),\tilde{V}(t;\tilde{Y_0},\tilde{V_0},\tilde \omega_0);\tilde\omega(t;\tilde{Y_0},\tilde{V_0},\tilde \omega_0) )$  $t\geq 0$ where $\tilde\omega_0$ and $\tilde\omega(t)$ denote the initial configuration and the evolved configuration of obstacles respectively.  Note that from now on we will skip the $\tilde{Y_0},\tilde{V_0},\tilde \omega_0$ dependence and consider $\tilde{T}^t(\tilde{Y_0},\tilde{V_0};\tilde\omega_0)=(\tilde{Y}(t),\tilde{V}(t);\tilde\omega(t))$ for notational simplicity. Moreover,  
we rescale variables according to the following scaling limit of Boltzmann-Grad type 
\begin{align}\label{scaling}
& \tilde{R}= \phi^{\frac{1}{3}} R,\,\tilde{r}= \phi^{\frac{1}{3}}r,\, \tilde{Y}= \phi^{\frac{1}{3}}Y, \\\nonumber
&  \tilde{V}= \phi V,\,  \tilde{v}= \phi v,\\\nonumber
&  \tilde{U}=\phi^{-\frac{2}{3}} U. 
\end{align}
Note that the rescaling for the velocity of the tagged particle  
gives a mean free flight time of order one, while the mean free path is order $\phi^{-\frac{2}{3}}$.  

The rescaled evolution is given by
\begin{equation}
T_\phi^t(Y_0,V_0;\omega_0)=\Big(\frac{1}{\phi^{\frac 1 3}}\tilde{Y}^{t}(\phi^{\frac 1 3}\tilde{Y_0},\phi \tilde{V_0};\tilde\omega_0),\frac{1}{\phi}\tilde{V}^{t}(\phi^{\frac 1 3} \tilde{Y_0},\phi\tilde{V_0};\tilde\omega_0); \omega(t)\Big)
=(Y(t),V(t);\omega(t))
\end{equation}
where $\omega $ is the configuration of obstacles obtained by $\tilde \omega$ rescaling the volumes. Moreover, we define the projector operator $\Pi [T_\phi^t(Y_0,V_0;\omega_0)]=(Y(t),V(t))$.

Let be $t^{*}\geq 0$, we introduce the set $\mathcal{F}(Y,V,t^{*})\in\Sigma$ which consists of obstacles configurations moving according to the free flow during the time interval $[0,t^{*}]$, i.e.
 \begin{equation}\label{eq:dom_free_mot}
\mathcal{F}(Y,V,t^{*}):=\{\omega\in \Omega\,:\, \inf_{k:\{(x_k,r_k)\}=\omega}|Y-(x_k-Ut^{*}e_1)|>\sigma(V^{\frac 1 3}+v_k^{\frac 1 3})\}
\end{equation}
where $\sigma=\left(\frac{3}{4\pi}\right)^{\frac 1 3}$ and $\mathcal{C}(Y,V,t^{*})\in \Sigma$ is the set of obstacles configurations suffering collisions at time $t^{*}$, i.e.
 \begin{equation}\label{eq:dom_coll}
 \mathcal{C}(Y,V,t^{*})=\{\omega\in \Omega\,:\, \inf_{k:\{(x_k,r_k)\}=\omega}|Y-(x_k-Ut^{*}e_1)|\leq \sigma(V^{\frac 1 3}+v_k^{\frac 1 3})\}.
\end{equation}
If $\omega\in \mathcal{C}(Y,V,t^{*})$ we define the merging operator 
\begin{equation}\label{def:merging_op}
\mathcal{A}(Y,V;\omega)=\left(\frac{VY+\sum_{k\in J}x_k v_k}{V}+\sum_{k\in J}v_k,V+\sum_{k\in J} v_k ;\omega\setminus J\right) 
\end{equation} 
where $J=J(Y,V,t^{*};\omega)$ is the set of indexes of obstacles in $\Omega$ satisfying $|Y-(x_k-Ut^{*}e_1)|< \sigma(V^{\frac 1 3}+v_k^{\frac 1 3})$. Note that \eqref{def:merging_op} is well defined since the configuration is locally finite. We will denote as $\omega\setminus J$ the set of obstacles in $\Omega$ removing those with indexes in $J$. 

If $\omega\notin \mathcal{C}(Y,V,t^{*})$ there exists a $\delta=\delta(\omega)>0$ such that $\omega \in\mathcal{F}(Y,V,t^{*}+\delta)$.

  We can now define the flow in the following way. 
\begin{itemize}
\item[1)] If $\omega(\bar t) \in \mathcal{F}(Y(\bar{t}),V(\bar{t}),t^{*})$ for some $t^{*}>\bar{t}$ we define
\begin{equation}\label{def:freeflow}
(Y(t),V(t);\omega(t))=(Y(\bar{t}),V(\bar{t});\omega(\bar{t})) \quad \forall t\in [\bar{t},t^{*}].
\end{equation}
\item[2)] We define a sequence inductively as follows. We set $(Y^0,V^0; \omega^0)=(Y(t^{-}),V(t^{-});\omega(t^{-}))$, %$J^0=J(Y^0,V^0;\omega^0)$ 
and define $(Y^1,V^1,\omega^1)=\mathcal{A}(Y^0,V^0;\omega^0)$. 
Then we have two possible situations: 
\begin{itemize}
\item[i)]  if $\omega^1 \in \mathcal{F}(Y_1,V_1,t)$ we are in case 1) and we have free flow.
\item[ii)] if $\omega^1 \in \mathcal{C}(Y_1,V_1,t)$ then we can define inductively $(Y^n,V^n,\omega^n)=\mathcal{A}(Y^{n-1},V^{n-1};\omega^{n-1})$ until reach a value of $n$
 such that $\omega^1 \in \mathcal{F}(Y^{n},V^{n};t)$. Then we return to step 1). 
\end{itemize}
\end{itemize}
Note that if the case 
2.ii) holds for arbitrary large values of $n$ the dynamics stops at time $t$ with an infinite sequence of coalescences. Therefore, the previous dynamics is not necessarly defined for arbitrary long times. We will provide later an example of configuration $\omega\in \Omega$ for which this happens (see Example \ref{ex:blowup}). %, otherwise we will have the regime of free flow described by case 1).
Nevertheless, we will prove that the coalescence dynamics stops after a finite number of steps which are then followed by a free motion of the tagged particle with probability one and the dynamics can be defined globally in time if $G(v)$ is compactly supported and $\phi$ positive but sufficiently small (cf. Proposition \ref{prop:wellpos1}). We observe that the fact that global well posedness can be proved only with probability one is typical in systems with infinitely many particles (see \cite{La}). On the other hand the dynamics above can be defined with probability close to one if $\phi$ is sufficiently small and if $G(v)$ satisfies \eqref{Ggamma}. % This will be used in order to prove that there is no blow up in finite time at least for suitable choices of the probability $G$.
We note that the evolution operator $T^t_{\phi}$ described by points 1) and 2) will be called from now on the CTP model, as we already anticipated in the Introduction. Furthermore, we remark that there is a family of CTP models depending on the choice of $(Y_0,V_0)\in \R^3\times [0,\infty)$ and $\phi>0$ but from now on we will skip this dependence. 

\subsection{Kinematic of binary collisions}\label{ssec:bincoll}
In the kinetic limit we are considering the only relevant coalescing processes are those involving the tagged particle with just a single particle in the background. Here we describe in detail the geometry of this type of coalescence events. 

We denote by $V=\frac{4}{3}\pi R^3$  the volume of the tagged particle and by $v=\frac{4}{3}\pi r^3 $ the volume of one generic obstacle. We are interested in describing how the volume of the tagged particle changes after one collision. Indeed, once the tagged particle merges with one single obstacle, we have that $(V,v)\to V+v$.  We denote  by $R'=(r^3+R^3)^{\frac 1 3}$ the radius of the tagged particle after the coalescence.
To have an exhaustive description of this dynamics we need informations on the position of the center of mass of the new particle after the coalescence. We introduce the variable $\theta\in[0,\frac{\pi}{2}]$ which is the angle between the vector $e_1=(1,0,0)$ and the vector connecting the center of the tagged particle with the center of the coalescing obstacle. %such that $\cos\theta=\frac{l_h}{R+r}=\frac{l_h}{d}$, $\sin\theta=\frac{l_v}{R+r}=\sqrt{1-\frac{l_h}{d}^2}$. %(to parametrize the impact parameter) 
Moreover we introduce the vector $\nu=\nu(\varphi)=(0,\cos \varphi,\sin\varphi)$ with $\varphi\in [0,2\pi]$ in such a way that the position of the centre of the obstacle is given by
$$x=X+\left(\cos\theta\,e_1+\sin\theta\,\nu\right)(R+r).$$ 
Then the new position and volume of the tagged particle defined by the operator $\mathcal{A}$ in \eqref{def:merging_op} are given by
\begin{align}\label{def:merg_dyn}
&X'=X+\frac{v}{V+v}(\cos\theta e_1+\sin\theta \,\nu)(R+r)\\\nonumber
&V'=V+v.
\end{align}
This operator plays the role of the scattering map for the Newtonian dynamics.

We observe that from now on we describe the motion of the tagged particle in terms of $(Y,V)$ where $Y:=X-Ute_1$, namely we are considering a coordinate system in which the original position of the tagged particle moves with constant speed $Ue_1$. As a consequence, the new position of the tagged particle according to \eqref{def:merg_dyn} becomes $Y'=Y+\frac{v}{V+v}(\cos\theta e_1+\sin\theta \,\nu)(R+r)$.

%%%%%%%%

\subsection{Main results}\label{sec:3}
We denote by $\mathcal{P}(\R^3\times [0,+\infty))$ the set of probability measures on $\R^3\times [0,+\infty)$ and by $\mathcal{M}_{+}(\R^3\times [0,+\infty))$ the set of Radon measure on $\R^3\times [0,+\infty)$ and define the solution of the microscopic coalescence process as follows. 
\begin{definition}\label{def:f_phi}
Let be $f_0\in \mathcal{P}(\R^3\times [0,+\infty))$, for any Borel set $A$ of $\R^3\times [0,+\infty)$ we define $f_\phi \in L^{\infty}( [0,T); \mathcal{M}_{+}(\R^3\times [0,+\infty)))$ as
\begin{equation}\label{defeq:f_phi}
\int_{A} f_\phi(Y,V,t) dY\,dV =\int_{\R^{3}\times [0,\infty]} \mu_\phi(\{\omega\,:\; \Pi[T_\phi^{t}(Y_0,V_0;\omega)]\in A\})f_0(Y_0,V_0) dY_0\,dV_0. 
\end{equation}
\end{definition}
Our goal is to prove that $ f_{\phi}(Y,V,t)\to f(Y,V,t)$ as $\phi\to 0$ in a suitable sense, where $f$ is a solution of the kinetic equation \eqref{eq:FKScaled} with initial datum $f_0$. We now introduce the meaning of solutions for the equation \eqref{eq:FKScaled} in the sense of measures. 
\begin{definition}\label{def:f_soleq}
Let be $T>0$ and $f\in C([0,T]; \mathcal{M}_{+}(\R^3\times [0,\infty))$. We say that $f$ is a weak solution of \eqref{eq:FKScaled} if $f_0\in \mathcal{P}(\R^3\times [0,\infty))$ and for any $\Psi \in C_{c}^1([0,T)\times \R^3 \times [0,\infty))$ we have
\begin{equation}\label{eq:f_soleq}
\begin{split}
&-\int_0^T \int_{\R^3} \int_{ [0,\infty)} f(Y,V,t)\{\partial_t\Psi(Y,V,t)+ \mathcal{C}[\Psi](Y,V,t)\} dYdVdt \\&
=\int_{\R^3} \int_{ [0,\infty)} f_0(Y,V) \Psi_0(Y,V)dYdV
\end{split}
\end{equation}
where the collision operator $\mathcal{C}$ is defined by 
\end{definition}

\begin{equation} %\label{op:CScaled}
\begin{split}
(\mathcal{C}[\Psi])(Y,V,t)=&\,U\!\int_0^{\infty}\hspace{-2.5mm}dv\int_0^{\frac{\pi}{2}}\hspace{-1.5mm}d\theta\int_0^{2\pi}\hspace{-2mm}d\varphi\, G(v) (R+r)^2 \sin\theta\cos\theta
 \\
 &\,\times \left[\Psi(Y+\frac{v}{V+v}(\cos\theta e_1+\sin\theta \,\nu)(R+r),V+v,t)-\Psi(Y,V,t)\right]. \\
\end{split}
\end{equation}

\begin{remark}
We observe that to define solutions of equation \eqref{eq:FKScaled} it is enough to assume that 
$$\int_{0}^{\infty}dv\, G(v)(1+v^{\frac 2 3}) < +\infty.$$
Under this conditions the operator $\mathcal{C}$ transforms continuous functions in continuous functions.
\end{remark}

We now state the first result of this paper.
\begin{theorem}\label{th:theorem1}
Let $G(v)\in \mathcal{M}_{+}([0,\infty))$ be such that 
\begin{equation}\label{ass:G}
\int_{0}^{\infty} v^{\gamma}G(v)\,dv<\infty, \qquad \gamma>2.%\frac{43}{3}.
\end{equation}
Let be $f_0 \in \mathcal{P}(\R^3\times [0,+\infty))$, $f_\phi \in L^{\infty}( [0,T); \mathcal{M}_{+}(\R^3\times [0,+\infty)))$ defined by \eqref{defeq:f_phi} and $T>0$. Then, for any Borel set $A\in \R^3\times [0,\infty)$, 
$$\int_{A}f_\phi(t) \to \int_{A }f (t)\quad \text{as}\quad \phi\to 0,$$  %in the sense of distribution
where $f$ is the unique solution of equation \eqref{eq:FKScaled} in the sense of Definition \ref{def:f_soleq}. The convergence is uniform in $[0,T]$. 
\end{theorem}	
													
\begin{remark}
The fact that there exists a unique solution of equation \eqref{eq:FKScaled} in the  sense of Definition \ref{def:f_soleq} will be proved in Section  \ref{ssec:WPKE}.			 									
\end{remark}	
															
The second main result of the paper concerns the global well posedness for the CTP model introduced in Section \ref{ssec:PM}.
\begin{theorem}\label{th:theorem2}
Let $G(v)\in \mathcal{M}_{+}([0,\infty))$ be compactly supported with $\supp G(v)\in [0,v_{*}]$. Then there exists a $\phi_{*}=\phi_{*}(v_{*})>0$ such that for any $\phi\leq \phi_{*}$ and any $(Y_0,V_0)\in \R^3\times [0,\infty)$ there exists $\tilde\Omega\subset \Omega, \; \tilde\Omega\in \Sigma$ where $\Sigma$ is the $\sigma-$algebra defined in Section \ref{ssec:PM}, such that $\PP(\tilde\Omega)=1$ and the dynamics of the CTP model is well defined for any $\omega\in \tilde\Omega$, for arbitrary long times.
\begin{equation}
\end{equation}
\end{theorem}

																%%%%%%%%%%

\section{Proof of the kinetic limit (Theorem \ref{th:theorem1})}\label{sec:proofTh1}

%\textcolor{red}{the plan of the proof}
\subsection{Definition of the set of good configurations ${\Omega}^1$}\label{ssec:goodconf}

%To show the global well posedness of the particle system w
We consider a decomposition of the configuration space  $\Omega$ defined in Section \ref{ssec:PM} as follows
$$\Omega=\Omega^1 \cup \Omega^2$$
where $\Omega^1$ denotes the set of good configurations, while $\Omega^2$ is the set of bad configurations. % (for instance the set of dense configurations of obstacles which 
More precisely $\Omega^1$ denotes the set of well separated configurations of obstacles so that we can deal with a dilute regime which is necessary to obtain a kinetic picture. 
We recall that in all this section $G(v)$ satisfies Assumption \ref{ass:G}.
 
\bigskip

Let $(Y_{t}(\omega),V_t(\omega))$ the evolved configuration of the tagged particle at time $t$. 
Given $T>0$, $\ep_0>0$ and $\phi>0$ we need to define the following quantities. We   
define 
\begin{equation}\label{ass:V_*2}
\lambda=\ep_0^{-\delta_1}\qquad \text{and} \quad 
V_{*}=\ep_0^{-\frac{2}{3(\gamma-1)}-\frac{\delta_2}{\gamma-1}}
\end{equation}
with $\delta_2>\delta_1> 0$. 
We will assume that $\ep_0\leq 1$ so that $\lambda \geq 1$ and $V_{*}\geq 1$. Note that this implies 
\begin{equation}\label{ass:V_*1}
%\frac{V_{*}^{\frac{4}{3}-\frac{1}{d}}\ep_0^{\frac{2}{3}}}{\lambda} \to \infty \quad \text{as} \quad \ep_0\to 0,
\frac{V_{*}^{\gamma-1}\ep_0^{\frac{2}{3}}}{\lambda}=\ep_0^{\delta_1-\delta_2} \to \infty \quad \text{as} \quad \ep_0\to 0. %\text{so that}\; V_{*}=\ep_0^{-\frac{2}{3(\gamma-1)}}.
\end{equation}
We introduce the maximal radius $R_{max}$ which is defined by 
 \begin{equation}\label{def:R_max}
 R_{max}:=\lambda V_{*}^{\frac{1}{3}}
 \end{equation}
and we consider the domain of length $T\phi^{-\frac 2 3}$ and size $R_{max}$ defined by
\begin{equation}\label{def:collcyl}
\mathcal{D}_{*}=({B}^{2}_{\phi^{\frac{1}{3}}R_{max}}(0)\times [0,T\,\phi^{-\frac 2 3}))\cup (S_{C_t}\cup S_{C_b}),
\end{equation}
where $S_{C_t}$ and $ S_{C_b}$ denote the spherical caps at the top and the bottom of the cylinder, ${B}^{2}_{r}(0)$ is the 2-dimensional ball of centre $0$ and radius $r$ and $\mathcal{D}_{*}=\mathcal{D}_{*}(\ep_0,\phi,T)$ but from now on we will drop this dependence. See Figure \ref{fig:tube}.

\begin{figure}[ht]
\centering
\includegraphics[scale= 0.17]{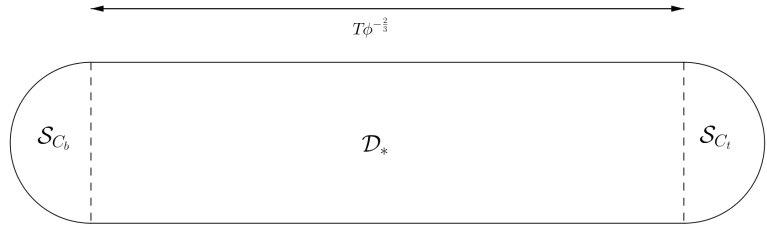}
\caption{The maximal collision tube $\mathcal{D}_{*}$.}
\label{fig:tube}
\end{figure}

%%%%%%%%%%%%%%%%%%%%%%%%%%%%%%%%%%%%%%%%%%%%%%%%%%%%%%%%%%%%%%%%%%%%%%%%%%
%%%%%%%%%%%%%%%%%%%%%%%%%%%%%%%%%%%%%%%%%%%%%%%%%%%%%%%%%%%%%%%%%%%%%%%%%%
%%%%%%%%%%%%%%%%%%%%%%%%%%%%%%%%%%%%%%%%%%%%%%%%%%%%%%%%%%%%%%%%%%%%%%%%%%

We first introduce an auxiliary set which will be used in the construction of the set of the good configurations $\Omega^1$, namely
\begin{equation}\label{def:goodconf_G}
\begin{split}
\mathcal{G}:=&\big\{ \omega\in \Omega \; : \; n(\mathcal{D}_{*})\leq \lambda \pi R^{2}_{\max}(T+1),\; n(S_{C_b}\cup S_{C_t})=0\quad \text{and}\quad \forall x_{j}(\omega)\; \;  
\\& 
\quad \text{s.t.}\;\; B_{\phi^{\frac{1}{3}}R_j(\omega)}(x_j(\omega))\cap \mathcal{D}_{*} \neq\emptyset \quad \text{we have}\; R_j(\omega)\leq \frac{3}{4\pi}(V_{*})^{\frac{1}{3}} \big\} 
\end{split}
\end{equation}
where %$n(\mathcal{D}_{j})=\# \{x_{k}(\omega)\; \; \text{s.t.}\;\; B_{R_k(\omega)}(x_k(\omega))\cap \mathcal{D}_{j} \neq\emptyset\}$ and 
$n(\mathcal{D}_{*})=\# \{x_{k}(\omega)\; \; \text{s.t.}\;\; B_{\phi^{\frac{1}{3}} R_k(\omega)}(x_k(\omega))\cap \mathcal{D}_{*} \neq\emptyset\}$.
Observe that $\mathcal{G}=\mathcal{G}(\ep_0,\phi,T,0,V_0)$. %${\Omega}^1={\Omega}^1(\ep_0,\phi,T).$
Notice that we can define the auxiliary set $\mathcal{G}$ for a different initial center of the tagged particle by means of a suitable translation of the set $\mathcal{D}_{*}$.

In order to define $\Omega^1$ we construct inductively a sequence of domains $\{ \mathcal{U}_{k} \}$  with $\mathcal{U}_{k} \subset \mathcal{G}$ in the following way. 
We start from $\mathcal{U}_{0}$. We set $\rho_0:=R_0+\frac{3}{4\pi}V^{\frac{1}{3}}_{*}$ which is the unrescaled enlarged radius at the first step and 
\begin{equation}\label{eq:defD0}
\mathcal{D}_0:=B^2_{\rho_0\phi^{\frac{1}{3}}}(0)\times [0,\ep_0 \phi^{-\frac{2}{3}}),
\end{equation}
the cylinder of size $\rho_0$, so that we define
\begin{equation}\label{def:U0}
\mathcal{U}_{0}:=\{\omega\in\mathcal{G}: \; \#(\{x_k(\omega)\}\cap \mathcal{D}_0)\leq 1 \;\;\text{and if }\;\; v_{k}(\omega)\in \mathcal{D}_0,\; v_{k}(\omega)\leq 1\}
\end{equation}
then the complementary set of $\mathcal{U}_{0}$ in $\mathcal{G}$ is given by
\begin{equation}\label{def:U0compl}
\mathcal{\bar{U}}_{0}:=\{\omega\in\mathcal{G}: \; \#(\{x_k(\omega)\}\cap \mathcal{D}_0)> 1 \;\;\text{or}\;\; [\#(\{x_k(\omega)\}\cap \mathcal{D}_0)=1 \;\;\text{and}\;\; v_{k}(\omega)\in \mathcal{D}_0,\; v_{k}(\omega)> 1]  \}.
\end{equation} 

To construct $\mathcal{U}_{1}$ we note that, given $\omega\in \mathcal{U}_{0}$, we have two possibilities:   
\begin{itemize}
\item[(1)] $ \#(\{x_k(\omega)\}\cap \mathcal{D}_0)=0$,
\item[(2)] $ \#(\{x_k(\omega)\}\cap \mathcal{D}_0)=1$.
\end{itemize}

If (1) holds it means that the tagged particle suffered no collisions in $\mathcal{D}_0$. It is then natural to define $\rho_1=\rho_0$, $X_1=X_0$, $V_1=V_0$ 
and 
\begin{equation}\label{eq:defD1}
\mathcal{D}_1=\mathcal{D}_1(\omega_{|_{\mathcal{D}_0}}):=B^2_{\rho_1\phi^{\frac{1}{3}}}(X_1)\times [\ep_0 \phi^{-\frac{2}{3}},2 \ep_0 \phi^{-\frac{2}{3}}),
\end{equation}
where we used the notation $\mathcal{D}_1(\omega_{|_{\mathcal{D}_0}})$ to indicate the dependence on the particles of the configuration $\omega$ restricted to the domain $\mathcal{D}_0$.

In the case (2) instead, we suppose that $(x_1,v_1)\in \omega$ is the obstacle belonging to the intersection, i.e. $x_1\in \{x_k(\omega)\}\cap \mathcal{D}_0$. Since a collision occurred it is then natural to define 
$\rho_1=\frac{3}{4\pi}(V_1^{\frac{1}{3}}+V_{*}^{\frac{1}{3}})$, $X_1=\frac{v_1}{V_0+v_1}x_1$, $V_1=V_0+v_1$ and
$\mathcal{D}_1=\mathcal{D}_1(\omega_{|_{\mathcal{D}_0}}):=B^2_{\rho_1\phi^{\frac{1}{3}}}(X_1)\times [\ep_0 \phi^{-\frac{2}{3}},2 \ep_0 \phi^{-\frac{2}{3}})$. 
Therefore we define 
\begin{equation}\label{def:U1}
\begin{split}
\mathcal{U}_{1}:=&\{\omega\in\mathcal{U}_{0}: \; [\#(\{x_j(\omega)\}\cap \mathcal{D}_{1}(\omega_{|_{\mathcal{D}_{0}}}) )=0] \quad \text{or}\\&
\;\; [\#(\{x_j(\omega)\}\cap \mathcal{D}_{1}(\omega_{|_{\mathcal{D}_{0}}}) )=\#\{x_j\}=1 \;\; \text{and}\;\; \#(\{x_j(\omega)\}\cap \mathcal{D}_{0})=0 \\&
\;\;\text{and}\; \; v_j(\omega)\leq (1+\theta(1,\omega))^{\frac{1}{\gamma}(1+\delta_3)} %\#(\{x_j(\omega)\}\cap \mathcal{D}_{k+1}(\omega_{|_{\bigcup_{l=0}^{k}\mathcal{D}_{l}}}))=0
] \}.
\end{split}
\end{equation}
Here we introduced the function 
\begin{equation}\label{def:f-teta}
\theta(k,\omega):=\sum_{j=1}^{k}\{\#\{x_l(\omega)\}\cap \mathcal D_{j}(\omega)\}
\end{equation} 
for the level $k=1$, i.e. $\theta(1,\omega)=1$. Note that, being $\theta(k,\omega)$ a function of the configuration $\omega$, it depends on all the previous history, i.e. the particles of the configuration $\omega$ contained in $\bigcup_{l=0}^{k-1}\mathcal{D}_{l}$.
We remark that the complement of $\mathcal{{U}}_{1}$ in $\mathcal{{U}}_{0}$ is 
\begin{equation}\label{def:U1compl}
\begin{split}
\mathcal{\bar{U}}_{1}:=&\{ \omega\in\mathcal{U}_{0}: \; \#(\{x_j(\omega)\}\cap\mathcal{D}_{1}(\omega_{|_{\mathcal{D}_{0}}}) )> 1\quad \\&
\quad \text{or}\quad [ \#(\{x_j(\omega)\}\cap \mathcal{D}_{1}(\omega_{|_{\mathcal{D}_{0}}}))=1 \; \text{and}\;  \#(\{x_j(\omega)\}\cap \mathcal{D}_{0})=1]
 \\&
\quad \text{or}\quad [ \#(\{x_j(\omega)\}\cap \mathcal{D}_{1}(\omega_{|_{\mathcal{D}_{0}}}))=1 \; \text{and}\;  \#(\{x_j(\omega)\}\cap \mathcal{D}_{0})=0
\\&
\quad \text{and} \,\, v_j(\omega)> (1+\theta(1,\omega))^{\frac{1}{\gamma}(1+\delta_3)}] \}.
\end{split}
\end{equation}
We now define $\mathcal{U}_{k+1}$ iteratively.  Given $\omega\in \mathcal{U}_{k}$ we consider the two different cases:
\begin{itemize}
\item[(1)] $ \#(\{x_j(\omega)\}\cap \mathcal{D}_k(\omega))=0$,
\item[(2)] $ \#(\{x_j(\omega)\}\cap \mathcal{D}_k(\omega))=1$.
\end{itemize}
As we observed before, if (1) holds it means that the tagged particle suffered no collisions in $\mathcal{D}_k(\omega)$. We define $\rho_{k+1}=\rho_{k}$, $X_{k+1}=X_{k}$, $V_{k+1}=V_{k}$ and 
\begin{equation}\label{eq:defDk}
\mathcal{D}_{k+1}=\mathcal{D}_{k+1}(\omega_{|_{\bigcup_{l=0}^{k}\mathcal{D}_{l}}}):=B^2_{\rho_{k+1}\phi^{\frac{1}{3}}}(X_{k+1})\times [k\ep_0 \phi^{-\frac{2}{3}},(k+1) \ep_0 \phi^{-\frac{2}{3}}).
\end{equation}
 If (2) is satisfied
we assume that $(x_{k+1},v_{k+1})$ is the obstacle such that $x_{k+1}\in \{x_j(\omega)\}\cap \mathcal{D}_k(\omega)$. Thus, we set $\rho_{k+1}=\frac{3}{4\pi}(V_{k+1}^{\frac{1}{3}}+V_{*}^{\frac{1}{3}})$, $X_{k+1}=X_{k}+\frac{v_{k+1}}{V_{k}+v_{k+1}}(x_{k+1}-X_{k})$, $V_{k+1}=V_{k}+v_{k+1}$ and
$\mathcal{D}_{k+1}=\mathcal{D}_{k+1}(\omega_{|_{\bigcup_{l=0}^{k}\mathcal{D}_{l}}}):=B^2_{\rho_{k+1}\phi^{\frac{1}{3}}}(X_{k+1})\times [k\ep_0 \phi^{-\frac{2}{3}},(k+1) \ep_0 \phi^{-\frac{2}{3}})$.
Hence, we define 
\begin{equation}\label{def:Uk+1}
\begin{split}
\mathcal{U}_{k+1}:=&\{\omega\in\mathcal{U}_{k}: \; [\#(\{x_j(\omega)\}\cap \mathcal{D}_{k+1}(\omega_{|_{\bigcup_{l=0}^{k}\mathcal{D}_{l}}}) )=0] \quad \text{or}\\&
\;\; [\#(\{x_j(\omega)\}\cap \mathcal{D}_{k+1}(\omega_{|_{\bigcup_{l=0}^{k}\mathcal{D}_{l}}}) )=\#\{x_j\}=1 \;\; \text{and}\;\; \#(\{x_j(\omega)\}\cap \mathcal{D}_{k}(\omega_{|_{\bigcup_{l=0}^{k-1}\mathcal{D}_{l}}}))=0 \\&
\;\;\text{and}\; \; v_j(\omega)\leq (1+\theta(k,\omega))^{\frac{1}{\gamma}(1+\delta_3)} %\#(\{x_j(\omega)\}\cap \mathcal{D}_{k+1}(\omega_{|_{\bigcup_{l=0}^{k}\mathcal{D}_{l}}}))=0
] \},
\end{split}
\end{equation}
where $\delta_3>0$ is sufficiently small. %\textcolor{red}{We will precise later $\delta_3$.}
Then, the complementary set of $\mathcal{U}_{k+1}$ in $\mathcal{U}_{k}$ is
\begin{equation}\label{def:Uk+1compl}
\begin{split}
\mathcal{\bar{U}}_{k+1}:=&\{\omega\in\mathcal{U}_{k}: \; \#(\{x_j(\omega)\}\cap \mathcal{D}_{k+1}(\omega_{|_{\bigcup_{l=0}^{k}\mathcal{D}_{l}}}))> 1\quad \\&
\quad \text{or}\quad [ \#(\{x_j(\omega)\}\cap \mathcal{D}_{k+1}(\omega_{|_{\bigcup_{l=0}^{k}\mathcal{D}_{l}}}))=1 \; \text{and}\;  \#(\{x_j(\omega)\}\cap \mathcal{D}_{k}(\omega_{|_{\bigcup_{l=0}^{k-1}\mathcal{D}_{l}}}))=1]
 \\&
\quad \text{or}\quad [ \#(\{x_j(\omega)\}\cap \mathcal{D}_{k+1}(\omega_{|_{\bigcup_{l=0}^{k}\mathcal{D}_{l}}}))=1 \; \text{and}\;  \#(\{x_j(\omega)\}\cap \mathcal{D}_{k}(\omega_{|_{\bigcup_{l=0}^{k-1}\mathcal{D}_{l}}}))=0
\\&
\quad \text{and} \,\, v_j(\omega)> (1+\theta(k,\omega))^{\frac{1}{\gamma}(1+\delta_3)} ]\}.
\end{split}
\end{equation}
Note that, by construction, the sets $\{\mathcal{\bar{U}}_{k}\}$ are mutually disjoint, namely $\mathcal{\bar{U}}_{j}\cap \mathcal{\bar{U}}_{m}= \emptyset$ for $j\neq m$.
Moreover, we observe that $\mathcal{\bar{U}}_{k+1}$is the union of disjoint sets
\begin{equation}
\begin{split}
&\mathcal{\bar{U}}_{k+1}^{(1)}:=\{\omega\in\mathcal{U}_{k}: \; \#(\{x_j(\omega)\}\cap \mathcal{D}_{k+1}(\omega_{|_{\bigcup_{l=0}^{k}\mathcal{D}_{l}}}))> 1\},\\&
\mathcal{\bar{U}}_{k+1}^{(2)}:=\{\omega\in\mathcal{U}_{k}: \; \#(\{x_j(\omega)\}\cap \mathcal{D}_{k+1}(\omega_{|_{\bigcup_{l=0}^{k}\mathcal{D}_{l}}}))=1 \; \text{and}\;  \#(\{x_j(\omega)\}\cap \mathcal{D}_{k}(\omega_{|_{\bigcup_{l=0}^{k-1}\mathcal{D}_{l}}}))=1\},\\&
\mathcal{\bar{U}}_{k+1}^{(3)}:=\{\omega\in\mathcal{U}_{k}: \; \#(\{x_j(\omega)\}\cap \mathcal{D}_{k+1}(\omega_{|_{\bigcup_{l=0}^{k}\mathcal{D}_{l}}}))=1 \; \text{and}\;  \#(\{x_j(\omega)\}\cap \mathcal{D}_{k}(\omega_{|_{\bigcup_{l=0}^{k-1}\mathcal{D}_{l}}}))=0\\&
 \hspace{1.5cm} \text{and} \,\, v_j(\omega)> (1+\theta(k,\omega))^{\frac{1}{\gamma}(1+\delta_3)} \}.
\end{split}
\end{equation}

We then define the set of good configurations $\Omega^1={\Omega}^1(\ep_0,\phi,T,0,V_0)=\mathcal{U}_{[\frac{T}{\ep_0}]}$. 
We observe that $\mathcal{U}_{[\frac{T}{\ep_0}]}= \bigcap _{k=0}^{[\frac{T}{\ep_0}] }  \mathcal{U}_{k}$  since $ \mathcal{U}_{k+1}\subset  \mathcal{U}_{k}$ $\forall k$. 
We can define $\Omega^1={\Omega}^1(\ep_0,\phi,T,Y_0,V_0)$ for different initial center $Y_0$ of the tagged particle by means of a suitable translation of the sets $\mathcal{D}_{*}$ and $\mathcal{D}_{k}$. The main result we prove in this section is the following
\begin{theorem}\label{th:goodconf}
Given $M>1$, let $T>0$. For any $\delta>0$ there exists $\ep_{*}=\ep_{*}(\delta, T)>0$ such that if $0<\ep_0<\ep_{*}$ and $\phi<\ep_0^{\beta}$ 
with $\beta=\frac{2}{3(\gamma-1)}+4\delta_1+\frac{\delta_2}{(\gamma-1)}$, then
$$\PP(\Omega^1)\geq 1 -\delta,$$
for any $V_0\in [0,M]$.
\end{theorem}

In order to prove Theorem \ref{th:goodconf} we will use some auxiliary results listed below.

\begin{proposition}\label{prop:1}
Let $\mathcal{G}=\mathcal{G}(\ep_0,\phi,T,0,V_0)$ be the set defined by \eqref{def:goodconf_G}. Let $T>0$. For any $\delta>0$ there exists $\ep_{*}=\ep_{*}(\delta, T,M)>0$ %(independent on $R_{max}$) 
such that if $0<\ep_0<\ep_{*}$ and $\phi<\ep_0^{\beta}$, with $\beta=\frac{2}{3(\gamma-1)}+4\delta_1+\frac{\delta_2}{(\gamma-1)}$, then %$\beta=\frac{2}{3(\gamma-1)}+\frac{\delta_2}{\gamma-1}+1$, then
$$\PP(\mathcal{G})\geq 1-\frac{\delta}{2} .$$
\end{proposition}

\begin{proposition}\label{prop:goodcfg}
Let $\{ \mathcal{U}_{k} \}$, $\{ \mathcal{\bar{U}}_{k} \}$ be the sequence of domains defined by \eqref{def:Uk+1} and by \eqref{def:Uk+1compl}. Given $M>1$, let $T>0$ and assume $\delta_3>0$ such that 
\begin{equation}\label{ass:delta3}
\frac{2}{9}\big(\frac{1}{\gamma}(1+\delta_3)+1\big)<\frac{1}{3}
\end{equation} 
with $\delta_2>0$ sufficiently small (depending on $\delta_3$ and $\gamma$). For any $\delta>0$ there exist $\ep_{*}=\ep_{*}(M,T)>0$ and $\phi_{*}>0$ such that if $\ep_0<\ep_{*}$ and $\phi<\phi_{*}$
$$\PP\big(\bigcup_{k=0}^{[\frac{T}{\ep_0}] } \mathcal{\bar{U}}_{k}\big)\leq \frac{\delta}{2},$$ 
for any $V_0\in [0,M]$. 
\end{proposition}
\begin{remark}
Observe that it is possible to find $\delta_3$ such that Assumption \ref{ass:delta3} holds since we assumed $\gamma>2$. 
\end{remark}

\subsubsection{Proof of Proposition \ref{prop:1}}
We will prove Proposition \ref{prop:1} with the help of the following Lemmata.
\begin{lemma}\label{claim:1}
Let $T>0$. For any $\delta>0$ there exists $\ep_{*}=\ep_{*}(\delta,T)>0$ sufficiently small such that if $\ep_0<\ep_{*}$ and $\phi<1$ then
\begin{equation}\label{eq:PW}
\PP(\mathcal{W} )\leq\frac{\delta}{10} 
\end{equation}
where $\mathcal{W}:=\{ \, \omega\in \Omega\; : \;  \exists (x_{j}(\omega),v_j(\omega)) \; \text{s.t.}\;\; v_j(\omega)> V_{*} %\\& 
\; \text{and} \; B_{\phi^{\frac{1}{3}}R_j(\omega)}(x_j(\omega))\cap \mathcal{D}_{*} \neq\emptyset  \}.$
\end{lemma}

\begin{proof} %[Proof of Lemma \ref{claim:1}] 
We consider the collision cylinder $\mathcal{D}_{*}$ defined by \eqref{def:collcyl}. We want to show that the probability of having an intersection between the cylinder $\mathcal{D}_{*}$ and a particle with volume $v\geq V_{*}$ is small.   
To this end we partition the line $[V_{*},+\infty)$  into intervals $[2^lV_{*},2^{l+1}V_{*})$, with $l\geq 0$, 
and define the sets 
\begin{equation}\label{def:W_l}
\begin{split}
\mathcal{W}_{l}:=&\{\omega\in \Omega\; : \; \text{there exists at least one obstacle }\, (x_{j}(\omega),v_j(\omega)) \; \text{s.t.}\;\; v_j(\omega)\in [2^{l}V_{*},2^{l+1}V_{*})	\\& 
\quad  \text{and} \;\; B_{\frac{3}{4\pi}(\phi\, v_j(\omega))^{\frac{1}{3}}}(x_j(\omega))\cap \mathcal{D}_{*} \neq\emptyset \}.
\end{split}
\end{equation}
We recall that the probability of finding an obstacle with volume which varies in the interval $[2^{l}V_{*},2^{l+1}V_{*})$ is given by 
$$\int_{2^{l+}V_{*}}^{2^{l+1}V_{*}}G(v)\,dv\leq \frac{1}{(2^{l}V_{*})^{\gamma}}\int_{2^{l}V_{*}}^{2^{l+1}V_{*}}v^{\gamma}G(v)\,dv\leq \frac{M_{\gamma}}{(2^{l}V_{*})^{\gamma}}$$ where $M_\gamma:=\int_{0}^{\infty}v^{\gamma}G(v)\,dv$. Moreover we define
\begin{equation}\label{def:tildeW_l}
\begin{split}
\tilde{\mathcal{W}_{l}}:=&\{\omega\in \Omega\; : \; \text{there exists at least one obstacle }\, (x_{j}(\omega),v_j(\omega)) \; \text{s.t.}\;\;  x_{j}(\omega)\in \mathcal{C}(l,\ep_0,R_{max})\\&
\quad v_j(\omega)\in [2^{l}V_{*},2^{l+1}V_{*}) \} %\;\;  \text{and} \;\; B_{\frac{3}{4\pi}(v_j(\omega))^{\frac{1}{3}}}(x_j(\omega))\cap \mathcal{D}_{*} \neq\emptyset \}
\end{split}
\end{equation}
where $$\mathcal{C}(l,\ep_0,\phi,T):={B}^{2}_{\phi^{\frac{1}{3}}\rho}(0)\times [-\phi^{\frac{1}{3}}\rho,T\,\phi^{-\frac 2 3}+\phi^{\frac{1}{3}}\rho] \quad\text{with}\quad \rho=\rho(l,\ep_0)=R_{max}+\frac{3}{4\pi}(2^{l+1}V_{*})^{\frac{1}{3}}.$$ 
Note that $\mathcal{C}$ is the enlargement of $\mathcal{D}_{*}$ by an amount of size $\frac{3}{4\pi}(2^{l+1}V_{*})^{\frac{1}{3}}$. It is straightforward to see that $\PP\big(\mathcal{W}_{l}\big)\leq \PP\big(\tilde{\mathcal{W}_{l}}\big)$. Therefore, we look at 
\begin{equation}\label{eq:probW_l} 
\begin{split}
\PP\big(\tilde{\mathcal{W}_{l}}\big)&=\sum_{n=1}^{\infty} \frac{1}{n!}\exp\left(-|\mathcal{C}|\int_{2^{l}V_{*}}^{2^{l+1}V_{*}} G(v)\,dv\right)\;\left(|\mathcal{C}|\int_{2^{l}V_{*}}^{2^{l+1}V_{*}} G(v)\,dv\right)^n\\&
=1-\exp\left(-|\mathcal{C}|\int_{2^{l}V_{*}}^{2^{l+1}V_{*}} G(v)\,dv\right)\\& % since 1-e^{-x}\leq x for x\geq 0
\leq |\mathcal{C}|\int_{2^{l}V_{*}}^{2^{l+1}V_{*}} G(v)\,dv \leq  |\mathcal{C}|\frac{M_{\gamma}}{(2^{l}V_{*})^{\gamma}},
\end{split}
\end{equation}
where we used that $1-e^{-x}\leq x$ for $x\geq 0$. Moreover, using the definition of $\rho$ and \eqref{def:R_max}, we have
\begin{equation*}
\begin{split}
|\mathcal{C}|&\leq \pi \phi^{\frac{2}{3}}\rho^2 (T\phi^{-\frac{2}{3}}+2 \phi^{\frac{1}{3}} \rho))\\&
%=\pi \phi^{\frac{2}{3}}(R_{max}+\frac{3}{4\pi}(2^{l+1}V_{*})^{\frac{1}{3}})^2  (T\phi^{-\frac{2}{3}}+2\phi^{\frac{1}{3}}(R_{max}+\frac{3}{4\pi}(2^{l+1}V_{*})^{\frac{1}{3}}))\\&
\leq 2\,\pi(\lambda^2V_{*}^{\frac{2}{3}}+\frac{9}{16\pi^2}(2^{l+1}V_{*})^{\frac{2}{3}})  (T+2\phi(\lambda V_{*}^{\frac{1}{3}}+\frac{3}{4\pi}(2^{l+1}V_{*})^{\frac{1}{3}}))\\&
\leq %just using phi\leq 1
 C\, (T+1) V_{*}\lambda^3 2^{l}, 
 \end{split}
\end{equation*}
where $C>0$ is a numerical constant.
Thus, from \eqref{eq:probW_l}  we get 
$$\PP(\mathcal{W})=\PP\big(\bigcup_{l=0}^{\infty}\mathcal{W}_{l}\big)\leq \sum_{l=0}^{\infty}  \PP\big(\mathcal{W}_{l}\big)\leq C\, (T+1)M_{\gamma} \lambda^3 V_{*}^{1-\gamma}  \sum_{l=0}^{\infty} 2^{-(\gamma-1)l}\leq \frac{\delta}{10},$$
where we used that $\gamma>2$, Assumption \ref{ass:V_*2}, and we chose $\ep_0>0$ sufficiently small. This concludes the proof. 
\end{proof}
\begin{lemma}\label{claim:2} 
Let $T>0$. For any $\delta>0$ there exist $\ep_{*}=\ep_{*}(\delta,T)>0$ such that if $\ep_0<\ep_{*}$ and $\phi<\ep_0^{\beta}$, $\beta=\frac{2}{9(\gamma-1)}+4\delta_1+\frac{\delta_2}{3(\gamma-1)}$
then 
$$
\PP\big(\{\omega\; : \; n(\mathcal{D}_{*})> \lambda \pi R^{2}_{\max}(T+1) \}\big)\leq\frac{\delta}{10} .$$
\end{lemma} 
We remark that in the statement above we considered $(T+1)$ instead of $T$ to overcome problems connected to the smallness of $T$. 
\begin{proof} 
We recall that $\mathcal{D}_{*}:=({B}^{2}_{\phi^{\frac{1}{3}}R_{max}}(0)\times [0,T\,\phi^{-\frac 2 3}))\cup (\mathcal{S}_{C_t}\cup \mathcal{S}_{C_b})$. Using the assumption $\phi<\ep_0^{\beta}$ as well as \eqref{ass:V_*2} and \eqref{def:R_max} %the ass is needed becuase R_max is growing
we can estimate the diameter of the regions $\mathcal{S}_{C_t}$, $\mathcal{S}_{C_b}$ by $\phi^{-\frac 2 3}$ and thus we obtain
$\mathcal{D}_{*}\subset {B}^{2}_{\phi^{\frac{1}{3}}R_{max}}(0)\times [-(T+1) \,\phi^{-\frac 2 3},2(T+1) \,\phi^{-\frac 2 3}))$. Thus, it follows that $|\mathcal{D}_{*}|\leq 3\pi R^2_{max}(T+1) $. %\leq 3\pi R^2_{max}(T+1)$.  
We consider
\begin{equation*}
\PP\big(\{\omega\; : \; n(\mathcal{D}_{*})> \lambda \pi R^{2}_{\max}(T+1) \}\big)=\sum_{n=\ceil*{\lambda \pi R^{2}_{max}(T+1)}}^{\infty}\frac{|\mathcal{D}_{*}|^n}{n!}e^{-|\mathcal{D}_{*}|}
\end{equation*}
and apply Lemma \ref{lem:Poisson} with $N=N(\ep_0,T)=\ceil*{\lambda \pi R^{2}_{max}(T+1)}$ and $\zeta=|\mathcal{D}_{*}|$. Note that, choosing $\ep_0$ sufficiently small, $N\geq \frac 3 4 \lambda \pi R^{2}_{max}(T+1)$. Then, choosing $\xi_{*}=\frac{4}{\lambda}$ the condition $\zeta \leq \xi_{*}N$ in \eqref{eq:boundPoisson_lem} is satisfied and we get %recall that \ceil*{x}\leq \frac 3 4 \ceil*{x}
$$\sum_{n=\ceil*{\lambda \pi R^{2}_{max}(T+1)}}^{\infty}\frac{|\mathcal{D}_{*}|^n}{n!}e^{-|\mathcal{D}_{*}|}\leq e^{-|\log(\xi_{*})|\ceil*{\lambda \pi R^{2}_{max}(T+1)}}\leq e^{-C |\log(\xi_{*})|\ep_0^{-(\beta+2\delta_1)/3}}$$
with $\beta$ as in the statement of the Lemma and $C=\frac 3 4 \pi (T+1)$. Therefore, if $\ep_0\leq \ep_{*}(\delta,T)$, we have that $e^{-|\log(\xi_{*})|\ep_0^{-(\beta+2\delta_1)/3}}\leq \frac{\delta}{10}$ which proves Lemma \ref{claim:2}.  
\end{proof}

\begin{proofof}[Proof of Proposition \ref{prop:1}]
To complete the proof we need to prove that 
\begin{equation}\label{eq:0part_caps}
\PP(\{\omega\,:\, n(\mathcal{S}_{C_b}\cup \mathcal{S}_{C_t})\geq 1 \})\leq \frac{\delta}{10}.
\end{equation}
It would then follow from Lemma \ref{claim:1}, Lemma \ref{claim:2} and \eqref{eq:0part_caps}, since $\mathcal{G}=\mathcal{W}^{c}\cap \{\omega\; : \; n(\mathcal{D}_{*})\leq \lambda \pi^2R^{2}_{\max}(T+1) \}\cap \{\omega\,:\, n(\mathcal{S}_{C_b}\cup \mathcal{S}_{C_t})=0 \}$ where $\mathcal{W}^{c}$ is the complement of $\mathcal{W}$ introduced in Lemma \ref{claim:1}, that
$$\PP\big(\mathcal{G}\big)\geq 1-\frac{\delta}{10}-\frac{\delta}{10}-\frac{\delta}{10}=1-\frac{3\delta}{10} .$$ 
We observe that to prove \eqref{eq:0part_caps} we have

\begin{equation*}
\begin{split}
\PP(\{\omega\,:\, n(\mathcal{S}_{C_b}\cup \mathcal{S}_{C_t})\geq 1 \})&=\sum_{n=1}^{\infty}\frac{|\mathcal{S}_{C_b}\cup \mathcal{S}_{C_t}|^n}{n!}e^{-|\mathcal{S}_{C_b}\cup \mathcal{S}_{C_t}|}=(1-e^{-|\mathcal{S}_{C_b}\cup \mathcal{S}_{C_t}|})\\&
\leq |\mathcal{S}_{C_b}\cup \mathcal{S}_{C_t}|\leq \frac{4}{3}\pi R^{3}_{max}\phi =\frac{4}{3}\pi \lambda ^3 V_{*}\phi.
\end{split}
\end{equation*}
Due to the assumption $\phi<\ep_0^{\beta}$, using \eqref{ass:V_*2}, we obtain 
$$\PP(\{\omega\,:\, n(\mathcal{S}_{C_b}\cup \mathcal{S}_{C_t})\geq 1 \})\leq C \ep_0^{\delta_1} \leq \frac{\delta}{10}$$
if $\ep_0$ is sufficiently small.

\end{proofof}

\bigskip

\subsubsection{Proof of Proposition \ref{prop:goodcfg}}
Before proving Proposition \ref{prop:goodcfg} we need a preliminary result.
\begin{lemma}\label{lem:unionDj}
Suppose that $M\geq 1$ and $\delta_3, \delta_2$ as in Proposition \ref{prop:goodcfg}. There exists $\ep_{*}=\ep_{*}(M,T)>0$ sufficiently small such that for any $\omega\in \Omega^1(\ep_0,T,\phi)$ if $\phi<\ep_0^{\beta}$ and if $\ep_0 \leq \ep_{*}$ then 
\begin{equation}\label{eq:unionDj}
\bigcup_{j=0}^{k}\mathcal{D}_{j}(\omega_{|_{\bigcup_{l=0}^{j-1}\mathcal{D}_{l}}}))\subset \mathcal{D}_{*}\quad \forall \; 0\leq k\leq \left[{\frac{T}{\ep_0}}\right],
\end{equation}
for any $V_0\in [0,M]$. Here the $\{\mathcal{D}_{k}(\omega_{|_{\bigcup_{l=0}^{k-1}\mathcal{D}_{l}}}))\}$ are the domains defined by \eqref{eq:defDk}. Moreover we have
\begin{equation}\label{eq:bdrhok}
\rho_{k}\leq %\bar{C}(V_0^{\frac 1 3}+(\lambda^{\frac{1}{\gamma}(1+\delta_3)+1} V_{*}^{\frac{2}{9}(\frac{1}{ \gamma}(1+\delta_3)+1)}(T+1)^{\frac{1}{3}(\frac{1}{\gamma}(1+\delta_3)+1})+V_{*}^{\frac{1}{3}})\leq
 \frac{\lambda}{2} V_{*}^{\frac 1 3}.
\end{equation}
%with $\bar{C}$ is a numerical constant.
\end{lemma}
\begin{proof}
We prove $\eqref{eq:unionDj}$ arguing by induction for $k=0,\dots, \left[{\frac{T}{\ep_0}}\right]$. 
For $k=0$ it is straightforward using the definition of $\mathcal{D}_0$ and $\mathcal{D}_{*}$, see \eqref{def:collcyl} and  \eqref{eq:defD0}, if $\ep_{0}<\ep_{*}(M,T)$. Then, we assume that \eqref{eq:unionDj} holds for $k$ and we want to show that it holds for $k+1$.
From Lemma \ref{app:claim} in Appendix \ref{appendix1} we have the following bound for the displacement of the center of mass of the tagged particle at level $k$, namely
 \begin{equation}\label{eq:Xkbd}
%|X_{k}|\leq C \phi^{\frac{1}{3}} ((V_{k}^{\frac{1}{3}}-V_{0}^{\frac{1}{3}})+V_{*}^{\frac{1}{3}}(\log V_{k}-\log V_{0})),
|X_{k}|\leq\frac{9}{2\pi} \phi^{\frac{1}{3}} (V_{k}^{\frac{1}{3}}-V_{0}^{\frac{1}{3}})
\end{equation}
while the volume satisfies
\begin{equation}\label{eq:Vkdef}
V_{k}=V_{0}+\sum_{j=0}^{k}v_{j}.
\end{equation}
Using that $\omega\in \mathcal{U}_k$ with $\mathcal{U}_k$ defined by \eqref{def:Uk+1}
\begin{equation}\label{eq:bdsumvj}
 \begin{split}
\sum_{j=0}^{k}v_{j}&\leq \sum_{\underset{\{j\,:\, \theta(j,\omega)-\theta(j-1,\omega)=1\}}{j=0}}^{k}(1+\theta(j,\omega))^{\frac{1}{\gamma}(1+\delta_3)}%\\&
\leq \sum_{j=0}^{N_{coll}}(1+j)^{\frac{1}{\gamma}(1+\delta_3)} \leq (N_{coll}+2)^{\frac{1}{\gamma}(1+\delta_3)+1} 
\end{split}
\end{equation}
where $ N_{coll}=n(\mathcal{D}_{*})$. Since $\mathcal{U}_k\subset \mathcal{G}$ and thanks to \eqref{def:goodconf_G}, using also \eqref{def:R_max}, then we have  
$$N_{coll}\leq \pi \lambda R^2_{max}(T+1)\leq \pi \lambda^3 V_{*}^{\frac{2}{3}}(T+1).$$
Therefore, from \eqref{eq:Vkdef} using \eqref{eq:bdsumvj}, we get 
 \begin{equation}\label{eq:Vkbd}
 \begin{split}
V_{k}&\leq %V_{0}+\sum_{\underset{\{j\,:\, \theta(j,\omega)-\theta(j-1,\omega)=1\}}{j=0}}^{k}(1+\theta(j,\omega))^{\frac{1}{\gamma}(1+\delta_3)}\leq V_{0}+\sum_{j=0}^{N_{coll}}(1+j)^{\frac{1}{\gamma}(1+\delta_3)}\\&\leq  
V_{0}+(N_{coll}+2)^{\frac{1}{\gamma}(1+\delta_3)+1} \leq V_{0}+(\pi \lambda^3V_{*}^{\frac{2}{3}}(T+1)+2)^{\frac{1}{\gamma}(1+\delta_3)+1}.
\end{split}
\end{equation}
We observe that $(V_{0}+\sum_{j=0}^{k}v_{j})^{\frac{1}{3}}-V_0^{\frac{1}{3}}\leq (\sum_{j=0}^{k}v_{j})^{\frac{1}{3}}$. %because (1+x)\leq (x^{\frac{1}{3}}+1)^3 
Thus, from \eqref{eq:Xkbd}, using \eqref{eq:Vkdef} and \eqref{eq:Vkbd}, we obtain
\begin{equation}\label{eq:Xkbd2}
\begin{split}
|X_{k}|&\leq  \frac{9}{2\pi} \phi^{\frac{1}{3}}((\pi \lambda^3V_{*}^{\frac{2}{3}}(T+1))^{\frac{1}{\gamma}(1+\delta_3)+1}+2)^{\frac{1}{3}}\\&
%\leq C \phi^{\frac{1}{3}}(\lambda^{\frac{1}{\gamma}(1+\delta_3)+1} V_{*}^{\frac{2}{9}(\frac{1}{ \gamma}(1+\delta_3)+1)}(T+1)^{\frac{1}{3}(\frac{1}{\gamma}(1+\delta_3)+1}+2^{\frac{1}{3}})\\&
\leq  \frac{9}{\pi}\phi^{\frac{1}{3}}(\lambda^{\frac{1}{\gamma}(1+\delta_3)+1} V_{*}^{\frac{2}{9}(\frac{1}{ \gamma}(1+\delta_3)+1)}(T+1)^{\frac{1}{3}(\frac{1}{\gamma}(1+\delta_3)+1})\leq \frac{\lambda}{2} V_{*}^{\frac 1 3}.
%+V_{*}^{\frac{1}{3}}(\log V_{k}-\log V_{0})).
\end{split}
\end{equation}
This inequality holds for $\ep_0$ sufficiently small thanks to Assumption \ref{ass:delta3} with $\delta_2$ sufficiently small. %Notice that $\delta_2>\delta_1>0.$

We recall that $\rho_{k}=\frac{3}{4\pi}(V_{k}^{\frac{1}{3}}+V_{*}^{\frac{1}{3}})$. Hence, thanks to \eqref{eq:Vkbd}
\begin{equation*}%\label{eq:bdrhok}
\rho_{k}\leq \bar{C}(V_0^{\frac 1 3}+(\lambda^{\frac{1}{\gamma}(1+\delta_3)+1} V_{*}^{\frac{2}{9}(\frac{1}{ \gamma}(1+\delta_3)+1)}(T+1)^{\frac{1}{3}(\frac{1}{\gamma}(1+\delta_3)+1})+V_{*}^{\frac{1}{3}})%\leq \frac{\lambda}{2} V_{*}^{\frac 1 3}
\end{equation*}
where $\bar{C}$ is a numerical constant. Then \eqref{eq:bdrhok} holds for $\ep_0$ (depending on $M$) sufficiently small thanks to Assumption \ref{ass:delta3} with $\delta_2$ sufficiently small. 
Therefore, we have that %\mathcal{D}_{k}(\omega_{|_{\bigcup_{l=0}^{k-1}\mathcal{D}_{l}}})):=B_{\rho_{k+1}\phi^{\frac{1}{3}}}(X_{k+1})\times [k\ep_0 \phi^{-\frac{2}{3}},(k+1) \ep_0 \phi^{-\frac{2}{3}})
$$\mathcal{D}_{k}(\omega_{|_{\bigcup_{l=0}^{k-1}\mathcal{D}_{l}}}))\subset B_{|X_{k}|+\rho_k \phi^{\frac{1}{3}}}(0)\times [k\ep_0 \phi^{-\frac{2}{3}},(k+1) \ep_0 \phi^{-\frac{2}{3}}) \subset B_{Z_{k}\phi^{\frac{1}{3}}}(0)\times [k\ep_0 \phi^{-\frac{2}{3}},(k+1) \ep_0 \phi^{-\frac{2}{3}})$$
with 
$Z_k\leq  %\bar{C}((\lambda^{\frac{1}{\gamma}(1+\delta_3)+1} V_{*}^{\frac{2}{9}(\frac{1}{ \gamma}(1+\delta_3)+1)}(T+1)^{\frac{1}{3}(\frac{1}{\gamma}(1+\delta_3)+1})+V_{*}^{\frac{1}{3}})\leq 
\lambda V_{*}^{\frac{1}{3}}=R_{max}.$
\end{proof}

We can now conclude the proof of Proposition \ref{prop:goodcfg}.

\begin{proofof}[Proof of Proposition \ref{prop:goodcfg}]
We consider at first $\PP(\mathcal{\bar{U}}_{0})$ where $\mathcal{\bar{U}}_{0}$ is defined by \eqref{def:U0compl}. We observe that the volume of $\mathcal{D}_0$ is $|\mathcal{D}_0|=\pi (\rho_0\phi^{\frac 1 3})^2\ep_0 \phi^{-\frac{2}{3}}=\pi \rho_0^2 \ep_0$ and we obtain
\begin{equation}\label{eq:pbarU_0}
\begin{split}
\PP(\mathcal{\bar{U}}_{0})&= \PP(\{\omega\in\Omega: \; \#(\{x_k(\omega)\}\cap \mathcal{D}_0)> 1 \;\;\text{or}\;\; [\#(\{x_k(\omega)\}\cap \mathcal{D}_0)=1 \;\;\text{and}\;\; v_{k}(\omega)\in \mathcal{D}_0,\; v_{k}(\omega)> 1]  \})\\&
\leq \sum_{n=2}^{\infty} \frac{|\mathcal{D}_0|^n}{n!}e^{-|\mathcal{D}_0|}+e^{-|\mathcal{D}_0|\int_{1}^{\infty}G(v)dv}|\mathcal{D}_0|\int_{1}^{\infty}G(v)dv\\&
=e^{-|\mathcal{D}_0|}(e^{|\mathcal{D}_0|}-1-|\mathcal{D}_0|)+e^{-|\mathcal{D}_0|\int_{1}^{\infty}G(v)dv}|\mathcal{D}_0|\int_{1}^{\infty}G(v)dv\\&
 \leq \frac{|\mathcal{D}_0|^2}{2} + |\mathcal{D}_0| M_{\gamma} e^{-|\mathcal{D}_0|\int_{1}^{\infty}G(v)dv}\\&
 \leq \frac{\pi^2}{2} \rho_0^4 \ep_0^2+\pi \rho_0^2 \ep_0 M_{\gamma}.
\end{split}
\end{equation}

We now consider $\PP(\mathcal{\bar{U}}_{k})$ with $\mathcal{\bar{U}}_{k}$ defined by \eqref{def:Uk+1compl}. Since $\mathcal{\bar{U}}_{k}$ results to be the union of disjoint sets we have
\begin{equation}\label{eq:splitUk}
\begin{split}
\PP(\mathcal{\bar{U}}_{k})&=\PP(\mathcal{\bar{U}}_{k}^{(1)})+\PP(\mathcal{\bar{U}}_{k}^{(2)})+\PP(\mathcal{\bar{U}}_{k}^{(3)}).
\end{split}
\end{equation}
With the same strategy used in \eqref{eq:pbarU_0} we get 
\begin{equation}\label{eq:pbarU_k^(1)}
\PP(\mathcal{\bar{U}}_{k}^{(1)})\leq \frac{\pi^2}{2} \rho_{k}^4 \ep_0^2, %+\pi \rho_{k}^2 \ep_0 M_{\gamma},
\end{equation}
where we used that $|\mathcal{D}_{k}|= \pi \rho_{k}^2 \ep_0$.
We look at $\PP(\mathcal{\bar{U}}_{k}^{(2)})$. We set $(x_k,v_k)=\eta_k$ so that $d\eta_{k}=G(v_k)\,dx_k\, dv_k$ and use the convention $\int^{n_k}_{\mathcal{D}_k\times \R^{+}}d\eta_k =1$ if $n_k=0$ and $\int^{n_k}_{\mathcal{D}_k\times \R^{+}}d\eta_k =\int_{\mathcal{D}_k\times \R^{+}} d\eta_k$ if $n_k=1$ so that
\begin{equation}\label{eq:pbarU_k^(2)}
\begin{split}
\PP(\mathcal{\bar{U}}_{k}^{(2)})&=\sum_{n_0\in\{0,1\}}\dots\sum_{n_{k-2}\in\{0,1\}}e^{-|\mathcal{D}_{0}|}\int_{\mathcal{D}_{0}\times [0,\infty)}^{n_0}d\eta_0\dots  \int_{\mathcal{D}_{k-2}(\xi^{(k-3)})\times [0,\infty)}^{n_{k-2}}d\eta_{k-2}e^{-|\mathcal{D}_{k-2}(\xi^{(k-3)})|}\\&
\quad \int_{\mathcal{D}_{k-1}(\xi^{(k-2)})\times [0,\infty)}^{n_{k-1}=1}\hspace{-7mm}d\eta_{k-1}e^{-|\mathcal{D}_{k-1}(\xi^{(k-2)})|}\int_{\mathcal{D}_{k}(\xi^{(k-1)})\times [0,\infty)}^{n_{k}=1}\hspace{-7mm}d\eta_{k}e^{-|\mathcal{D}_{k}(\xi^{(k-1)})|}\chi(\{\omega\,:\; \omega\in\mathcal{U}_{k-1}\})\\&
\leq \sum_{n_0\in\{0,1\}}\dots\sum_{n_{k-2}\in\{0,1\}}e^{-|\mathcal{D}_{0}|}\int_{\mathcal{D}_{0}\times [0,\infty)}^{n_0}d\eta_0\dots  \int_{\mathcal{D}_{k-2}(\xi^{(k-3)})\times [0,\infty)}^{n_{k-2}}d \eta_{k-2}e^{-|\mathcal{D}_{k-2}(\xi^{(k-3)})|}\\&
\quad\;|\mathcal{D}_{k-1}(\xi^{(k-2)})||\mathcal{D}_{k}(\xi^{(k-1)})|\\&
\leq  \pi^2 \rho_{k}^2\rho_{k-1}^2 \ep_0^2\leq   \frac{\pi^2}{2}( \rho_{k}^4+\rho_{k-1}^4) \ep_0^2.
\end{split}
\end{equation}
Writing these formulas we are using standard properties of the Poisson measure. The domains $\mathcal{D}_{k}$ are defined by \eqref{eq:defDk}. 
Thus, combining \eqref{eq:pbarU_k^(1)} and \eqref{eq:pbarU_k^(2)}, we get %\rho_{k-1}\leq \rho_k$
\begin{equation}
\PP(\mathcal{\bar{U}}_{k}^{(1)})+\PP(\mathcal{\bar{U}}_{k}^{(2)})\leq \frac{5\pi^2}{2} \rho_{k}^4 \ep_0^2.
\end{equation}
By using \eqref{eq:bdrhok} we have that $\rho_{k}^4 \leq \lambda^4V_{*}^{\frac 4 3}$. 
Then
\begin{equation}\label{eq:Uk1Uk2}
\PP(\mathcal{\bar{U}}_{k}^{(1)})+\PP(\mathcal{\bar{U}}_{k}^{(2)})\leq \frac{5\pi^2}{2} \lambda^4V_{*}^{\frac 4 3} \ep_0^2.
\end{equation}
Therefore, using \eqref{eq:splitUk} and \eqref{eq:Uk1Uk2}, we have 
\begin{equation}\label{eq:psumUk}
\begin{split}
\PP\big(\bigcup_{k=0}^{[\frac{T}{\ep_0}] } \mathcal{\bar{U}}_{k}\big)&\leq %\sum_{k=0}^{[\frac{T}{\ep_0}]}\PP(\mathcal{\bar{U}}_{k})\leq 
\sum_{k=0}^{[\frac{T}{\ep_0}]}(\PP(\mathcal{\bar{U}}_{k}^{(1)})+\PP(\mathcal{\bar{U}}_{k}^{(2)})+\PP(\mathcal{\bar{U}}_{k}^{(3)})) \\&
\leq \sum_{k=0}^{[\frac{T}{\ep_0}]}(\PP(\mathcal{\bar{U}}_{k}^{(1)})+\PP(\mathcal{\bar{U}}_{k}^{(2)}))+ \sum_{k=0}^{[\frac{T}{\ep_0}]} \PP(\mathcal{\bar{U}}_{k}^{(3)})\\&
 \leq \frac{5}{2}\pi^2 \lambda^4 \,T\,V_{*}^{\frac 4 3} \ep_0+ \sum_{k=0}^{[\frac{T}{\ep_0}]}\PP(\mathcal{\bar{U}}_{k}^{(3)}).
\end{split}
\end{equation}
The first term on the right hand side tends to zero as $\ep_0 \to 0$
using the explicit expression for $V_{*}$ given by \eqref{ass:V_*2} if $\delta_2$ is sufficiently small. 

We now look at the second term on the right hand side of \eqref{eq:psumUk}. We have
\begin{equation}\label{eq:pbarU_k^(3)}
\begin{split}
\PP(\mathcal{\bar{U}}_{k}^{(3)})&=\PP(\{\omega\in\mathcal{U}_{k-1}\!:\#(\{x_j(\omega)\}\cap \mathcal{D}_{k}(\omega_{|_{\bigcup_{l=0}^{k-1}\mathcal{D}_{l}}}\!\!))\!=\!1 \; \text{and}\;  \#(\{x_j(\omega)\}\cap \mathcal{D}_{k-1}(\omega_{|_{\bigcup_{l=0}^{k-2}\mathcal{D}_{l}}}\!\!))\!=\!0\\&
 \hspace{1cm} \text{and} \,\, v_j(\omega)> (1+\theta(k,\omega))^{\frac{1}{\gamma}(1+\delta_3)} \})\\&
= \!\sum_{n_0\in\{0,1\}}\!\!\dots\!\!\sum_{n_{k-2}\in\{0,1\}} \hspace{-0.4cm}e^{-|\mathcal{D}_{0}|}\int_{\mathcal{D}_{0}\times [0,\infty)}^{n_0}\hspace{-0.5cm} d\eta_0\dots  \int_{\mathcal{D}_{k-2}(\xi^{(k-3)})\times [0,\infty)}^{n_{k-2}}\hspace{-0.5cm} d\eta_{k-2}\,e^{-|\mathcal{D}_{k-2}(\xi^{(k-3)})|}\\&
% \int_{\mathcal{D}_{k-1}(\xi^{(k-2)})\times [0,\infty)}^{n_{k-1}=0}d\eta_{k-1}
 \hspace{0.8cm} e^{-|\mathcal{D}_{k-1}(\xi^{(k-2)})|}\int_{\mathcal{D}_{k}(\xi^{(k-1)})\times [0,\infty)}^{n_{k}=1}d\eta_{k}\,e^{-|\mathcal{D}_{k}(\xi^{(k-1)})|}\chi(\{\omega\,:\; \omega\in\mathcal{U}_{k-1}\})\times\\&
 \hspace{0.8cm}\times \chi(\{v_j(\omega)> (1+\theta(k,\omega))^{\frac{1}{\gamma}(1+\delta_3)} \})\\&
 = \!\sum_{n_0\in\{0,1\}}\!\!\dots\!\!\sum_{n_{k-2}\in\{0,1\}}\hspace{-0.4cm}e^{-|\mathcal{D}_{0}|}\int_{\mathcal{D}_{0}\times [0,\infty)}^{n_0}\hspace{-0.5cm} d\eta_0\dots  \int_{\mathcal{D}_{k-2}(\xi^{(k-3)})\times [0,\infty)}^{n_{k-2}} \hspace{-0.9cm} d\eta_{k-2}\, e^{-|\mathcal{D}_{k-2}(\xi^{(k-3)})|-|\mathcal{D}_{k-1}(\xi^{(k-2)})|}\\&
% \int_{\mathcal{D}_{k-1}(\xi^{(k-2)})\times [0,\infty)}^{n_{k-1}=0}d\eta_{k-1}
 %\hspace{0.8cm} e^{-|\mathcal{D}_{k-1}(\xi^{(k-2)})|} 
\hspace{0.6cm}  \int_{(\mathcal{D}_{k}(\xi^{(k-1)})\cap \{v_k(\omega)> (1+\theta(k,\omega))^{\frac{1}{\gamma}(1+\delta_3)} \} )\times [0,\infty)}^{n_{k}=1}\hspace{-0.5cm}d\eta_{k}\,e^{-|\mathcal{D}_{k}(\xi^{(k-1)})|}\chi(\{\omega\,:\; \omega\in\mathcal{U}_{k-1}\}).
\end{split}
\end{equation}
Note that the characteristic function $\chi(\{\omega\,:\; \omega\in\mathcal{U}_{k-1}\})=\chi(\{\omega\,:\; \omega\in \bigcap_{j=0}^{k-1}\mathcal{U}_{j}\})$ since the sequence $\{\mathcal{U}_{j}\}$ is decreasing in $j$. 
We observe that
\begin{equation*}
\begin{split}
&\int_{(\mathcal{D}_{k}(\xi^{(k-1)})\cap\{v_k(\omega)> (1+\theta(k,\omega))^{\frac{1}{\gamma}(1+\delta_3)}\})\times [0,\infty)}^{n_{k}=1}d\eta_{k} e^{-|\mathcal{D}_{k}(\xi^{(k-1)}|}\chi(n_k=1)\\&
\leq \chi(n_k=1)\frac{M_{\gamma}}{(1+\theta(k,\xi^{(k-1)}))^{1+\delta_3}}|\mathcal{D}_{k}(\xi^{(k-1)}|\\&
= \chi(n_k=1)\frac{M_{\gamma}}{(1+\theta(k,\xi^{(k-1)}))^{1+\delta_3}}\pi\ep_0\rho_k^2 
\\& \leq \chi(n_k=1)\frac{M_{\gamma}}{(1+\theta(k,\xi^{(k-1)}))^{1+\delta_3}}\pi\ep_0 \frac{\lambda^2}{4} V_{*}^{\frac 2 3},%\leq \chi(n_k=1)\frac{M_{\gamma}}{(1+\theta(k,\xi^{(k-1)}))^{1+\delta_3}}\pi (\ep_0^{\alpha}) perch� \ep_0\rho_k^2\leq \ep_0^{\alpha}
\end{split}
\end{equation*}
where in the last inequality we used \eqref{eq:bdrhok} and $\theta$ is defined as in \eqref{def:f-teta}. Therefore \eqref{eq:pbarU_k^(3)} becomes
\begin{equation}\label{eq:pbarUk^3bis} 
\begin{split}
\PP(\mathcal{\bar{U}}_{k}^{(3)})
&\leq \frac{\pi\ep_0\lambda^2}{4} V_{*}^{\frac 2 3} \!\!\sum_{n_0\in\{0,1\}}\!\!\dots\!\!\sum_{n_{k-2}\in\{0,1\}}\hspace{-0.4cm}e^{-|\mathcal{D}_{0}|}\int_{\mathcal{D}_{0}\times [0,\infty)}^{n_0}\hspace{-0.5cm}d\eta_0\dots  \int_{\mathcal{D}_{k-2}(\xi^{(k-3)})\times [0,\infty)}^{n_{k-2}}\hspace{-0.9cm}d\eta_{k-2}\,e^{-|\mathcal{D}_{k-2}(\xi^{(k-3)})|}\\&
% \int_{\mathcal{D}_{k-1}(\xi^{(k-2)})\times [0,\infty)}^{n_{k-1}=0}d\eta_{k-1}
\quad \int_{\mathcal{D}_{k-1}(\xi^{(k-2)})\times [0,\infty)}^{n_{k-1}=0}\hspace{-0.8cm}d\eta_{k-1}\,e^{-|\mathcal{D}_{k-1}(\xi^{(k-2)})|}\frac{M_{\gamma}\chi(n_k=1)}{(1+\theta(k,\xi^{(k-1)}))^{1+\delta_3}}\,\chi(\{\omega:\omega\in\mathcal{U}_{k-1}\})\\&
\leq \frac{\pi\ep_0\lambda^2}{4} V_{*}^{\frac 2 3} \sum_{n_0=0}^{\infty}\dots\sum_{n_{k-2}=0}^{\infty}\frac{e^{-|\mathcal{D}_{0}|}}{n_0!\dots n_{k-1}!}\int_{\mathcal{D}_{0}\times [0,\infty)}^{n_0}d\eta_0\dots  \int_{\mathcal{D}_{k-2}(\xi^{(k-3)})\times [0,\infty)}^{n_{k-2}} \hspace{-0.6cm} d\eta_{k-2}\\&
% \int_{\mathcal{D}_{k-1}(\xi^{(k-2)})\times [0,\infty)}^{n_{k-1}=0}d\eta_{k-1}
\quad e^{-|\mathcal{D}_{k-2}(\xi^{(k-3)})|}\int_{\mathcal{D}_{k-1}(\xi^{(k-2)})\times [0,\infty)}^{n_{k-1}=0}d\eta_{k-1}e^{-|\mathcal{D}_{k-1}(\xi^{(k-2)})|}\frac{M_{\gamma}\,\chi(n_k=1)}{(1+\sum_{l=0}^{k-1} n_l)^{1+\delta_3}}
 %\quad\quad   \frac{\pi\ep_0\lambda^2}{4} V_{*}^{\frac 2 3}.
\end{split}
\end{equation}
where we used that $\theta(k,\xi^{(k-1)})= \sum_{l=0}^{k-1} n_l$. %\textcolor{red}{Check if it is an identity!}
We observe that, in order to take into account more than one collision, i.e. $n_k\geq 1$ for any $k$, in the previous formula we generalized the definition of the domains $\mathcal{D}_{k}(\xi^{(k-1)})$ given by \eqref{eq:defDk}, according to the following rules for multiple collision events: $\rho_{k+1}=\frac{3}{4\pi}(V_{k+1}^{\frac{1}{3}}+V_{*}^{\frac{1}{3}})$, $X_{k+1}=X_{k}+\frac{\sum_{j=1}^{n_{k}}v^{(k+1)}_{j}(x^{(k+1)}_{j}-X_{k})}{V_{k}+\sum_{j=1}^{n_{k}}v^{(k+1)}_{j}}$, $V_{k+1}=V_{k}+\sum_{j=1}^{n_{k}}v^{(k+1)}_{j}$. Moreover, since  %\sum_{l=1}^{n_{k-1}}v_l^{(k-1)}
\begin{equation}\label{eq:unitsum}
%\begin{split}
1=%&
  \sum_{n_k=0}^{\infty}\dots\! \sum_{n_L=0}^{\infty}\frac{e^{-|\mathcal{D}_{k}(\xi^{(k-1)})|}}{n_{k}!}\int_{\mathcal{D}_{k}(\xi^{(k-1)})\times [0,\infty)}^{n_{k}}d\eta_{k}%\\&\quad %\int_{\mathcal{D}_{k+1}(\xi^{(k)})\times [0,\infty)}^{n_{k+1}}d\eta_{k+1}
 \dots  \frac{e^{-|\mathcal{D}_{L}(\xi^{(L-1)})|}}{n_{L}!}\int_{\mathcal{D}_{L}(\xi^{(L-1)})\times [0,\infty)}^{n_{L}}d\eta_{L} \,
%\end{split}
\end{equation}
for any $k=0,\dots, L$ and $L\geq 0$, using \eqref{eq:unitsum} in \eqref{eq:pbarUk^3bis} we get
\begin{equation}\label{eq:pbarU_k^(3)-2}
\begin{split}
\PP(\mathcal{\bar{U}}_{k}^{(3)})\leq &
\frac{\pi\ep_0\lambda^2}{4} V_{*}^{\frac 2 3} \sum_{n_0=0}^{\infty}\dots\sum_{n_{k-1}=0}^{\infty}\frac{e^{-|\mathcal{D}_{0}|}}{n_0!\dots n_{k-1}!}\int_{\mathcal{D}_{0}\times [0,\infty)}^{n_0}d\eta_0\dots  \int_{\mathcal{D}_{k-2}(\xi^{(k-3)})\times [0,\infty)}^{n_{k-2}}d\eta_{k-2}\\&
% \int_{\mathcal{D}_{k-1}(\xi^{(k-2)})\times [0,\infty)}^{n_{k-1}=0}d\eta_{k-1}
e^{-|\mathcal{D}_{k-2}(\xi^{(k-3)})|}\int_{\mathcal{D}_{k-1}(\xi^{(k-2)})\times [0,\infty)}^{n_{k-1}}d\eta_{k-1}e^{-|\mathcal{D}_{k-1}(\xi^{(k-2)})|}\frac{M_{\gamma}\,\chi(n_k=1)}{(1+\sum_{l=0}^{k-1} n_l)^{1+\delta_3}}\\& 
\sum_{n_k=0}^{\infty}\frac{e^{-|\mathcal{D}_{k}(\xi^{(k-1)})|}}{n_{k}!} \int_{\mathcal{D}_{k}(\xi^{(k-1)})\times [0,\infty)}^{n_{k}}\hspace{-0.6cm}d\eta_{k}\dots\sum_{n_L=0}^{\infty} \frac{e^{-|\mathcal{D}_{L}(\xi^{(L-1)})|}}{n_{L}!}\int_{\mathcal{D}_{L}(\xi^{(L-1)})\times [0,\infty)}^{n_{L}}\hspace{-0.6cm}d\eta_{L}.
\end{split}
\end{equation}
We look at $\sum_{k=0}^{L}\PP(\mathcal{\bar{U}}_{k}^{(3)})$ when $L=[\frac{T}{\ep_0}]$. More precisely
\begin{equation}\label{eq:psumbarU_k^(3)}
\begin{split}
\sum_{k=0}^{L}\PP(\mathcal{\bar{U}}_{k}^{(3)})&\leq \frac{\pi\ep_0\lambda^2}{4} V_{*}^{\frac 2 3} \,M_{\gamma}\sum_{n_0=0}^{\infty}\dots \sum_{n_L=0}^{\infty} \int_{\mathcal{D}_{k}(\xi^{(k-1)})\times [0,\infty)}^{n_{k}}d\eta_{k}\dots  \int_{\mathcal{D}_{L}(\xi^{(k-1)})\times [0,\infty)}^{n_{L}}d\eta_{L} \\& e^{-|\mathcal{D}_{k}(\xi^{(k-1)})|-|\mathcal{D}_{L}(\xi^{(L-1)})|}\frac{1}{n_{k}!\dots n_{L}!}\sum_{k=1}^{L}\frac{\chi(n_k=1)}{(1+\sum_{l=0}^{k-1}n_l)^{1+\delta_3}}.
\end{split}
\end{equation}
Using that
 $$\sum_{m=1}^{L}\frac{\chi(n_k=1)}{(1+\sum_{l=0}^{m-1}n_l)^{1+\delta_3}}\leq \sum_{m=1}^{\infty}\frac{\chi(n_k=1)}{(1+\sum_{l=0}^{m-1}n_l)^{1+\delta_3}}\leq \sum_{m=1}^{\infty}\frac{1}{(1+m)^{1+\delta_3}}\leq C_{\delta_3},$$
 we obtain 
 \begin{equation}\label{eq:psumbarU_k^(3)-1}
\begin{split}
\sum_{k=0}^{[\frac{T}{\ep_0}]}\PP(\mathcal{\bar{U}}_{k}^{(3)})&\leq \pi \,\ep_0 \frac{\lambda^2}{4} V_{*}^{\frac 2 3} M_{\gamma}C_{\delta_3},
\end{split}
\end{equation}
which tends to zero as $\ep_0\to 0$ due to Assumption \ref{ass:V_*2} since $\gamma>2$. To conclude the proof we plug \eqref{eq:psumbarU_k^(3)-1} in \eqref{eq:psumUk}. 
\end{proofof}

\subsection{Adjoint equation of \eqref{eq:FKScaled} and its Duhamel representation formula}\label{sec:adjeq}
At this level we define a suitable adjoint equation of \eqref{eq:FKScaled} and we write the Duhamel representation formula for its solution which will be a useful tool for the proof of the kinetic limit in Section \ref{sec:derivation}.
We consider a function $\Psi(Y,V,t)$ on the state space and we write the evolution equation for this function. 
The resulting equation is the so called backward Kolmogorov equation which reads as
\begin{equation}\label{eq:BKScaled}
\partial_t\Psi(Y,V,t)=(\mathcal{C}[\Psi])(Y,V,t)
\end{equation}
with
\begin{equation}\label{op:CScaled}
\begin{split}
(\mathcal{C}[\Psi])(Y,V,t)=&\,U\!\int_0^{\infty}\hspace{-2.5mm}dv\int_0^{\frac{\pi}{2}}\hspace{-1.5mm}d\theta\int_0^{2\pi}\hspace{-2mm}d\varphi\, G(v) (R+r)^2 \sin\theta\cos\theta
 \\
 &\,\times \left[\Psi(Y+\frac{v}{V+v}{n}(\theta,\varphi)(R+r),V+v,t)-\Psi(Y,V,t)\right]. \\
\end{split}
\end{equation}
Here ${n}(\theta,\varphi):=(\cos\theta e_1+\sin\theta \,\nu)$.

We remark that the operator $\mathcal{C}$ defined in \eqref{op:CScaled} is the generator of the Markov process for the tagged particle where a given particle described by the coordinates $(Y,V)$, with $Y:=X-Ut e_1$, is transformed in a particle with new coordinates $(Y',V')$ according to \eqref{def:merg_dyn}.
Moreover, the operator $\mathcal{C}$ defined in \eqref{op:CScaled} transforms continuous functions in continuous functions. % vero anche se G è una misura
The connection between equations \eqref{eq:FKScaled} and \eqref{eq:BKScaled} is the following. If $f$ and $\Psi$ are sufficiently smooth functions we have
\begin{equation}\label{duality}
\int_{\R^3	\times [0,\infty)}f(Y,V,t)\psi(Y,V,0)dYdV=\int_{\R^3\times [0,\infty)} f(Y,V,0)\psi(Y,V,t)dYdV.
\end{equation}
Indeed,  if $f$ and $\Psi$ are sufficiently smooth to justify the computations, this follows from
\begin{equation}\label{eq:duality2}
\begin{split}
&\partial_t \left(\int f(Y,V,T-t)\Psi(Y,V,t)\,dY\,dV\right) \\&
= -\int \partial_t f(Y,V,T-t)\Psi(Y,V,t)\,dY\,dV+\int f(Y,V,T-t)\partial_t \Psi(Y,V,t)\,dY\,dV\\&
%= -\int (\mathcal{Q}[f])(Y,V,T-t)\Psi(Y,V,t)\,dY\,dV+\int f(Y,V,T-t)(\mathcal{C}[\Psi])(Y,V,t)\,dY\,dV.
= -\langle \mathcal{Q}[f],\Psi \rangle+\langle\mathcal{C}[\Psi],f\rangle=0,
\end{split}
\end{equation}
where $\mathcal{Q}[f]$ is defined by \eqref{eq:FKScaled}.
The last identity follows using that %$Y'=Y+\frac{v}{V+v}(R+r)\,{n}(\theta,\varphi)$, $V'=V+v$, $R'=R+r=\left(\frac{3}{4\pi}\right)^{\frac{1}{3}}(V^{\frac{1}{3}}+v^{\frac{1}{3}})$
\begin{equation}\label{eq:duality3}
\begin{split}
\langle \mathcal{C}[\Psi],f\rangle &=\int\mathcal{C}\Psi(Y,V,t)f(Y,V,t)\,dVdY\\
&=\,U\!\int_{\R^3}dY\int_0^{\infty}\hspace{-2.5mm}dV\int_0^{\infty}\hspace{-2.5mm}dv\int_0^{\frac{\pi}{2}}\hspace{-1.5mm}d\theta\int_0^{2\pi}\hspace{-2mm}d\varphi\, G(v) (R+r)^2 \sin\theta\cos\theta\\
&\quad\times \Big[f(Y,V,t)\Psi(Y+\frac{v}{V+v}(R+r)\,{n}(\theta,\varphi),V+v,t)-f(Y,V,t)\Psi(Y,V,t)\Big], \\\nonumber
\end{split}
\end{equation}
as well as the change of variables $Y'=Y+\frac{v}{V+v}(R+r)\,{n}(\theta,\varphi)$, $V'=V+v$, $R'=R+r=\left(\frac{3}{4\pi}\right)^{\frac{1}{3}}(V^{\frac{1}{3}}+v^{\frac{1}{3}})$ and Fubini Theorem. 
We now write the solution of \eqref{eq:BKScaled} by using Duhamel's formula. %, perturbing around the ``loss" term.   
We set 
\begin{equation}\label{def:oplingain}
l(V,v,\theta,\varphi):=\frac{v}{V+v}(R+r)\,{n}(\theta,\varphi) 
\end{equation} 
with $R=R(V)$, $r=r(v)$, the transition kernel 
\begin{equation}\label{eq:transkernel}
k(V,v,\theta):=G(v)(R+r)^2 \sin\theta\cos\theta
\end{equation}
 and $I:=[0,+\infty]\times[0,\pi/2]\times[0,2\pi]$. We define the operator 
 \begin{equation}\label{def:opK}
  \text{K}[\Psi](Y,V,t)=\int_{I}dv\,d\theta\,d\varphi\;k(V,v,\theta)\Psi(Y+l(V,v,\theta,\varphi),V+v,t)
  \end{equation}
  and 
   \begin{equation}\label{def:oplambda}
  \displaystyle \lambda(V)=\int_{I}dv\,d\theta\,d\varphi\;k(V,v,\theta).
  \end{equation}
 Then, equation \eqref{eq:BKScaled} can be rewritten as
\begin{equation}\label{eq:scaled_eq}
\begin{split}
\partial_t\Psi(Y,V,t)%&=\alpha\int_{I}dv\,d\theta\,d\varphi\;k(V,v,\theta)\Big[\Psi(Y+l(V,v,\theta,\varphi),V+v,t)-\Psi(Y,V,t) \Big]\\
&=U \left(\text{K}[\Psi]\right)(Y,V,t)-U\lambda(V)\Psi(Y,V,t).
\end{split}
\end{equation}
Using Duhamel's formula we get
\begin{equation}\label{eq:intBKS}
\Psi(Y,V,t)=\Psi_0(Y,V)\,e^{-U\lambda(V)t}+ U \int_0^t e^{-U\lambda(V)(t-t_1)}\text{K}[\Psi](Y,V,t_1)\, dt_1.
\end{equation}
By iterating the equation above, we obtain the series expansion solution
\begin{align}\label{eq:sol}
\nonumber
\Psi(Y,V,t)=&\,\Psi_0(Y,V)\,e^{-U\lambda(V)t}+\sum_{n>0}U^n\int_{0}^{t}dt_n\int_{0}^{t_{n}}dt_{n-1}\dots\int_{0}^{t_{2}}dt_1 \\&
 \times \left[e^{-U\lambda(\cdot)(t-t_n)}\text{K}\,e^{-U\lambda(\cdot)(t_n-t_{n-1})}\text{K}
\dots \text{K}\,e^{-U\lambda(\cdot)(t_1)}\Psi_0(\cdot)\right](Y,V).\\\nonumber
\end{align}

In order to formulate the previous results we define the following space
\begin{equation}\label{def:X_M}
X_M:=C_b(\R^3\times \R^{+})\cap \{\supp \Psi_0(Y,V)\subset \R^3\times [0,M]\}
\end{equation}
for some $0<M<\infty.$ 
\begin{proposition}
Let be $G\in \mathcal{P}( [0,+\infty))$ such that
\begin{equation}\label{ass:Gsolpsi}
\int_{0}^{\infty} dv\,G(v)v^{\frac 2 3}<\infty.
\end{equation}
%Suppose that $\supp \Psi_0(Y,V)\subset [0,M]$ for some $0<M<\infty.$
Then
\begin{itemize}
\item[i)] 
$\text{K}:X_M\to X_M$ defines a bounded operator
and satisfies 
\begin{equation}
\|K[\Psi]\|_{L^\infty(\R^3\times \R^{+})}\leq C(1+M^{\frac 2 3})\|\Psi\|_{L^\infty(\R^3\times \R^{+})}%\quad \quad \forall \Psi
\end{equation}
for any $\Psi\in X_M$ and $C>0$ is a numerical constant depending on the integral in \eqref{ass:Gsolpsi}.
%\end{equation}
\item[ii)] $\lambda\in C([0,\infty))$ and $ 0\leq \lambda(V)\leq C(1+V^{\frac 2 3})$, with $C>0$ a numerical constant depending on the integral in \eqref{ass:Gsolpsi}.
\item[iii)] Equation \eqref{eq:sol} gives the representation formula for the unique solution $\Psi$ of \eqref{eq:BKScaled} in $C^1([0,\infty); X_M)$. %$C^1([0,\infty); C_b(\R^3\times \R^{+}))$.
\end{itemize}
\end{proposition}

\begin{proof}
To prove item i) we first notice that the operator $K$ maps $X_M$ to the set of functions supported in $\{V\leq M\}$ since for any $\Psi\in X_M$ if $V\geq M$ then $\Psi(Y+l,V+v,t)=0$ for any $v\geq 0$.  We now collect some continuity properties of the function $k$ defined by \eqref{eq:transkernel} and the function $l$ defined by  \eqref{def:oplingain} and
observe that $V\to k(V,v,\theta)$ is a continuous function for any $v$ and any $\theta$. In order to check the continuity of $V\to K[\Psi](Y,V,t)$ at $V=0$ we fix $\eta>0$ but sufficiently small and consider two cases: when $v\leq \eta$ and $v\geq \eta.$ When $v\leq \eta$ we use the fact that $|l|$, with $l$ defined by  \eqref{def:oplingain} is bounded by 
$|l(V,v,\theta,\varphi)|\leq (V^\frac 1 3+ \eta^{\frac 1 3})$ with $ \eta$ sufficiently small. In the other case 
$|l|$ is uniformly continuous and $|l(V_1,v,\theta,\varphi)-l(V_2,v,\theta,\varphi)|$ results to be small if the difference $|V_1-V_2|$ is small. Then the continuity of $K[\Psi]$ follows.

Item ii) follows from the definition of $\lambda$ in \eqref{def:oplambda} using the Assumption \eqref{ass:Gsolpsi}.
To prove iii) we notice that the %Duhamel 
series \eqref{eq:sol} converges for any $\Psi_0\in X_M$ thanks to item i) and satisfies \eqref{eq:intBKS}. Therefore it satisfies \eqref{eq:BKScaled}. Moreover, uniqueness follows by  
means of a standard contractivity argument for short times.
\end{proof}

\subsection{Rigorous derivation: from particle system to kinetic equation}\label{sec:derivation}
\subsubsection{Convergence to the solution of the adjoint equation \eqref{eq:BKScaled}}

We are interested in the time evolution of 
\begin{equation}\label{def:testpsi}
\Psi_\phi(Y,V,t):=S_\phi(t)\Psi_0(Y,V)=\EE_{\mu_{\phi}}[\Psi_0(\Pi\big[T^{t}(Y,V;\omega)\big])],
\end{equation}
where the initial datum $\Psi_0$ is a sufficiently smooth function and we recall that $\Pi $ is the projector operator, defined in Section \ref{ssec:PM}, such that $\Pi [T_\phi^t(Y_0,V_0;\omega_0)]=(Y(t),V(t))$. Note that $\Psi_\phi$ is well defined thanks to Proposition \ref{prop:wellpos1}. Indeed, we observe that, differently to the classical Lorentz Gas, the particle dynamics introduced in Section \ref{ssec:PM} is not time reversible. %This implies that we cannot use the backward flow to define the dynamics but we shall use the forward flow f
%or this purpose. Moreover, to have a good definition of observable we look at the evolution of the test function $\Psi$ instead of the probability density $f$. 
Due to this, instead of using the forward Kolmogorov equation for the probability density $f(Y,V,t)$ it is more convenient to consider the backward Kolmogorov equation satisfied by a test function $\Psi(Y,V,t)$.

In this section we will prove the following Theorem.
\begin{theorem}\label{th:kineticpsi}
Assume that $\Psi_0\in X_M$ with $X_M$ defined by \ref{def:X_M}. Let $T>0$ and $\Psi_\phi(Y,V,t)$ be defined in \eqref{def:testpsi}. Then
\begin{equation}
\lim_{\phi\to 0}\Psi_\phi(Y,V,t) = \Psi(Y,V,t) \quad \text{in}\quad C([0,T]; C_{b}(\R^3\times [0,\infty)))
\end{equation}
where $\Psi(Y,V,t)$ solves \eqref{eq:BKScaled}. %The convergence is in $C(\R^3\times\R^{+}\times \R)$.
\end{theorem}

 \begin{proof}

We introduce
$$\chi_{\text{good}}(\omega):= \chi(\{ \omega=\{(\mbf{y}_{k},\mbf{w}_{k})\}_{k\in\N} \in \Omega^1 \}) $$%\chi(\{\mbf{y}_{N}\in \B^N\; \text{s.t.}\;\forall k=1,\dots,N\quad |y_k-X_0|>V_0^{\frac 1 3}+v_k^{\frac 1 3}\})$$
where $\chi(\{\cdot\})$ denotes the characteristic function of the set $\{\cdot\}$ and $\Omega^1={\Omega}^1(\ep_0,\phi,T,Y,V)$ is the set of good configurations defined %by $\Omega^1=\mathcal{U}_{[\frac{T}{\ep_0}]}$ 
in Section \ref{ssec:goodconf}. Hence 
$$\Psi_\phi(Y,V,t)=\EE_{\mu_{\phi}}[\chi_{\text{good}}(\omega)\Psi_0(\Pi\big[T^{t}(Y,V;\omega)\big])]+\EE_{\mu_{\phi}}[\chi_{\text{bad}}(\omega)\Psi_0(\Pi\big[T^{t}(Y,V;\omega)\big])]$$
with 
\begin{equation}
\begin{split}\label{eq:error1}
\EE_{\mu_{\phi}}[\chi_{\text{bad}}(\omega)\Psi_0(\Pi\big[T^{t}(Y,V;\omega)\big])]&=\int_{\Omega}\chi_{\text{bad}}(\omega)\Psi_0(\Pi\big[T^{t}(Y,V;\omega)\big])d\mu(\omega)\\&
\leq \|\Psi_0\|_{\infty}\int_{\Omega^2}d\mu(\omega)\\&\leq\|\Psi_0\|_{\infty}\delta,
\end{split}
\end{equation}
where $\Omega^2={\Omega}^2(\ep_0,\phi,T,Y,V)$ is defined in Section \ref{ssec:goodconf} and $\delta>0$ can be chosen arbitrary small if $\ep_0$ and $\phi$ are small as in Theorem \ref{th:goodconf}.  Consequently, we can restrict to $\EE_{\mu_{\phi}}[\chi_{\text{good}}(\omega)\Psi_0(\Pi\big[T^{t}(Y,V;\omega)\big])]$.

We note that $\chi_{\text{good}}=1$ if the particle belongs to the set of good configurations. 
In particular if $\omega$ is in the set of the good configurations $\Omega^1$ it follows from Lemma \ref{lem:unionDj} that a scatterer outside any ball $\B$ containing $\mathcal{D}_{*}$, defined by \eqref{def:collcyl}, does not influence the dynamics of the tagged particle. Then we can write explicitly the expectation value in the equation above as
\begin{equation}\label{eq:psi_g1}
\begin{split}
\EE_{\mu_{\phi}}[\chi_{\text{good}}(\omega)\Psi_0(\Pi\big[T^{t}(Y,V;\omega)\big])]=&\,e^{-|\B|}\sum_{N\geq 0}\frac{1}{N!}\int_{\B^N}d\mbf{y}_{N}\,\int_{I^N} d\mbf{w}_{N}G(\mbf{w}_N) \chi_{\text{good}}(\omega)\\&
\times \chi(\{\omega\in\Omega\,:\, \#\omega_{|_{\B}}=n\})\,\Psi_0(\Pi\big[T^{t}(Y,V;\omega)\big]),
\end{split}
\end{equation}
where $\B:=\B(0,R)$ for $R>0$, %\B(X_0,M(\ep_0))$ with $M(\ep_0)<\infty$,
$I:=[0,\infty)$ and 
 $\omega_{|_{\B}}=\{(\mbf{y}_{N},\mbf{w}_{N})\}$ with $\mbf{y}_{N}=y_1,\dots ,y_N$, $\mbf{w}_{N}=w_1,\dots,w_N$ is the set of configurations of obstacles in $\B$. Moreover, in \eqref{eq:psi_g1}, due to the characteristic function $\chi(\{\omega\in\Omega\,:\, \#\omega_{|_{\B}}=n\})$ which allows to consider exactly $n$ obstacles, we have
$$T^{t}(Y,V;\omega)=T^{t}(Y,V;\omega_{|_{\B}})\quad \forall \omega\in D_{n}(\B):=\{\omega\in\Omega\,:\, \#\omega_{|_{\B}}=n\}\subset\Omega$$ 
and $\chi_{\text{good}}(\omega)= \chi_{\text{good}}(\{(\mbf{y}_{N},\mbf{w}_{N})\})$, $\forall \omega\in D_{n}(\B)$ . Therefore \eqref{eq:psi_g1} reads as 
\begin{equation}\label{eq:psi_g2}
\begin{split}
&\EE_{\mu_{\phi}}[\chi_{\text{good}}(\omega)\Psi_0(\Pi\big[T^{t}(Y,V;\omega)\big])]\\&
=\,e^{-|\B|}\sum_{N\geq 0}\frac{1}{N!}\int_{\B^N}\!\!d\mbf{y}_{N}\!\int_{I^N} \!\!d\mbf{w}_{N}G(\mbf{w}_N)\chi_{\text{good}}(\omega_{|_{\B}}) \Psi_0(\Pi\big[T^{t}(Y,V;\omega_{|_{\B}})\big]) \\&% \chi_{\text{good}}(\{(\mbf{y}_{N},\mbf{w}_{N})\})\,\Psi_0(T^{t}(Y,V;\omega_{|_{\B}})),
=e^{-|\B|}\sum_{N\geq 0}\frac{1}{N!}\int_{\B^N}\!\!d\mbf{y}_{N}\!\int_{I^N} \!\!d\mbf{w}_{N}G(\mbf{w}_N)(\chi_{\text{good}}(\omega_{|_{\B}})-1)\Psi_0(\Pi\big[T^{t}(Y,V;\omega_{|_{\B}})\big])\\&
\quad +e^{-|\B|}\sum_{N\geq 0}\frac{1}{N!}\int_{\B^N}\!\!d\mbf{y}_{N}\!\int_{I^N} \!\!d\mbf{w}_{N}G(\mbf{w}_N)\Psi_0(\Pi\big[T^{t}(Y,V;\omega_{|_{\B}})\big]).
\end{split}
\end{equation}
We look at the first sum in the r.h.s. of \eqref{eq:psi_g2} and we get 
\begin{equation}\label{eq:error_psi_g2}
\begin{split}
%\EE_{\mu_{\phi}}[(1-\chi_{\text{good}})(\omega)\Psi_0(T^{t}(Y,V;\omega))]=
&e^{-|\B|}\sum_{N\geq 0}\frac{1}{N!}\int_{\B^N}\!\!d\mbf{y}_{N}\!\int_{I^N} \!\!d\mbf{w}_{N}G(\mbf{w}_N)(\chi_{\text{good}}(\omega_{|_{\B}})-1)\Psi_0(\Pi\big[T^{t}(Y,V;\omega_{|_{\B}})\big])\\&
\leq \|\Psi_0\|_{\infty}\,e^{-|\B|}\sum_{N\geq 0}\frac{|\B|^N}{N!}\frac{1}{|\B|^N}\int_{\B^N}\!\!d\mbf{y}_{N}\!\int_{I^N} \!\!d\mbf{w}_{N}G(\mbf{w}_N)(\chi_{\text{good}}(\omega_{|_{\B}})-1)\\&
\leq  \|\Psi_0\|_{\infty}\,\sum_{N\geq 0}\PP(\omega\in\Omega^2\,|\, \omega\in D_{n}(\B))\PP(\omega\in D_{n}(\B))\\&
\leq \|\Psi_0\|_{\infty}\PP(\omega\in\Omega^2)\leq \|\Psi_0\|_{\infty} \delta,
\end{split}
\end{equation}
where, in the third step, we used the law of total probability. It follows that for $\delta>0$ arbitrary small as in Theorem \ref{th:goodconf} the right hand side of \eqref{eq:error_psi_g2} tends to zero in the limit $\ep_0\to 0$. Therefore 
\begin{equation}
\begin{split}
\left\vert\EE_{\mu_{\phi}}[\Psi_0(\Pi\big[T^{t}(Y,V;\omega)\big])]-\bar{\Psi}_\phi(Y,V,t)\right\vert  &
\leq 2\|\Psi_0\|_{\infty}\delta 
\end{split}
\end{equation}
where 
$$\bar{\Psi}_\phi(Y,V,t):=e^{-|\B|}\sum_{N\geq 0}\frac{1}{N!}\int_{\B^N}\!\!d\mbf{y}_{N}\!\int_{I^N} \!\!d\mbf{w}_{N}G(\mbf{w}_N)\Psi_0(\Pi\big[T^{t}(Y,V;\omega_{|_{\B}})\big]).$$
Therefore we just need to show that $\bar{\Psi}_\phi$ converges, in the limit $\phi\to 0$, to the solution of \eqref{eq:BKScaled}.

We now distinguish the obstacles of the configuration $\{(\mbf{y}_{N},\mbf{v}_{N})\}$ 
which influence the motion of the tagged particle, called \textit{internal} obstacles, and the \textit{external} ones. More precisely $y_k$ is internal if
\begin{equation}
\inf_{0\leq s\leq t} \left\{ |\Pi_1\big[T^{s}(Y,V;\omega_{|_{\B}})\big]-(y_k-Us)| \leq V(s;\omega_{|_{\B}})^{\frac 1 3}+w_k ^{\frac 1 3}\right\}
\end{equation}
while $y_k$ is external if 
\begin{equation}
\inf_{0\leq s\leq t} \left\{ |\Pi_1\big[T^{s}(Y,V;\omega_{|_{\B}})\big]-(y_k-Us)| > V(s;\omega_{|_{\B}})^{\frac 1 3}+w_k ^{\frac 1 3}\right\}.
\end{equation}
Here  $\Pi\big[T^{t}(Y,V;\omega_{|_{\B}})\big]=(Y^{t}(Y_0,V_0;\omega_{|_{\B}}),V^{t}(Y_0,V_0;\omega_{|_{\B}}))$ is the projector of the flow defined in Section \ref{ssec:PM} and $\Pi_1\big[T^{t}(Y,V;\omega_{|_{\B}})\big]$ denotes the first component. 
A given configuration $\{(\mbf{y}_{N},\mbf{w}_{N})\}$ of $\B^N\times I^N$ can be decomposed as 
$\{(\mbf{y}_{N},\mbf{w}_{N})\}=\{(\mbf{x}_{n},\mbf{v}_{n})\}\cup \{(\mbf{z}_{m},\mbf{u}_{m})\} $
where $ \{(\mbf{z}_{m},\mbf{u}_{m})\} $ is the set of all external obstacles and $\{(\mbf{x}_{n},\mbf{v}_{n})\}$ is the set of all internal ones.

The integration over the external obstacles can be performed explicitly. In fact, we have
\begin{equation}
\begin{split}\label{eq:intext}
\bar{\Psi}_\phi(Y,V,t)&=e^{-|\B|}\sum_{n\geq 0}\frac{1}{n!}\int_{\B^n}d\mbf{y}_{n}\,\int_{I^n} d\mbf{v}_{n}\,G(\mbf{v}_n)\,\sum_{m\geq 0} \frac{1}{m!}\int_{\B^m}d\mbf{y}_{m}\,\int_{I^m} d\mbf{u}_{m}\,G(\mbf{u}_m)\\&
\quad \times \chi(\{\text{the}\;\{(\mbf{x}_{n},\mbf{v}_{n})\}\;\text{are}\;\text{internal}\;\text{and}\; \text{the}\;\{(\mbf{z}_{m},\mbf{u}_{m})\}\;\text{are}\;\text{external}\})\\&
\quad \times \Psi_0(\Pi\big[T^{t}(Y,V;\omega_{|_{\B}})\big]).
\end{split}
\end{equation}

We characterize the evolution of the tagged particle by means of the following sequence 
\begin{equation}\label{def:evtagpar}
(Y_0,V_0),(Y_1^{-},V_0),(Y_1^{+},V_1), (Y_2^{-},V_1), (Y_2^{+},V_2), \dots, (Y_k^{-},V_{k-1}), (Y_{k}^{+},V_{k}) ,\dots
\end{equation}
where we denote by $Y_{k}^{-}\in \R^3$ the position of the tagged particle after the free flight before the coalescence and by $Y_{k}^{+}\in \R^3$ the position of the tagged particle immediately after the coalescence. Notice that the values of  $(Y_k^{-},Y_{k}^{+},V_{k})$ are computed according to the evolution for the CTP model given in Section \ref{ssec:PM}.
We will denote the cylinder of radius $R$ and height $l$ by
 \begin{equation}\label{def:cyltube}
 C(R,l):=\{x\in\R^3: \,0\leq x_1\leq l, \; 0\leq (x_2^2+x_3^2)^{\frac{1}{2}}\leq R\}.
\end{equation}
We then define the tube 
$\T_t(\{(\mbf{x}_{n},\mbf{v}_{n})\})=\T_t(\{(\mbf{x}_{n},\mbf{v}_{n})\}; (Y;V)) \subseteq\R^3$ up to time $t$, of varying width around the trajectory of the tagged particle (in the space of positions). This is obtained moving the tagged particle with speed $U$ in the horizontal component and coalescing it with any colliding obstacle as in \eqref{def:merg_dyn} up to time $t$. More precisely
\begin{equation}\label{def:uniontubes}
\T_t(\{(\mbf{x}_{n},\mbf{v}_{n})\})=\bigcup_{j=0}^{n-1}\left[{ Y_{j}^{+}+C((\frac{3}{4\pi}V_j)^{\frac{1}{3}}, (Y_{j+1}^{-}-Y_{j}^{+})\cdot e_1)}\right]\cup \left[{ Y_{n}^{+}+C((\frac{3}{4\pi}V_n)^{\frac{1}{3}}, (Y(t)-Y_{n}^{+})\cdot e_1)}\right].
\end{equation}
Here we used the convention that $Y_0^{+}=Y_0$. We recall that $n$ is the number of collisions up to the time $t$.
 Notice that
\begin{equation}\label{def:tube}
\{(\mbf{z}_{m},\mbf{u}_{m})\} \; \text{are}\;\text{external}\;\,\text{iff} \;\, \text{dist}(\mbf{z}_{m},\T_t(\{(\mbf{x}_{n},\mbf{v}_{n})\}))> \left(\frac{3\mbf{u}_{m}}{4\pi}\right)^{\frac 1 3}.
\end{equation}
Thus, from \eqref{eq:intext} we obtain that
\begin{equation}
\begin{split}\label{eq:int1}
&e^{-|\B|}\sum_{n\geq 0}\frac{1}{n!}\int_{\B^n}d\mbf{x}_{n}\,\int_{I^n} d\mbf{v}_{n}\,G(\mbf{v}_n)\, \chi(\{\{(\mbf{x}_{n},\mbf{v}_{n})\}\;\text{int.}\})\,  \Psi_0(\Pi\big[T^{t}(Y,V;\omega_{|_{\B}})\big]) \\&
\quad\times \sum_{m\geq 0} \frac{1}{m!}\int_{I^m} d\mbf{u}_{m}\,G(\mbf{u}_m) \int_{\B^m} \!d\mbf{z}_{m} \chi(\{\mbf{z}_{m}\,\text{s.t.}\,\text{dist}(\mbf{z}_{m},\T_t(\{(\mbf{x}_{n},\mbf{v}_{n})\}))> \left(\frac{3\mbf{u}_{m}}{4\pi}\right)^{\frac 1 3}\})
\end{split}
\end{equation}
and for any $m$ %due to the fact that the obstacles are independently distributed in disjoint sets,
the characteristic function factorizes as
$$\chi(\{\mbf{z}_{m}\,\text{s.t.}\,\text{dist}(\mbf{z}_{m},\T_t(\{(\mbf{x}_{n},\mbf{v}_{n})\}))> \mbf{u}_{m}^{1/3}\})=\prod_{k=1}^{m}\chi(\{z_{k}\,\text{s.t.}\,\text{dist}(z_{k},\T_t(\{(\mbf{x}_{n},\mbf{v}_{n})\}))> \left(\frac{3u_k}{4\pi}\right)^{\frac 1 3}\})$$
so that \eqref{eq:int1} becomes
\begin{equation}
\begin{split}\label{eq:int2}
&e^{-|\B|}\sum_{n\geq 0}\frac{1}{n!}\int_{\B^n}d\mbf{x}_{n}\,\int_{I^n} d\mbf{v}_{n}\,G(\mbf{v}_n)\, \chi(\{\{(\mbf{x}_{n},\mbf{v}_{n})\}\;\text{int.}\})\, \Psi_0(\Pi\big[T^{t}(Y,V;\omega_{|_{\B}})\big])\\&
\quad\times \sum_{m\geq 0} \frac{1}{m!}\int_{I^m} d\mbf{u}_{m}\,G(\mbf{u}_m) \int_{\B^m} d\mbf{z}_{m}\prod_{k=1}^{m}\chi(\{z_{k}\,\text{s.t.}\,\text{dist}(z_{k},\T_t(\{(\mbf{x}_{n},\mbf{v}_{n})\}))> \left(\frac{3u_k}{4\pi}\right)^{\frac 1 3}\})\\&
=e^{-|\B|}\sum_{n\geq 0}\frac{1}{n!}\int_{\B^n}d\mbf{x}_{n}\,\int_{I^n} d\mbf{v}_{n}\,G(\mbf{v}_n)\, \chi(\{\{(\mbf{x}_{n},\mbf{v}_{n})\}\;\text{int.}\})\, \Psi_0(\Pi\big[T^{t}(Y,V;\omega_{|_{\B}})\big])\\&
 \quad\times \sum_{m\geq 0} \frac{1}{m!}\int_{I^m} d\mbf{u}_{m}\,G(\mbf{u}_m) \prod_{k=1}^{m} | \{z_{k}\,\text{s.t.}\,\text{dist}(z_{k},\T_t(\{(\mbf{x}_{n},\mbf{v}_{n})\}))> \left(\frac{3u_k}{4\pi}\right)^{\frac 1 3}\} | \\&
 =e^{-|\B|}\sum_{n\geq 0}\frac{1}{n!}\int_{\B^n}d\mbf{x}_{n}\,\int_{I^n} d\mbf{v}_{n}\,G(\mbf{v}_n)\, \chi(\{\{(\mbf{x}_{n},\mbf{v}_{n})\}\;\text{int.}\})\, \Psi_0(\Pi\big[T^{t}(Y,V;\omega_{|_{\B}})\big])\\&
 \quad\times \sum_{m\geq 0} \frac{1}{m!} \left(\int_{0}^{\infty}\!du\,G(u)| \{z\;\text{s.t.}\,\text{dist}(z,\T_t(\{(\mbf{x}_{n},\mbf{v}_{n})\}))> \left(\frac{3u}{4\pi}\right)^{\frac 1 3}\} |\right)^m \\&
 %e^{-|\B|}e^{\int_{0}^{\infty}du\,g(u)[|\B|-(|\T(\{(\mbf{x}_{n},\mbf{v}_{n})\})|+|B_{u}^{1/3}}(0)|)]}}\\&
 =\sum_{n\geq 0}\frac{1}{n!}\int_{\B^n}d\mbf{x}_{n}\,\int_{I^n} d\mbf{v}_{n}\,G(\mbf{v}_n)\, e^{-\int_{0}^{\infty}\!du\,G(u)| \{z\;\text{s.t.}\,\text{dist}(z,\T_t(\{(\mbf{x}_{n},\mbf{v}_{n})\}))\leq \left(\frac{3u}{4\pi}\right)^{\frac 1 3}\} |}\\&
 \quad \times \chi(\{\{(\mbf{x}_{n},\mbf{v}_{n})\}\;\text{int.}\})\, \Psi_0(\Pi\big[T^{t}(Y,V;\omega_{|_{\B}})\big]).
\end{split}
\end{equation}
Note that in the last step we used that $\int_{0}^{\infty}du\,G(u)=1$.  We order the obstacles according to the collision sequence, i.e. $x_i$ is collided before $x_j$ if $i<j$.
Therefore we can rewrite the decomposition for the tube \eqref{def:uniontubes} as 
$$\T_t(\{(\mbf{x}_{n},\mbf{v}_{n}\})=\bigcup_{i=0}^{n} \T_{[t_i,t_{i+1}]}(\{(\mbf{x}_{n},\mbf{v}_{n})\})$$
where $ \T_{[t_i,t_{i+1}]}(\{(\mbf{x}_{n},\mbf{v}_{n})\})$ is the tube in the time interval $[t_i,t_{i+1}]$ for all $i=1,\dots,n$, namely 
\begin{equation}
 \T_{[t_i,t_{i+1}]}(\{(\mbf{x}_{n},\mbf{v}_{n})\})=\left[{ Y_{i}^{+}+C((\frac{3}{4\pi}V_i)^{\frac{1}{3}}, (Y_{i+1}^{-}-Y_{i}^{+})\cdot e_1)}\right]
\end{equation}
with the convention that $t_0=0$ and $t_{n+1}=t$.
We consider 
\begin{equation}\label{eq:appr_tube}
\mathcal{E}:= | \{z\;\text{s.t.}\,\text{dist}(z,\T_t(\{(\mbf{x}_{n},\mbf{v}_{n})\}))\leq \left(\frac{3u}{4\pi}\right)^{1/3}\}| - \sum_{i=0}^{n}\vert \T^{(u)}_{[t_i,t_{i+1}]}(\{(\mbf{x}_{n},\mbf{v}_{n})\}) \vert   
\end{equation}
where 
$$\T^{(u)}_{[t_i,t_{i+1}]}(\{(\mbf{x}_{n},\mbf{v}_{n})\})=\left[{ Y_{i}^{+}+C( \big(\frac{3}{4\pi}\big)^{1/3}(u^{\frac{1}{3}}+V_i^{\frac{1}{3}}), (Y_{i+1}^{-}-Y_{i}^{+})\cdot e_1)}\right].$$
We then have that 
\begin{equation}\label{eq:errortube}
\vert \mathcal{E} \vert \leq Cn\,\phi ((V(t))^{\frac 2 3}u^{\frac 1 3} +u).
\end{equation}
Moreover 
\begin{equation*}
 \displaystyle  e^{-\int_{0}^{\infty}\!du\,G(u)| \{z\;\text{s.t.}\,\text{dist}(z,\T(\{(\mbf{x}_{n},\mbf{v}_{n})\}))\leq\left( \frac{3u}{4\pi}\right)^{1/3}\} |} =e^{-\sum_{i=0}^{n}\int_{0}^{\infty}\!du\,G(u) | \T^{(u)}_{[t_i,t_{i+1}]}(\{(\mbf{x}_{n},\mbf{v}_{n}\})|} e^{-\int_{0}^{\infty}\!du\,G(u)\mathcal{E}(u)}.
 \end{equation*}
 Then we get
 \begin{equation}
\begin{split}\label{eq:barpsi2}
\bar\Psi_\phi(Y,V,t)&=\sum_{n\geq 0}\frac{1}{n!}\int_{\B^n}d\mbf{x}_{n}\,\int_{I^n} d\mbf{v}_{n}\,G(\mbf{v}_n)\, e^{-\sum_{i=0}^{n}\int_{0}^{\infty}\!du\,G(u) | \T^{(u)}_{[t_i,t_{i+1}]}(\{(\mbf{x}_{n},\mbf{v}_{n}\})|}e^{-\int_{0}^{\infty}\!du\,G(u)\mathcal{E}(u)}  \\&
 \quad\times\chi(\{\{(\mbf{x}_{n},\mbf{v}_{n})\}\;\text{int.}\})\, \Psi_0(\Pi\big[T^{t}(Y,V;\omega_{|_{\B}})\big]). %+\mathcal{E}(\phi),
\end{split}
\end{equation}

Following the original idea of Gallavotti (see \cite{G}) we perform the following change of variables
\begin{align}\label{change var}
& 0\leq t_1<t_2<\dots<t_n\leq t\\\nonumber
& x_1,\dots,x_n\longrightarrow t_1,\beta_1,\dots, t_n,\beta_n 
\end{align}
where $\beta_i=\beta_i(\theta_i,\varphi_I)$ and $t_i$ are the ``collision parameters" and the entrance time of the tagged particle in the protection disk around $x_i$. Therefore we can construct a trajectory $S$ (for the tagged particle) in each time interval $[t_{i},t_{i+1}]$ for $i=0,1,\dots,n$ as we did in Section \ref{ssec:PM} (see items 1) and 2) ). More precisely we notice that the trajectory can be obtained as follows.
For $s\in(t_{i}^{+},t_{i+1}^{-})$ we have the free motion %described by  \eqref{def:flow1}
\begin{equation*}
\Pi\big[S^{s}(Y_0,V_0;\{t_j,\beta_j,v_j\})\big] =\Big(Y_0,V_0\Big)
  \quad \text{if}\quad \inf_{k\in\N}\{|{x}_k-Us-Y_0|-(V^{\frac 1 3}+v_k^{\frac 1 3})\}>0
\end{equation*}
 with the ``jump condition", at time $t_{i+1}$,
 \begin{equation}
 \begin{split}\label{def:flow2}
& \Pi\big[S^{t_{i+1}^{+}}(Y_0,V_0;\{t_j,\beta_j,v_j\})\big]=\left(Y^{t_{i+1}^{-}}(Y_0,V_0)+\frac{v_i (V_i^{\frac 1 3}+v_i^{\frac 1 3})n(\theta,\varphi)}{V^{t_{i+1}^{-}}(Y_0,V_0)+v_i},V^{t_{i+1}^{-}}(Y_0,V_0)+v_i\right) \quad\\
& \text{if}\qquad \quad  \inf_{k\in\N}\{|x_k-Ut_{i+1}^{-}-Y_0|-(V^{\frac 1 3}+v_k^{\frac 1 3})\}=|x_i-Ut_{i+1}^{-} -Y_0|-(V^{\frac 1 3}+v_i^{\frac 1 3})=0\\
& \text{and}\quad \inf_{k\in\N\setminus \{i\}}\{|x_k-Ut_{i+1}^{-}-Y_0|-(V^{\frac 1 3}+v_k^{\frac 1 3})\}>0.
 \end{split}
  \end{equation}
 Note that here we used the reference frame in which the obstacles are moving and the tagged particle is fixed as we already discussed in Section \ref{ssec:PM}.

We remark that, since we restricted our analysis to the set of good configurations of obstacles, we do not need the definition of the flow with multiple collisions because we are not considering them and $S^{t}(Y_0,V_0;\{t_j,\beta_j,v_j\})\equiv T^{t}(Y_0,V_0;\omega_{|_{\B}})$ for $\omega\in\Omega^{1}$.
Moreover, due to the fact that we are outside the set of bad configurations of obstacles we observe that the map \eqref{change var} is one-to-one and 
$$\frac{dx_1\dots dx_n}{n!}=U^n \chi\big(0\leq t_1<t_2<\dots<t_n\leq t\big)dt_1 d\beta_1\, \bar k(V_1,v_1,\theta_1)\dots dt_n  d\beta_n \,\bar k(V_n,v_n,\theta_n)$$%\delta\left(\sum_{i=0}^n t_i=t\right)dt_1 d\beta_1 k(V,v_1,\theta_1)\dots dt_n  d\beta_n k(V,v_n,\theta_n)$$
where $\bar k(V_i,v_i,\theta_i)=(V_i^{\frac 1 3}+v_i^{\frac 1 3})^2 \sin\theta_i \cos\theta_i$ and $d\beta_i=d\theta_i d\varphi_i$.

By performing the change of variables described above, equation \eqref{eq:barpsi2} can be written as
\begin{equation}
\begin{split}\label{eq:change_vr}
\bar{\Psi}_\phi(Y,V,t)&=\sum_{n\geq 0}\,  U^n\int_{0}^{t} dt_{n} \int d\beta_n dv_n e^{-\int_{0}^{\infty}{du\,G(u)(|\T^{(u)}_{[t_n,t]}(\{(\mbf{x}_{n},\mbf{v}_{n})\})|)}} k(V_n,v_n,\theta_n)  \dots\\&
\quad \times \int_{0}^{t_2} dt_1 \int d\beta_1 dv_1 e^{-\int_{0}^{\infty}{du\,G(u)(|\T^{(u)}_{[0,t_1]}(\{(\mbf{x}_{n},\mbf{v}_{n})\})|)}}e^{-\int_{0}^{\infty}\!du\,G(u)\mathcal{E}(u)}k(V_1,v_1,\theta_1)\\&
\quad \times \Psi_0(Y_1,V_1),
\end{split}
\end{equation}
where $k(V_i,v_i,\theta_i)$ is the transition kernel defined by \eqref{eq:transkernel}, namely $k(V_i,v_i,\theta_i)=G(v_i)(V_i^{\frac 1 3}+v_i^{\frac 1 3})^2 \sin\theta_i \cos\theta_i$.
Moreover, for $i=0,1,\dots,n,$ we have 
\begin{equation*}
\begin{split}
\int_{0}^{\infty}{du\,G(u)(|\T^{(u)}_{[t_i,t_{i+1}]}(\{(\mbf{x}_{n},\mbf{v}_{n})\})|)}&=\int_{0}^{\infty}{du\,G(u)[\pi(V_i^{\frac 1 3}+u^{\frac 1 3})^2U(t_{i+1}-t_{i})]}\\&=U\lambda(V_i)(t_{i+1}-t_{i}). %\quad  i=1,\dots,n,$$ 
\end{split}
\end{equation*}
Then
\begin{equation}
\begin{split}\label{eq:change_vr2}
\bar{\Psi}_\phi(Y,V,t)&=\sum_{n\geq 0}\,U^n \int_{0}^{t} dt_{n} \int d\beta_n dv_n  e^{- U\lambda(V_n)(t-t_{n})}k(V_n,v_n,\theta_n) \dots\\&
\quad  \int_{0}^{t_2} dt_1 \int d\beta_1 dv_1 e^{-U\lambda (V_1) t_1}e^{-\int_{0}^{\infty}\!du\,G(u)\mathcal{E}(u)}k(V_1,v_1,\theta_1)
\Psi_0(Y_1,V_1).%\Psi_0(T^{t}\circ \dots \circ T^{t_1}(Y,V;\omega_{|_{\B}}))
\end{split}
\end{equation}
Using equation \eqref{eq:errortube} and the fact that the initial datum $\Psi_0$ is compactly supported, %(and the first moment is bounded)
Lebesgue's convergence Theorem implies that  the r.h.s. of \eqref{eq:change_vr2}  converges in the limit $\phi\to 0$ to the r.h.s. of \eqref{eq:sol}. Therefore $\lim_{\phi\to 0}\bar{\Psi}_\phi(Y,V,t)$ solves \eqref{eq:scaled_eq} and the result follows.
\end{proof}

\subsubsection{Conclusion of the Proof of Theorem \ref{th:theorem1}}\label{sssec:conclproof}
In the previous section we showed that 
$
\Psi_\phi(Y,V,t)%:=S_\phi(t)\Psi_0(Y,V)
=\EE_{\mu_{\phi}}[\Psi_0(\Pi\big[T^{t}(Y,V;\omega)\big])]
$
converges to $\Psi(Y,V,t)$, the unique solution of \eqref{eq:BKScaled}, as $\Phi\to 0.$ To prove Theorem \ref{th:theorem1} we first define a measure $f$ by means of
\begin{equation}\label{def:fmeasure}
\int_{A}f(Y,V,t)dYdV=\int_{\R^{3}\times \R_{+}}f_0(Y,V)\Psi(Y,V,t)dYdV
\end{equation}
where $\Psi$ solves \eqref{eq:BKScaled} with $\Psi_0=\chi_{A}$ with $A$ compact subset of $ \R^3\times [0,M]$, $0<M<\infty$.

Now we observe that 
thanks to Definition \ref{def:f_phi} 
we have that
\begin{equation}\begin{split}\label{eq:exp_chiA}
\int_{A} f_{\phi}(Y,V,t)dYdV&=%\EE_{\mu_{\phi}}[\chi_A(\Pi\big[T^{t}(Y,V;\omega)\big])]=\\&
\int_{\R^3	\times [0,\infty)} f_{0}(Y,V)\EE_{\mu_{\phi}}[\chi_A(\Pi\big[T^{t}(Y,V;\omega)\big])]dYdV\\&= 
\int_{\R^3\times [0,\infty)} f_0(Y,V)\Psi_\phi(Y,V,t)dYdV,
\end{split}
\end{equation}
where we used that $\EE_{\mu_{\phi}}[\chi_A(\Pi\big[T^{t}(Y,V;\omega)\big])]=\mu_\phi(\{\omega\,:\; \Pi[T_\phi^{t}(Y_0,V_0;\omega)]\in A\})$.
Moreover, Proposition \ref{th:kineticpsi} guarantees that 
\begin{equation}\label{eq:proofkinpsi}
\int_{\R^3\times [0,\infty)} f_0(Y,V)\Psi_\phi(Y,V,t)dYdV\underset{\phi\to 0}{\longrightarrow} \int_{\R^3\times [0,\infty)} f_0(Y,V)\Psi(Y,V,t)dYdV %\quad \text{as}\quad \phi\to 0$$
\end{equation}
and, using the definition of $f$ in \eqref{def:fmeasure}, the right hand side of \eqref{eq:proofkinpsi} reads as
\begin{equation}\label{eq:duality_chiA}
\int_{\R^3\times [0,\infty)} f_0(Y,V)\Psi(Y,V,t)dYdV=\int_{\R^3\times [0,\infty)} f(Y,V,t)\Psi_0(Y,V)dYdV=\int_{A} f(Y,V,t)dYdV.
\end{equation}
Therefore, from \eqref{eq:exp_chiA} and \eqref{eq:duality_chiA} we get
$$
\int_{A} f_{\phi}(Y,V,t)dYdV\underset{\phi\to 0}{\longrightarrow}\int_{A} f(Y,V,t)dYdV.
$$
In order to conclude the proof of Theorem \ref{th:theorem1} we need to show that the measure $f$ defined by \eqref{def:fmeasure} solves equation \eqref{eq:FKScaled} in the sense of Definition \ref{def:f_soleq}.
Indeed we note that if $\tilde{\Psi}\in C([0,t^{*}];L^{\infty}(\R^3\times\R^+))$ is the solution of 
\begin{equation}\label{eq:tildepsi}
\partial_t{\tilde{\Psi}}(Y,V,t)+\mathcal{C}[\tilde{\Psi}](Y,V,t)=0 \qquad 0\leq t\leq t^*,
\end{equation}
then we have 
\begin{equation}\label{eq:dualitytildepsi}
 \int_{\R^3	\times [0,\infty)}f(Y,V,t^*){\tilde{\Psi}}(Y,V,t^*)dYdV=\int_{\R^3\times [0,\infty)} f_0(Y,V){\tilde{\Psi}}(Y,V,0)dYdV.
\end{equation}
This follows from \eqref{def:fmeasure} with $\Psi(\cdot,t)={\tilde{\Psi}}(\cdot,t^*-t)$.
We want to check that the measure $f$ defined by \eqref{def:fmeasure} satisfies \eqref{eq:f_soleq}. To do this, given a test function $\Psi$ as in Definition \ref{def:f_soleq} we define $\xi(\cdot,t)$ as 
\begin{equation}\label{eq:tildepsixi}
%\left\{\begin{array}{ll}
\partial_t {\tilde{\Psi}}(Y,V,t)+\mathcal{C}[{\tilde{\Psi}}](Y,V,t)=\xi(Y,V,t). % &\\
\end{equation}
Note that we have the following representation formula for ${\tilde{\Psi}}$, i.e.
${\tilde{\Psi}}(\cdot,t)=-\int_t^T \phi(\cdot,t,s)ds $
where $ \phi(\cdot,t,s)$ solves
\begin{equation*}
\left\{\begin{array}{ll}
\partial_t \phi(\cdot,t,s)+\mathcal{C}[\phi](\cdot,t,s)=0, \quad 0\leq t\leq s &\\
\phi(\cdot,s,s)=\xi(\cdot,s).&
\end{array}\right.
\end{equation*}
This representation formula follows from straightforward computations as well as from the uniqueness of solutions for the problem \eqref{eq:tildepsixi} in $[0,t^{*}]$ with initial data $\tilde{\Psi}(\cdot,t^*)=0$. 
We then compute
\begin{equation*}
\begin{split}
&\int_{0}^T \int_{\R^3	\times [0,\infty)} f(Y,V,t)\{\partial_t {\tilde{\Psi}}(Y,V,t)+\mathcal{C}[{\tilde{\Psi}}](Y,V,t)\} dYdVdt\\&
=\int_{0}^T dt \int_{\R^3	\times [0,\infty)}  f(Y,V,t) \xi(Y,V,t) dYdVdt=\int_{0}^T \int_{\R^3	\times [0,\infty)} f(Y,V,t) \phi(\cdot,t,t)dYdVdt \\& 
=\int_{0}^T \int_{\R^3	\times [0,\infty)} f_0(Y,V) \phi(\cdot,0,t)dYdVdt =-\int_{\R^3\times [0,\infty)} f_0(Y,V) {\tilde{\Psi}}(Y,V,0) dYdV,
\end{split}
\end{equation*}
where we used \eqref{eq:dualitytildepsi} and the representation formula for ${\tilde{\Psi}}$. Then \eqref{eq:f_soleq} follows.

\section{Analysis of the limit kinetic equation}
\subsection{Well-posedness of the kinetic equation}\label{ssec:WPKE}
\begin{theorem}
For any $f_0\in \mathcal{P}(\R^3\times [0,\infty))$ there exists a unique solution of \eqref{eq:FKScaled} in the sense of Definition \ref{def:f_soleq}
\end{theorem}
\begin{proof}
The existence of one solution of the form \eqref{def:fmeasure} has been proved in Section \ref{sssec:conclproof}. In order to prove the uniqueness %follows by duality. More precisely,
we consider test functions as in Definition \ref{def:f_soleq} doing the particular choice $\Psi(Y,V,t)\chi_\ep(t;T)$ with $\chi_\ep\in C^1_c([0,\infty))$ for any $\Psi \in C^1_c(\R^3\times \R^+\times [0,\infty))$ and  $\chi_\ep(t;T)\to \chi_{[0,T]}(t)$ as $\ep\to 0$. Thus 
\begin{equation}
\begin{split}
&\int_{\R^3}  \hspace{-0.1cm}\int_{ [0,\infty)} \hspace{-0.5cm}f(Y,V,T) \Psi(Y,V,T)dYdV-\int_0^T  \hspace{-0.2cm} \int_{\R^3}  \hspace{-0.1cm}\int_{ [0,\infty)}\hspace{-0.5cm}f(Y,V,t)\{\partial_t\Psi(Y,V,t)+ \mathcal{C}[\Psi](Y,V,t)\} dYdVdt \\&
=\int_{\R^3}  \hspace{-0.1cm}\int_{ [0,\infty)}  \hspace{-0.5cm} f_0(Y,V) \Psi_0(Y,V)dYdV.
\end{split}
\end{equation}
Now we solve \eqref{eq:tildepsi} with $\tilde\Psi(Y,V,t^*)=\tilde\Psi_n(Y,V) $ for $t^*=T$ and $\Psi_n\to \chi_A$ with $A$ compact subset of $\R^3\times [0,M]$ and we obtain the representation formula \eqref{def:fmeasure} which gives uniqueness for the measure $f$.
\end{proof}
\subsection{Long time behaviour for the distribution of volumes}\label{ssec:LTB}
\begin{assumption}
We assume that $G(v)$ is such that
$$\int_0^\infty G(v)v^{\frac{5}{3}+\theta} dv<\infty, \quad \theta>0.$$
\end{assumption}
We are interested in the asymptotic behaviour of \eqref{eq:FKScaled}.
More precisely, we consider the solution $f$ of \eqref{eq:FKScaled} and set $F(V,t):=\int_{\R^3}dY f(Y,V,t)$. Notice that $F$ is the average of $f$ with respect to the position variable $Y$ and satisfies
\begin{equation}\label{eq:evF}
\partial_t F(V,t)=\,%\frac{U\phi^{\frac{2}{3}}}{2}\left(\frac{3}{4\pi}\right)^{\frac{2}{3}} 
\lambda \left(\int_0^{V} \hspace{-2.5mm} dv \,G(v)((V-v)^{\frac{1}{3}}+v^{\frac{1}{3}})^2 F(V-v,t)-\int_0^{\infty}\hspace{-2.5mm} dv \,G(v)(V^{\frac{1}{3}}+v^{\frac{1}{3}})^2 F(V,t)\right)
\end{equation}
since $\int_0^{\frac{\pi}{2}}\hspace{-1.5mm}d\theta\int_0^{2\pi}d\varphi\sin\theta \cos\theta=\pi$. Here $\lambda= U\pi \left(\frac{3}{4\pi}\right)^{\frac{2}{3}}$. %$\lambda= \frac{U\phi^{\frac{2}{3}}}{2}\left(\frac{3}{4\pi}\right)^{\frac{2}{3}}$. 
We notice that existence and uniqueness of solutions in the same sense of Definition \ref{def:f_soleq} can be obtained by similar arguments that we skip at this level.

Our goal is to prove the following Theorem.

\begin{theorem}\label{th:asymp}
Let $F\in \mathcal{P}(\R^+)$ solution of \eqref{eq:evF}. We have 
\begin{equation}
F(W t^{3},t)t^3\overset{*}{\rightharpoonup} \delta(W-a)\quad \text{as}\;\; t\to\infty \quad \text{in}\quad \mathcal{M}(\R^+)
\end{equation}
where $a=\frac{\lambda}{27}\big(\int_0^\infty v\,G(v)\,dv\big)^3$.
\end{theorem}
To prove this Theorem we will use the adjoint equation of \eqref{eq:evF} which can be obtained from \eqref{eq:BKScaled} considering test functions independent of the variable $Y$. Thus $\Psi(V,t)$  satisfies
\begin{equation}\label{eq:evPsiaverage}
\partial_t \Psi(V,t)=\,-\lambda \int_0^{\infty}\hspace{-2.5mm}dv \,G(v)(V^{\frac{1}{3}}+v^{\frac{1}{3}})^2\left[ \Psi(V+v,t)- \Psi(V,t)\right].
\end{equation}
Then the following duality formula holds:
\begin{equation}\label{duality_volumes}
\int_{ [0,\infty)} \psi(V,T) F(V,T)\,dV =\int_{ [0,\infty)} \psi(V,0) F(V,0)\, dV.
\end{equation}
We will use the following Proposition.

\begin{proposition}\label{prop:asymp}
Let be $T>0$. Let $\varphi_0\in C_c([0,\infty))$. Let $\Psi_T(V,t)$ be the solution in $0\leq t\leq T$ of \eqref{eq:evPsiaverage} in $C^1([0,\infty);L^{\infty}(\R^+))$ with final value $\Psi_T(V,T)=\varphi_0\Big(\frac{V}{T^{3}}\Big)$.
We then have that $\Psi_{T}(V,0)\to \varphi_0(a)$ as $T\to +\infty$ uniformly on the compact sets of $[0,+\infty)$ and $a$ is as in Theorem \ref{th:asymp}. Moreover $|\Psi_{T}(V,0)|\leq \|\varphi_0\|_{\infty}$.
\end{proposition}

\begin{proofof}[Proof of Proposition \ref{prop:asymp}]
We first prove the result for $\varphi_0\in C^2_c([0,\infty))$.  
We perform the following change of variables $\xi=\frac{V}{T^3}$ and $\tau=\frac{t}{T}$ and set $\phi(\xi,\tau):=\Psi(T^3\xi,T\tau)$. Thus, $\phi(\xi,\tau)$ satisfies 
\begin{equation}\label{eq:backphiT}
\left\{\begin{array}{ll}\vspace{1mm}
\partial_{\tau}\phi(\xi,\tau)=-\lambda T\int_0^\infty dv G(v) (T\xi^{\frac{1}{3}}+v^{\frac{1}{3}})^2[\phi(\xi+\frac{v}{T^3},\tau)-\phi(\xi,\tau)] &\\
\phi(\xi,1)=\varphi_0(\xi).&.
\end{array}\right.
\end{equation}
We now set $s=1-\tau$ in order to obtain the forward equation of \eqref{eq:backphiT} which reads as 
\begin{equation}\label{eq:forwphiT}
\left\{\begin{array}{ll}\vspace{1mm}
\partial_{s}\phi(\xi,s)=\lambda T\int_0^\infty dv G(v) (T\xi^{\frac{1}{3}}+v^{\frac{1}{3}})^2[\phi(\xi+\frac{v}{T^3},s)-\phi(\xi,s)] =A_T\phi(\xi,s)&\\
\phi_T(\xi,0)=\varphi_0(\xi),&.
\end{array}\right.
\end{equation}
while the limiting problem, when $T\to\infty$, is given by 
\begin{equation}\label{eq:phiinfinity}
\left\{\begin{array}{ll}\vspace{1mm}
\partial_{s}\phi(\xi,s)=\lambda \xi^{\frac{2}{3}}\frac{\partial \phi(\xi,s)}{\partial \xi} =A_\infty \phi(\xi,s)&\\
\phi_T(\xi,0)=\varphi_0(\xi).&.
\end{array}\right.
\end{equation}
We consider $\tilde X_M:=\{f\in C_b(\R^{+})\;:\;\supp f\subset [0,M]\}$ which is the analogous of the set $X_M$ defined by \eqref{def:X_M} when we skip the position variable dependence.
We observe that $(\tilde X_M,\|\cdot\|_{\infty})$ is a Banach space and $\tilde X_M$ is a closed subset of $C([0,\infty])$ in the uniform topology. Here we denote by $[0,\infty]$ the compactification of the real line $\R$ and we assume constant boundary conditions, i.e. $\varphi\in \tilde X_M$ if $\varphi(\xi)=\varphi(M)$, $\forall \xi\geq M$. 
Then the operators $A_T,\,A_\infty:\tilde X_M\to \tilde X_M$ have the following domains
\begin{align}
& D(A_T)= \tilde X_M\\&
D(A_\infty)=\{f\in \tilde X_M\;:\; \xi^{\frac 2 3 }\partial_{\xi} f(\xi)\in \tilde X_M\}, \quad \bar{D}(A_\infty)=\tilde X_M .
\end{align} 
Moreover, we notice that the space $D=C([0,\infty])\cap \tilde X_M$ is a core for $A_\infty$.
Indeed, solving $(A_\infty-\mu \mathbb{I})\phi(\xi)=h(\xi)$ for $h\in C([0,\infty])$ and assuming without loss of generality $\mu=1$, we get $\phi(\xi)= e^{3\xi^{\frac 1 3}} \int_{\xi}^{\infty} d\eta \, e^{-3\eta^{\frac 1 3}}\, \frac{h(\eta)}{\eta^{\frac 2 3}}$. We note that the operator $A_T$ is the generator of a semigroup of contractions $U_T(t)$ in $(\tilde X_M,\|\cdot\|_{\infty})$.

Our goal is to show that,  as $T\to\infty$, $A_T f$ converges to the limiting operator $A_\infty f$ for any function $f\in D$. This guarantees, by applying Trotter-Kurtz Theorem (see Theorem 2.12 in \cite{Li}), %Theorem 2.9 in \cite{Li}), 
the uniform convergence of the corresponding semigroups. 
\end{proofof}

\begin{proofof}[Proof of Theorem \ref{th:asymp}]
Using the duality formula \eqref{duality_volumes} we get
\begin{equation*}
\begin{split}
&\int_0^\infty \varphi_0(W)F(WT^3,T)T^3\,dW=\int_0^\infty \varphi_0\Big(\frac{V}{T^3}\Big)F(V,T)dV \\&=
\int_0^\infty F(V,0)U_T(1)\varphi_0\Big(\frac{V}{T^3}\Big) dV=\int_0^\infty F(V,0)\phi_T\Big(\frac{V}{T^3},1\Big) dV. 
\end{split}
\end{equation*}
We look at 
\begin{equation*}
\begin{split}
\left\vert\phi_T\Big(\frac{V}{T^3},1\Big) -\phi_\infty( 0,1) \right\vert&\leq  \left\vert\phi_T \Big(\frac{V}{T^3},1\Big) -\phi_\infty\Big(\frac{V}{T^3},1\Big) \right\vert+\left\vert\phi_\infty\Big(\frac{V}{T^3},1\Big) -\phi_\infty( 0,1) \right\vert\\&
\leq \sup_{\frac{V}{T^3}\in [0,\infty]}  \left\vert\phi_T\Big(\frac{V}{T^3},1\Big) -\phi_\infty\Big(\frac{V}{T^3},1\Big) \right\vert \\&
\quad + \left\vert\varphi_0\Big(\Big(\frac{V}{T^3}^3+\frac 1 3\Big)^{\frac 1 3}\Big) -\varphi_0\Big(\Big(\frac 1 3\Big)^{\frac 1 3}\Big) \right\vert.
\end{split}
\end{equation*}
In the last inequality we used the explicit form of $\phi_\infty$, solution of \eqref{eq:phiinfinity} by characteristics. Using Proposition \ref{prop:asymp} and the fact that $\varphi_0\in C_c(\R^+)\subset \tilde X_M$ we have that $\left\vert\phi_T\Big(\frac{V}{T^3},1\Big) -\phi_\infty( 0,1) \right\vert\to 0$ as $T\to\infty$.
\end{proofof}

\section{Global well posedness for the CTP model}\label{sec:WP}
The main purpose of this section is to prove well posedness for the CTP model which is the content of Theorem \ref{th:theorem2}. This proof will require two ingredients: first we will prove that the coalescence events have a finite number of steps with probability one, and second we prove that the total length of the free flights of the tagged particle is infinite with probability one. As a preliminary step we introduce a suitable notation in the next section.
\subsection{Probability of a sequence of coagulation events and free flights}\label{ssec:CharCTP}		 														

%%%%%%

We define for every configuration $\omega\in \Omega$ and an initial condition for the tagged particle $(X_0,V_0)$ a sequence of the form 
\begin{equation}\label{def:dyn_sequence}
\begin{split}
%&(\b x^{1,1}, \b v^{1,1}; \b x^{1,2}, \b v^{1,2}; \dots;\b x^{1,m_{1}}, \b v^{1,m_{1}}; \emptyset, l_1,\beta_1,w_1;%\\&\quad 
%\b x^{2,1}, \b v^{2,1}; \b x^{2,2}, \b v^{2,2}; \dots;\b x^{2,m_{2}}, \b v^{2,m_{2}}; \emptyset, l_2,\beta_2,w_2;\dots),
&(\b x^{1,1}, \b v^{1,1}; \b x^{1,2}, \b v^{1,2}; \dots;\b x^{1,m_{1}}, \b v^{1,m_{1}}; l_1,\beta_1,w_1;%\\&\quad 
\b x^{2,1}, \b v^{2,1}; \b x^{2,2}, \b v^{2,2}; \dots;\b x^{2,m_{2}}, \b v^{2,m_{2}}; l_2,\beta_2,w_2;\dots),
\end{split}
\end{equation}
where $\b x^{j,l}=\{x_{1}^{j,l},\dots,x_{n_{j,l}}^{j,l}\}$ with $x_{k}^{j,l}\in \R$ %they are set.not ordered
and $\b v^{j,l}=\{v_{1}^{j,l},\dots,v_{n_{j,l}}^{j,l}\}$ with $v_{k}^{j,l}\in \R^{+}$, $l_k> 0$, $\beta_k\in S^2$, $w_k\geq 0$.
These sequences are defined according to the construction of the flow given in Section \ref{sec:PM} in \eqref{eq:dom_free_mot}, \eqref{eq:dom_coll}, \eqref{def:merging_op}. 
The pairs $(\b x^{j,l}, \b v^{j,l})$ describe the set of particles coalescing at any single step as given by \eqref{eq:dom_coll}. We remark that we call step any group in this sequence separated from the rest by semicolon. Notice that each step could be of coagulation type, i.e. $[\b x^{j,l},\b v^{j,l}]$, or of flight type, i.e. $[l_k,\beta_k,w_k]$. 
Here $l_k$ is the length of a free flight between two sequences of coagulation events, $\beta_k$ is the vector $\beta_k=\cos\theta e_1+\sin\theta \,\nu$ defined by \eqref{def:merg_dyn} and $w_k$ is the volume of the obstacle colliding with the tagged particle after the free flight. Notice that the collisions at the end of a free flight are only binary collisions with probability one. This allows to ignore multiple collisions at the end of a free flight. 
We remark however that the probability of having multiple collisions during a coalescence step is strictly positive. This is the reason why, in general, the number of elements in the set $(\b x^{j,l}, \b v^{j,l})$ is larger or equal than one. We call coalescence events a sequence of coalescence steps between free flights.  
It might happen with positive probability that a sequence of coalescence events in between two free flights $\b x^{j,1}, \b v^{j,1}; \b x^{j,2}, \b v^{j,2}; \dots;\b x^{j,m_{j}}, \b v^{j,m_{j}}$ is empty.

In some cases we do not need the full information contained in the sequence \eqref{def:dyn_sequence} and it could be enough to deal with the number of obstacles involved in the coagulation at each step and the length of the free flights. This allow to consider the sequence of the form
$$(n_{1,1}\dots n_{1,m_1},l_1,n_{2,1}\dots n_{2,m_2},l_2,\dots, n_{k,1}\dots n_{k,m_k},l_{k},\dots )$$
which can be written, in a more compact way, as
$$(M_1,l_1,M_2,l_2,\dots, M_{k},l_{k},\dots )$$
where $M_j$ denotes the $j-$th level of coalescence, i.e. $M_j=n_{j,1}\dots n_{j,m_j}$ and we allow the case $M_j=\emptyset$. We could, eventually, include informations on the vectors $\beta_k$ if needed.

We are interested in computing probabilities of the form $\PP(M_1,l_1,M_2,l_2,\dots, M_{k},l_{k})$. %The main difficulty that we need to handle is that the events considered are not independent, due to the dependence of the integration domains on the previous history, and therefore the Poisson process does not factorize.  
We need to keep track of the dependence of the integration domains on the previous history, in order to reduce the computations of the probabilities to products of probabilities of independent events.

For a given $\omega \in\Omega$ we define for any level $(\star)$ a family of domains $W_{(\star)}(v,\omega, \xi^{(\star-1)})$ parametrized by the volume $v$. The level $(\star)$ indicates at which step in the sequence \eqref{def:dyn_sequence} we are. As we observed before, the step $(\star)$ could be of coagulation type %i.e. $[x^{j,l},v^{j,l}]$, 
or of flight type. %, i.e. $[l_k,\beta_k,w_k]$. 
The domains $W_{(\star)}(v,\omega, \xi^{(\star-1)})$ depend on the previous history and they have the property that any obstacle with volume larger than $v$ contained in these domains would have collided with the tagged particle. Here we use the shorthand notation $\xi^{(\star-1)}$ to denote the portion of the sequence \eqref{def:dyn_sequence} which consists of all the steps which are previous to the step $(\star)$, note that $\xi^{(0)}=(X_0,V_0)$. % (not included $(\star)$). 
We define the domain at level $(\star)=1$ by $W_{1}(v,\omega, \xi^{(0)})=B_{\frac{3}{4\pi}(v^{\frac 1 3}+V_0^{\frac 1 3})}(X_0)$.  
We then construct the positions and volumes $X_{\star},V_{\star}$ inductively as follows.
More precisely, if the step $(\star)$ is of coagulation type we define
\begin{equation}\label{eq:defW_coag}
W_{(\star)}(v,\omega, \xi^{(\star-1)}):=B_{\frac{4}{3\pi}(V_{{\star}-1}^{\frac{1}{3}}+v^{\frac{1}{3}})}(X_{{\star}-1}),
\end{equation}
and $X_{\star},V_{\star}$ by $(X_{\star},V_{\star})=\mathcal{A}(X_{{\star}-1},V_{{\star}-1}; \{\b x^{\star},\b v^{\star}\})$ where $\mathcal{A}$ is the merging operator defined by \eqref{def:merging_op}. We also define, for $(\star)\geq 1$, the following auxiliary domains which we denote as forbidden regions
\begin{equation}\label{def:forbreg_F}
F_{\star}(v,\xi^{(\star-1)})=F_{\star-1}(v,\xi^{(\star-2)})\cup W_{(\star)}(v,\omega, \xi^{(\star-1)}),
\end{equation}
where by definition $F_{0}=\emptyset$. % and $F_1(v,\xi^{(0)})=W_{1}(v,\omega, \xi^{(0)})$.%F_{0}(v,\xi^{(-1)})=\emptyset$.
If the step $(\star)$ is of flight type in order to define the domains $W_{(\star)}(v,\omega, \xi^{(\star-1)})$ we need to introduce some auxiliary regions. For any measurable $A\subset S^{2}_{+}:=\{n\in S^{2}\,:\, n\cdot e\geq 0\}$ we define the geometrical region
$$S_{A}(X_{{\star}-1},V_{{\star}-1},v):=\{x\in\R^3\, : \; x=X_{{\star}-1}+b\, e_1+(V_{{\star}-1}^{\frac{1}{3}}+v^{\frac{1}{3}})\beta, \;\; \beta\in A,\; b\in \R\}.$$ 
Notice that if $A_1\cap A_2= \emptyset$ then $S_{A_1}(X_{\star},V_{\star},v)\cap S_{A_2}(X_{\star},V_{\star},v)=\emptyset$. Moreover, for any measurable $A\subset S^{2}_{+}$ we define 
\begin{equation}\label{eq:defdomainQ}
Q(X_{{\star}-1},V_{{\star}-1};v,l,A,\xi^{(\star-1)}):=\bigcup_{s\in[0, l]}\big[B_{\frac{4}{3\pi}(V_{{\star}-1}^{\frac{1}{3}}+v^{\frac{1}{3}})}(X_{{\star}-1}+s\,e_1)\setminus F_{{\star}-1}(v,\xi^{(\star-2)})\big]\cap S_{A}(X_{{\star}-1},V_{{\star}-1},v).
\end{equation}
We now define  $W_{(\star)}(v,\omega, \xi^{(\star-1)})$  as 
\begin{equation}
W_{(\star)}(v,\omega, \xi^{(\star-1)}):=Q(X_{{\star}-1},V_{{\star}-1};v,\bar l,S^2_{+},\xi^{(\star-1)})%\bigcup_{A\in S^2_{+}}\bigcup_{l\in[0,\bar l]}Q(X_{\star},V_{\star};v,l,A,\xi^{(\star-2)}).
\end{equation}
 where $\bar l$ is
 the free flight length of the tagged particle in between the collision dynamics which is given by %spanned by the free flights of the tagged particle in between the collision dynamics, which depends on the previous history. More precisely, we set
\begin{equation*}
\bar{l}(\omega):=
\sup\{ l >0\,:\, \{(x_j(\omega),v_j(\omega)),\; v_j(\omega)\geq v\} \cap Q(X_{{\star}-1},V_{{\star}-1};v,l,A,\xi^{(\star-1)})=\emptyset, \forall v\geq 0 \}.
\end{equation*}
We define the forbidden region $F_{\star}(v,\xi^{(\star-1)})$ by means of \eqref{def:forbreg_F} and the following sets for $(\star)$ of coalescence type 
\begin{equation}\label{def:setBscalar}
B_{\star}(v,\xi^{(\star-1)}):=W_{(\star)}(v,\omega, \xi^{(\star-1)})\setminus F_{\star-1}(v,\xi^{(\star-2)}) \subset \R^3,
\end{equation}
and
\begin{equation}\label{def:setBvector}
B_{\star}(\b v^{(\star)}, \xi^{(\star-1)};\omega):=B_{\star}(v_1^{(\star)},\xi^{(\star-1)})\times \dots \times B_{\star}(v_{m_{\star}}^{(\star)},\xi^{(\star-1)})
\end{equation}
where $v^{(\star)}:=\{v^{(\star)}_1,\dots, v^{(\star)}_{m_{\star}}\}$.

We now describe the probability of finding the variables appearing in the sequence \eqref{def:dyn_sequence} in suitable measurable sets. More precisely we can compute the probability of finding $\b x^{k,j}\in U_{k,j}\subset (\R^3)^{n_{k,j}}$, $\b v^{k,j}\in Z_{k,j}\subset (\R)^{n_{k,j}}$, $l_k\in I_k\subset \R_{+}$, $\beta_k\in J_k\subset S^2_{+}$ and $w_k\in L_k\subset \R_{+}$ where $U_{k,j},\,Z_{k,j},\,I_k,\,J_k,\,L_k$ are measurable sets. We will denote this probability as $\PP(\{U_{k,j}\},\{Z_{k,j}\};\{I_k\},\{J_k\},\{L_k\})$.

To have a compact writing for these probabilities we now introduce a class of suitable operators. 
We set
\begin{equation}\label{def:Akj}
\begin{split}
A_{k,j}(n_{k,j},U_{k,j},Z_{k,j})=&\frac{\phi^{n_{k,j}}}{n_{k,j}!}\left [\int_{Z_{k,j}}d\b v^{k,j}G(\b v^{k,j}) \int_{B_{k,j}( \b v^{k,j}, \xi^{((k,j)-1)};\omega) \cap U_{k,j}} d\b x^{k,j}  \right] \\& \quad e^{-\phi\int_0^{\infty}G(v)|B_{k,j}(v,\xi^{((k,j)-1)})|dv},
\end{split}
\end{equation}
where $B_{k,j}(v,\xi^{((k,j)-1)})$ and $B_{k,j}(\b v^{k,j}, \xi^{((k,j)-1)};\omega)$ are defined by \eqref{def:setBscalar} and \eqref{def:setBvector} for $(\star)=(k,j)$.
Now we define 
\begin{align}\label{def:opC}
&A_{k,j}(n_{k,j})=A_{k,j}(n_{k,j},(\R^3)^{n_{k,j}},(\R)^{n_{k,j}}),\nonumber\\&
A_{k,j}:=\sum_{n=0}^{\infty}A_{k,j}(n), \qquad C_k:=\sum_{m=0}^{\infty}A_{k,1}\dots A_{k,m}.
\end{align} 

We introduce also the operator 
\begin{equation}\label{eq:expvolQ}
\begin{split}
F_{\star}(I_{\star},J_{\star},L_{\star})=&\phi \int_{I_{\star}}dl\int_{J_{\star}}d\beta \int_{L_{\star}}dv\, G(v)\,det (J) \chi(l,\beta,\xi^{(\star-2)})\,e^{-\phi\int_{0}^{\infty} G(v)|Q(X_{\star-1},V_{\star-1};v,l,S^{2}_{+},\xi^{(\star-2)})| }  %\\& 
%\frac{\partial (|Q(X_{\star-1},V_{\star-1};v,l,A,\xi^{(\star-2)})|)}{\partial (A,l)},
\end{split}
\end{equation}
where $\star=k$. Here, we recall that $\beta=\beta(\theta,\varphi)$ and $ \chi(l,\beta,\xi^{(\star-2)})= \chi(v,l,\theta,\varphi;X_{{\star}-1},V_{{\star}-1},\xi^{(\star-2)}))=\mathbb{1}_{\{x \in \R^3\setminus F_{\star-1}(v,\xi^{(\star-2)})\}}$.
The relation between the position $x$ and the angular variables $\theta,\varphi$ and the length of the flights $l$ is given by means of a suitable \textit{change of variables} that has been used before in Section \ref{sec:derivation}, see \eqref{change var}.
More precisely 
\begin{equation}\label{changevarPoisson}
\begin{split}
&x\to (l,\beta)\\&
\beta=\beta(\theta,\varphi), \quad \theta\in\big[0,\frac{\pi}{2}\big], \; \varphi\in [0,2\pi].
\end{split}
\end{equation}
 Since 
$$x=X_{\star-1}+l\,e_1+(V_{\star-1}^{\frac{1}{3}}+v^{\frac{1}{3}})(\cos\theta,\sin\theta\cos\varphi,\sin\theta\sin\varphi),$$
with $e_1=(1,0,0)$, we have that the Jacobian of the change of variable is given by $dx=\det(J)\,dl\, d\theta \, d\varphi$ with
\begin{equation}
J= \frac{\partial (x_1,x_2,x_3)}{\partial{d,\theta,\varphi}}=\left( \begin{array}{ccc}
1 & -(V_{\star-1}^{\frac{1}{3}}+v^{\frac{1}{3}})\sin\theta & 0 \\
0 & (V_{\star-1}^{\frac{1}{3}}+v^{\frac{1}{3}})\cos\theta\cos\varphi & -(V_{\star-1}^{\frac{1}{3}}+v^{\frac{1}{3}})\sin\theta\sin\varphi\\
0 & (V_{\star-1}^{\frac{1}{3}}+v^{\frac{1}{3}})\cos\theta\sin\varphi & (V_{\star-1}^{\frac{1}{3}}+v^{\frac{1}{3}})\sin\theta\cos\varphi \end{array} \right)
\end{equation} 
so that $dx=\det(J)\,dl\, d\theta \, d\varphi=(V_{\star-1}^{\frac{1}{3}}+v^{\frac{1}{3}})^2\sin\theta\cos\theta\,dl\, d\theta \, d\varphi$.
We note that
\begin{equation}\label{eq:volQ}
|Q(X_{{\star}-1},V_{{\star}-1};v,l,A,\xi^{(\star-2)})|=(V_{\star-1}^{\frac{1}{3}}+v^{\frac{1}{3}})^2\int_{0}^{l}d\eta \int_{A}d\theta d\varphi \cos\theta \sin\theta\, \chi(\eta,\theta,\varphi,\xi^{(\star-2)}).
\end{equation}
%where $ \chi(v,l,\theta,\varphi;X_{{\star}-1},V_{{\star}-1},\xi^{(\star-2)}))=\mathbb{1}_{\{x\in \R^3\setminus F_{\star-1}(v,\xi^{(\star-2)})\}}$.
Moreover, we set 
\begin{equation}\label{def:opFk}
F_k(I_k)=F_{k}(I_{k},S^2_{+},\R_+)\quad \text{and}\quad F_k:=F_k(\R_{+}). 
\end{equation}
This allows to write the probability $\PP(\{U_{k,j}\},\{Z_{k,j}\};\{I_k\},\{J_k\},\{L_k\})$ as 
\begin{equation}
\begin{split}
&A_{1,1}(n_{1,1},U_{1,1},Z_{1,1})A_{1,2}(n_{1,2},U_{1,2},Z_{1,2})\dots A_{1,m_1}(n_{1,m_1},U_{1,m_1},Z_{1,m_1})F_{1}(I_{1},J_{1},L_{1})\\&
A_{2,1}(n_{2,1},U_{2,1},Z_{2,1})A_{2,2}(n_{2,2},U_{2,2},Z_{2,2})\dots A_{2,m_2}(n_{2,m_2},U_{2,m_2},Z_{2,m_2})F_{2}(I_{2},J_{2},L_{2}) \dots  \\&
A_{k,1}(n_{k,1},U_{k,1},Z_{k,1})A_{k,2}(n_{k,2},U_{k,2},Z_{k,2})\dots A_{k,m_k}(n_{k,m_k},U_{k,m_k},Z_{k,m_k})F_{k}(I_{k},J_{k},L_{k}) \dots
\end{split}
\end{equation}
This formula follows from the fact that the particles are distributed homogeneously and the numbers of particles in disjoint domains are given by Poisson distributions which are independent in disjoint domains. 

Therefore we can compute the probabilities of sequences of coagulation events and flight as in \eqref{def:dyn_sequence} summing expressions of this form. Notice that this formula has to be thought as a sequence of integral operators acting in a non commutative way on the variables of the operators which are the subsequent ones in the formula, due to the depence on the previous variables of the variables $\xi^{(\star)}$.

We observe that if we are interested in computing the probability of events in which we are integrating over all the possible particle positions and volumes or all the possible numbers of particles or all the possible impact parameters and volumes in the case of the free flights we can the use the reduced operators $A_{k,j}(n_{k,j})$, $A_{k,j}$, $F_k(I_k)$ and $F_k$ given by \eqref{def:opC} and \eqref{def:opFk} respectively.

We will denote an arbitrary coalescence event by $C_k$, given by \eqref{def:opC}. Thus, if we consider sequences which contain at least in part of it a subsequence of arbitrary coalescence and flight events we will obtain, in the expression for the operators giving the probabilities, portions of the form
\begin{equation}
\dots C_mF_mC_{m+1}F_{m+1}\dots C_LF_L\dots  \quad L>m.
\end{equation}

\subsection{Coalescence events stop after a finite number of steps with probability one}

We now show that coalescence events of the form $[\b x^{k,1}, \b v^{k,1}; \b x^{k,2}, \b v^{k,2}; \dots;\b x^{k,m_{k}}, \b v^ {k,m_{k}}]$ stop after a finite number of steps with probability one since in principle $m_{k}$ might be infinite.
\begin{proposition}\label{prop:wellpos1}
Let $G(v)\in \mathcal{M}_{+}([0,\infty))$ be compactly supported with $\supp G(v)\in [0,v_{*}]$. Then there exists a $\phi_{*}=\phi_{*}(v_{*})>0$ such that for any $\phi\leq \phi_{*}$ and any $(Y_0,V_0)\in \R^3\times [0,\infty)$ there exists $\Omega^{*}\subset \Omega, \; \Omega^{*}\in \Sigma$ where $\Sigma$ is the $\sigma-$algebra defined in Section \ref{ssec:PM} such that $\PP(\Omega^{*})=1$ and such that for any $\omega \in \Omega^{*}$ the sequence \eqref{def:dyn_sequence} has the property that $m_j<\infty$ for any $j\in \N$.
\end{proposition}
We argue by induction. We consider any coalescence event $[\b x^{k,1}, \b v^{k,1}; \b x^{k,2}, \b v^{k,2}; \dots;\b x^{k,m_{k}}, \b v^ {k,m_{k}}]$ at level $k$, assuming that we have proved that all the previous ones have a finite number of steps with probability one. Notice that $k$ might be one. 
For the sake of simplicity in this section we will drop the index $k$. We will denote as $B_1(v,\xi^{(\star-1)})$ the domain $B_\star(v,\xi^{(\star-1)})$ defined by \eqref{def:setBscalar} with $(\star)=(k,1)$ and we write $(Y_0,V_0)=(Y_{\star-1}, V_{\star-1})$. We define then $B_2(v,\xi^{(1)})=B_{\star+1}(v,\xi^{(\star)})$ and $B_k(v,\xi^{(k-1)})=B_{\star+k}(v,\xi^{(\star+k-1)})$ for any $k$. We also define inductively the evolution of the tagged particle $(Y_{k},V_{k})=\mathcal{A}(Y_{k-1},V_{k-1};\{\b x^{k-1},\b v^{k-1}\})$ where $\mathcal{A}$ is the merging operator defined by \eqref{def:merging_op}.

We introduce now the following sequence of events $\{\Omega_k\}$ defined by
\begin{equation}\label{def:A1}
 \Omega_1:=\{\omega\,:\, \exists (x_k,v_k)\in\omega\; \text{s.t.} \; x_k\in B_1(v_k,\xi^{(\star-1)})\},
\end{equation}
\begin{equation}\label{def:A2_1}
 \Omega_2:=\{\omega\in \Omega_1\,:\, \exists (x_k,v_k)\in\omega\; \text{s.t.} \; x_k\in B_2(v_k,\xi^{(1)}) \}.
\end{equation}
Moreover, by iterating, we get
\begin{equation}\label{def:Ak}
 \Omega_k:=\{\omega\in \Omega_{k-1}\,:\,  \exists (x_j,v_j)\in\omega\; \text{s.t.} \; x_j\in B_k(v_j,\xi^{(k-1)})\}.
\end{equation}
We observe that the sequence above is such that $ \Omega_{k+1}\subset  \Omega_k$ $\forall k$.  
Our strategy to prove Proposition \ref{prop:wellpos1} will be to show that $\sum_{n=1}^{\infty}\PP( \Omega_n)<\infty$ and to apply then Borel-Cantelli Lemma.
To do this we need the following result which will play a crucial role in the rest of the paper.
\begin{lemma}\label{lem:geometriclemB_k}
Let $V_0\geq0$ be a given initial volume. Then there exists a constant $K>0$ such that 
\begin{equation}\label{eq:geometriclemB_k}
I(\xi^{(k)}):=\int_{0}^{\infty}G(v)|B_{k+1}(v,{\xi}^{(k)})| dv\leq K[(V_{k}-V_{k-1})+1]\qquad k=1,2,\dots
\end{equation}
where $B_{k+1}(v,{\xi}^{(k)})$ is as in \eqref{def:setBscalar}
 and $\displaystyle K=K(M_1)$, with $M_1=\int_{0}^{\infty} v\,G(v)dv$, is independent of the initial volume $V_0$.
\end{lemma}
We remark that the volume $V_0$ is finite since the number of coagulation steps before the level $(k,1)$ is finite with probability one. The meaning of this Lemma is that $\phi I(\xi^{(k-1)})$, i.e. the rate of the Poisson process at the level $k$, cannot be large unless there is a large difference of volumes in two successive steps, otherwise the average number of particles incorporated at the step $k$ will be small. %bounded.
\begin{remark}\label{rk:postagpc}
We observe that this is the only place where we use the explicit formula for the change of the position of the particle, see \eqref{def:merging_op}. Any other choice of coalescence operators for which the inequality \eqref{eq:bd-displY} is satisfied will imply Lemma \ref{lem:geometriclemB_k} and the corresponding dynamics can be defined globally in time with probability one.
\end{remark}
\begin{proof} %lem:geometriclemB_k
We observe that %$B_{k+1}({\xi}^{(k-1)},v)=D_{k+1}({\xi}^{(k)},v)\setminus D_1\cup \dots\cup D_{k}({\xi}^{(k-1)},v) \subseteq D_{k+1}({\xi}^{(k)},v)\setminus D_{k}({\xi}^{(k-1)},v)$
$|B_{k+1}({\xi}^{(k)},v)|\leq |B_{V_{k}^{\frac{1}{3}}+v^{\frac{1}{3}}}(Y_k)\setminus B_{V_{k-1}^{\frac{1}{3}}+v^{\frac{1}{3}}}(Y_{k-1})|$.
Therefore, since 
$$Y_{k}=Y_{k-1}+\frac{\sum_{l=1}^{n_k}(x_l^{(k)}-Y_{k-1})v_l^{(k)}}{V_{k-1}+\sum_{l=1}^{n_k}{v_l^{(k)}}}$$
and 
$$V_k=V_{k-1}+\sum_{l=1}^{n_k}{v_l^{(k)}},$$
we look at $|Y_{k}-Y_{k-1}|$ and it follows that 
\begin{equation}\label{eq:bd-displY}
|Y_{k}-Y_{k-1}|\leq (V_{k-1}^{\frac{1}{3}}+v^{\frac{1}{3}})\,\frac{\sum_{l=1}^{n_k} v_l^{(k)}}{V_{k-1}+\sum_{l=1}^{n_k}{v_l^{(k)}}}=\frac{(V_{k-1}^{\frac{1}{3}}+v^{\frac{1}{3}})}{V_k}(V_k-V_{k-1}). 
\end{equation}
Setting $\xi:=\frac{(V_k-V_{k-1})}{V_k}$, $\theta:=\big(\frac{v}{V_{k-1}}\big)^{\frac{1}{3}}$ and $Z:=Y_{k}-Y_{k-1}$ we get $|Z|\leq V_{k-1}^{\frac{1}{3}}(1+\theta)\xi$. Moreover, since 
$V_{k}^{\frac{1}{3}}+v^{\frac{1}{3}}=V_{k-1}^{\frac{1}{3}}\left(\frac{1}{(1-\xi)^{\frac{1}{3}}}+\theta \right)$, %V_{k-1}^{\frac{1}{3}}\big(\frac{1}{1-\frac{V_k-V_{k-1}}{V_k})^{\frac{1}{3}}+V_{k-1}^{\frac{1}{3}}\theta
we have that 
\begin{equation*}
\begin{split}
|B_{k+1}({\xi}^{(k)},v)|&\leq \vert B_{V_{k-1}^{\frac{1}{3}}\left(\frac{1}{(1-\xi)^{{1}/{3}}}+\theta \right)}(Z)\setminus B_{V_{k-1}^{\frac{1}{3}}(1+\theta)}(0)\vert \\&\leq 
 \vert B_{(1+\theta)\left[1+\frac{1}{(1+\theta)}\left(\frac{1}{(1-\xi)^{{1}/{3}}}-1 \right)\right]}(Z)\setminus B_{V_{k-1}^{\frac{1}{3}}(1+\theta)}(0)\vert.
\end{split}
\end{equation*} 
By performing the change of variables $Y=\frac{Z}{V_{k-1}^{\frac{1}{3}}(1+\theta)}$, the equation above reads as
\begin{equation*}
|B_{k+1}({\xi}^{(k)},v)|\leq V_{k-1}(1+\theta)^{3} \vert B_{\left(1+\frac{\lambda(\xi)}{(1+\theta)} \right)}(Y)\setminus B_{1}(0)\vert, \quad \lambda(\xi):=\frac{1}{(1-\xi)^{\frac{1}{3}}}-1.
\end{equation*} 
It results that $|Y|\leq \xi$, $\xi\in[0,1)$. We consider two possible cases.
\begin{itemize}
\item[i)] If $0\leq \xi\leq\frac{1}{2}$ then $\lambda(\xi)\leq C_0\xi$ where $C_0$ is a numerical constant and the change of the radius is $\left(1+\frac{\lambda(\xi)}{(1+\theta)} \right)-1\leq C_0\xi$. Thus, the volume of the forbidden region $\Delta(\xi,\theta):= B_{\left(1+\frac{\lambda(\xi)}{(1+\theta)} \right)}(Y)\setminus B_{1}(0)$ is such that $|\Delta(\xi,\theta)|\leq C_0\xi$. Moreover, since $\xi\leq\frac{1}{2}$ implies $V_k \leq 2\,V_{k-1}$, %since $V_k-V_{k-1}\leq \frac{1}{2} V_k$
we have 
\begin{equation}
\begin{split}
|B_{k+1}({\xi}^{(k)},v)|&\leq C_0 V_{k-1}(1+\theta)^{3} \xi \leq 4\,C_0 (1+\theta^3)V_{k-1}\frac{V_{k}-V_{k-1}}{V_{k}}\\&
% \leq 3\, C_1 \left(1+\big(\frac{v}{V_{k-1}}\big)^3\right)V_{k-1}\frac{V_{k}-V_{k-1}}{V_{k}}
\leq  4\,C_0(V_{k-1}+v)\frac{V_{k}-V_{k-1}}{V_{k}}.
\end{split}
\end{equation}
\item[ii)] If $\xi\geq\frac{1}{2}$ %$\lambda(\xi)\leq \frac{1}{4}$
the volume of the admissible region is bounded by $|\Delta(\xi,\theta)|\leq \bar C_1\left(1+\frac{\lambda(\xi)}{\theta+1}\right)^3 \leq C_1\left(\frac{\theta+\lambda(\xi)}{\theta+1}\right)^3$ where $C_1$, $\bar C_1$ are numerical constants. Here we used that $\lambda(\xi)\geq 2^{\frac{1}{3}}-1$ when $\xi\geq\frac{1}{2}$.
Moreover, since $\xi\geq\frac{1}{2}$ implies $V_k \geq 2\,V_{k-1}$, we have
\begin{equation}
\begin{split}
|B_{k+1}({\xi}^{(k)},v)|&\leq C_1 V_{k-1}(1+\theta)^{3} \left(\frac{\theta+\lambda(\xi)}{\theta+1}\right)^3 \leq  4 \,C_1 V_{k-1}(\theta^3+\lambda(\xi)^3)\\&
%\leq  \bar{C}_2 V_{k-1}\left(\frac{v}{V_{k-1}}+(\frac{1}{(1-\xi)^{\frac{1}{3}}}-1)^3)\right)
\leq  4 \,C_1 V_{k-1}\left(\frac{v}{V_{k-1}}+\frac{1}{(1-\xi)}\right)= 4\,C_1 V_{k-1}\left(\frac{v}{V_{k-1}}+\frac{V_k}{V_{k-1}}\right)=4\,C_1 (v+V_{k}).
\end{split}
\end{equation}
\end{itemize}
% We used that $(1+x^3)\leq 4(1+x^3)$
We can now estimate $\displaystyle I({\xi}^{(k)})= \int_{0}^{\infty}G(v)|B_{k+1}(v,{\xi}^{(k)})|dv$. In the case i) when $V_{k} \leq 2\,V_{k-1}$  we have
\begin{equation}
\begin{split}
I({\xi}^{(k)})&\leq 4\,{C}_0 V_{k-1}\frac{V_{k}-V_{k-1}}{V_{k}}\int_{0}^{\infty}G(v)dv+4\,{C}_0\frac{V_{k}-V_{k-1}}{V_{k}} \int_{0}^{\infty}v\,G(v)dv \\&
\leq 4\,C_0\left((V_{k}-V_{k-1})+M_1 \frac{V_{k}-V_{k-1}}{V_{k}}\right)\\&
\leq K ((V_{k}-V_{k-1})+1), 
\end{split}
\end{equation} 
where $M_1$ is the first order momentum. Here we used that $\int_{0}^{\infty} G(v)dv=1$ and that $V_k\geq V_0\geq 0$.

In the case ii), when $V_{k} \geq 2\,V_{k-1}$, we obtain
\begin{equation}
\begin{split}
I({\xi}^{(k)})&\leq 4\, C_1 V_{k}\int_{0}^{\infty}G(v)dv+4\, C_1\int_{0}^{\infty}v\,G(v)dv \\&
\leq 4\, C_1 (V_{k}+M_1) \\&
\leq 4\, C_1  ((V_{k}-V_{k-1})+1). 
\end{split}
\end{equation} 
Note that in the last inequality we used that $V_{k}-V_{k-1}\geq  \frac 1 2 V_{k}$.
This concludes the proof.
\end{proof}
We will also use the following calculus Lemma.
 \begin{lemma}\label{lem:boundprob}
 Suppose that $\{Z_k\}_{k=1}^{\infty}$ is a sequence of non negative numbers, satisfying $Z_k\leq 1$ for any $k$, such that % \PP(\Omega_k)=Z_k
  \begin{equation}
 \label{eq:boundprob2}
%Z_1\leq 1\\&
 Z_k\leq \varepsilon^{k}+\sum_{m=1}^{k-1}\varepsilon^{k-m} Z_m \qquad k\geq k_{*}+1% k=2,3,\dots ,
    \end{equation} 
    for some $\varepsilon\in [0,\frac 1 4]$ and some integer $k_{*}\geq 1$. Then 
\begin{equation}\label{eq:boundZm}
Z_k\leq (4\varepsilon)^{k-k_*} \qquad \forall k\geq k_*.
\end{equation} 
  \end{lemma}

   \begin{proofof}[Proof of Lemma \ref{lem:boundprob}]
  We first consider the case $k_{*}=1$. We claim that in this case 
  \begin{equation}\label{eq:boundZm-2}
Z_k\leq (2\varepsilon)^{k-k_*} \qquad \forall k\geq 1.
\end{equation} 
We argue by induction. The inequality \eqref{eq:boundZm} holds for $k=1$ by assumption. Now we suppose that the estimate \eqref{eq:boundZm} holds for any $k\leq L $ and we show that it holds also for $L+1$.
  By using \eqref{eq:boundprob2} it is straightforward to show that
  \begin{equation*}
   \begin{split}
  Z_{L+1}&\leq \varepsilon^{L+1}+\sum_{m=1}^{L}\varepsilon^{L+1-m}(2\varepsilon)^{m-1}\\&
  %\quad \varepsilon^{k}\left(1+\frac{1}{\ep}\sum_{m=1}^{k-1}2^{m-1}\right)
  =\varepsilon^{L+1}\left(1+\frac{1}{\ep}(2^{L}-1)\right)\leq (2\varepsilon)^{L},
   \end{split}
   \end{equation*}
   and \eqref{eq:boundZm-2} follows.
   
   Suppose now that $k_{*}> 1$. We have
  $$Z_k\leq \varepsilon^{k}+\sum_{m=1}^{k_{*}-1}\varepsilon^{k-m} Z_m+\sum_{m=k_{*}}^{k}\varepsilon^{k-m} Z_m$$
  for $k\geq k_*+1$.
We perform the following change of variables: $k=k_*-1+l$ and $Z_k=\zeta_l$ so that the formula above reads as
 \begin{equation*}
 \begin{split}
 \zeta_l& %\leq\varepsilon^{k_*-1+l}+\sum_{m=1}^{k_{*}-1}\varepsilon^{k-m} Z_m+\sum_{m=k_{*}}^{k}\varepsilon^{k-m} Z_m\\&
 \leq \varepsilon^{k_*-1+l}+\sum_{m=1}^{k_{*}-1}\varepsilon^{k-m} +\sum_{m=k_{*}}^{k}\varepsilon^{k-m} Z_m\\&
% \leq %\varepsilon^{k_*-1+l}+\ep^{l}\sum_{m=1}^{k_{*}-1}\varepsilon^{k_*-1-m} +\sum_{m=k_{*}}^{k}\varepsilon^{k-m} Z_m 
%\leq \varepsilon^{k_*-1+l}+\ep^{l}\frac{1}{1-\ep} +\sum_{m=k_{*}}^{k}\varepsilon^{k-m} Z_m
\leq \ep^{l}\Big(\varepsilon^{k_*-1}+\frac{1}{1-\ep} \Big)+\sum_{m=k_{*}}^{k}\varepsilon^{k-m} Z_m
  \end{split}
  \end{equation*}
  where, in the second term, we used that $Z_m\leq 1$.
  Moreover, we change the index of the last sum and set $m=k_*-1+j$. We get
   \begin{equation*}
 \begin{split}
  \zeta_l&\leq\ep^{l}\Big(\varepsilon^{k_*-1}+\frac{1}{1-\ep} \Big)+\sum_{j=1}^{l-1}\ep^{l-j} \zeta_j
 \end{split}
  \end{equation*}
  for any $l\geq 2$. Therefore, since $\ep<\frac 1 4$, we obtain
  $$\zeta_l\leq(2\ep)^{l}+\sum_{j=1}^{l-1}(2\ep)^{l-j} \zeta_j$$
  for any $l\geq 2$. Applying \eqref{eq:boundZm-2} with $\ep\to 2\ep$ and $Z_k\to \zeta_k$ we obtain \eqref{eq:boundZm} and the result follows. %  where in the first inequality we used \eqref{eq:boundprob2}
   \end{proofof}

  \bigskip

\begin{proofof}[Proof of Proposition \ref{prop:wellpos1}]
Our goal is to prove that $\sum_{n=1}^{\infty}\PP( \Omega_n)<\infty$. This allows to apply Borel-Cantelli Lemma which tell us that the probability that infinitely many of them occur is 0, that is $\PP(\underset{n\to\infty}{\limsup} \,\Omega_n)=0$, or equivalently $\PP(\bigcap_{n=1}^{\infty}\bigcup_{k\geq n}^{\infty} \Omega_k)=0$. 
With this purpose we need to compute $\PP( \Omega_n)$. We start from 
$$\PP( \Omega_1)=\sum_{n=1}^{\infty}e^{-\phi\int_{0}^{\infty}|B_1(v,\xi^{(0)})|G(v)dv}\frac{\phi^n}{n!} \left(\int_{0}^{\infty}|B_1(v,\xi^{(0)})|G(v)dv\right)^n$$
where $B_1(v,\xi^{(0)})$ is as in \eqref{def:setBscalar}
and
\begin{equation*}
\begin{split}
\PP( \Omega_2)&=\sum_{n_1=1}^{\infty} e^{-\phi \int_{0}^{\infty} |B_1(v,\xi^{(0)})|G(v)dv}\,\frac{\phi^{n_{1}}}{{n_{1}}!} \int_{0}^{\infty}G(v)dv_1 \dots \int_{0}^{\infty}G(v)dv_{n_{1}} 
\\&
\quad \int_{B_1(v_1,\xi^{(0)})}dx_1\dots\int_{B_1(v_{n_1},\xi^{(0)})}dx_{n_1} \,p(\Omega_2\,|\{(\b x^{(1)},\b v^{(1)})\})\\&
=\sum_{n_1=1}^{\infty} e^{-\phi \int_{0}^{\infty} |B_1(v,\xi^{(0)})|G(v)dv}\,\frac{\phi^{n_{1}}}{{n_{1}}!} \int_{0}^{\infty}G(v)dv_1 \dots \int_{0}^{\infty}G(v)dv_{n_{1}} \int_{B_1(v_1,\xi^{(0)})}dx_1\dots 
\\&
\quad \int_{B_1(v_{n_1},\xi^{(0)})}dx_{n_1} \sum_{n_2=1}^{\infty}e^{-\phi\int_{0}^{\infty}|B_2(v,\xi^{(1)}))|G(v)dv}\,\frac{\phi^{n_{2}}}{{n_{2}}!} \left(\int_{0}^{\infty}|B_2(v,\xi^{(1)}))|G(v)dv\right)^{n_2}.
\end{split}
\end{equation*}

In a compact form we can rewrite the equation above as
\begin{equation*}
\begin{split}
\PP( \Omega_2)&=\sum_{n_1=1}^{\infty} \sum_{n_2=1}^{\infty} e^{-\phi \int_{0}^{\infty} |B_1(v,\xi^{(0)})|G(v)dv}\,\frac{\phi^{n_{1}+n_{2}}}{{n_{1}}!{n_{2}}!} \int_{[0,\infty)^{n_1}}G(\b v^{(1)})d\b v^{(1)} \int_{B_1(\b v^{(1)},\xi^{(0)})}d\b x^{(1)}
\\&\quad 
e^{-\phi\int_{0}^{\infty}|B_2(v,\xi^{(1)}))|G(v)dv}\left(\int_{0}^{\infty}|B_2(v,\xi^{(1)}))|G(v)dv\right)^{n_2},
%=\sum_{n_1=1}^{\infty}\sum_{n_2=1}^{\infty}e^{-\phi\int_{0}^{\infty}|D_1(v)|G(v)dv-\phi\int_{0}^{\infty}|B_2(\xi^{(1)},v))|G(v)dv}\,\frac{\phi^{n_{1}+n_{2}}}{{n_{1}}!{n_{2}}!}
\end{split}
\end{equation*}
with $B_1(\b v^{(1)},\xi^{(0)}) $ as in \eqref{def:setBvector}. %$D_{1}(\b v^{(1)}):=D_{1}( v_1^{(1)})\times D_{1}(v_2^{(1)})\times \dots\times D_{1}(v_{n_1}^{(1)})$.
Therefore, by iterating, we can express $\PP(\Omega_k)$ as
\begin{equation}\label{eq:PAk}
\begin{split}
\PP( \Omega_k)= &\sum_{n_1=1}^{\infty}\sum_{n_2=1}^{\infty}\dots \sum_{n_k=1}^{\infty}\frac{\phi^{n_1+n_2+\dots +n_k}}{n_1!\, n_2! \dots n_k!} \int_{[0,\infty)^{n_1}}\hspace{-0.4cm}G(\b v^{(1)})d\b v^{(1)}\dots  \int_{[0,\infty)^{n_{k-1}}}\hspace{-0.4cm}G(\b v^{(k-1)})d\b v^{(k-1)}\\& 
\int_{B_1(\b v^{(1)},\xi^{(0)})}\hspace{-0.4cm}\b{x}^{(1)}\int_{B_2(\b v^{(2)},\xi^{(1)})}\hspace{-0.4cm}d\b{x}^{(2)}\dots \int_{B_{k-1}(\b v^{(k-1)},\xi^{(k-2)})} \hspace{-0.7cm}d\b{x}^{(k-1)}
\,e^{-\phi \int_{0}^{\infty} |B_1(v,\xi^{(0)})|G(v)dv}\\&
 e^{- \phi \int_{0}^{\infty} |B_2(v,\xi^{(1)}))|G(v)dv -\dots-\phi\int_{0}^{\infty} |B_{k}(v,{\xi}^{(k-1)})|G(v)dv}\left(\int_{0}^{\infty}|B_k(v,\xi^{(k-1)}))|G(v)dv\right)^{n_k}.%\\&
\end{split}
\end{equation}
In order to estimate $\PP( \Omega_k)$ we consider separately the two contributions due to 
$I(\xi^{k-1})\leq \frac{1}{\sqrt{\phi}}$ and $I(\xi^{k-1})\geq \frac{1}{\sqrt{\phi}}$. We have
\begin{equation}\label{eq:PAkchar}
\begin{split}
\PP( \Omega_k) &=\sum_{n_1=1}^{\infty}\sum_{n_2=1}^{\infty}\dots \sum_{n_k=1}^{\infty}\frac{\phi^{n_1+n_2+\dots +n_k}}{n_1!\, n_2! \dots n_k!} \int_{[0,\infty)^{n_1}}\hspace{-0.4cm}G(\b v^{(1)})d\b v^{(1)}\dots  \int_{[0,\infty)^{n_{k-1}}}\hspace{-0.4cm}G(\b v^{(k-1)})d\b v^{(k-1)}\\& 
\quad\int_{B_1(\b v^{(1)},\xi^{(0)})}\hspace{-0.4cm}d\b{x}^{(1)}\int_{B_2(\b v^{(2)},\xi^{(1)})}\hspace{-0.4cm}d\b{x}^{(2)}\dots \int_{B_{k-1}(\b v^{(k-1)},\xi^{(k-2)})} \hspace{-0.6cm} d\b{x}^{(k-1)}
\, e^{-\phi \int_{0}^{\infty} |B_1(v,\xi^{(0)})|G(v)dv}\\&
\quad e^{- \phi \int_{0}^{\infty} |B_2(v,\xi^{(1)}))|G(v)dv -\dots-\phi\int_{0}^{\infty} |B_{k}({v,\xi}^{(k-1)})|G(v)dv}\left(\int_{0}^{\infty}|B_k(v,\xi^{(k-1)}))|G(v)dv\right)^{n_k}\\&
\quad \chi(\{\xi\, :\,I(\xi^{k-1})\leq \frac{1}{\sqrt{\phi}} \})+\chi(\{\xi\, :\,I(\xi^{k-1})\geq \frac{1}{\sqrt{\phi}} \})\\&
=:\PP(\Omega_k)_{(1)}+\PP(\Omega_k)_{(2)}. %\quad
\end{split}
\end{equation}
From the $n_k$-th contribution in $\PP( \Omega_k)_{(1)}$ we get
$$\sum_{n_k=1}^{\infty} \frac{\phi^{n_k}}{n_k!} e^{-\phi I(\xi^{k-1})}I(\xi^{k-1})^{n_k}\leq (1-e^{-\phi |B_k(v,\xi^{(k-1)}|})\leq \sqrt{\phi}.$$
Therefore we obtain
\begin{equation}\label{eq:POmegak}
\PP( \Omega_k)_{(1)}\leq \sqrt{\phi}\,\PP(\Omega_{k-1}).
\end{equation}
For what concerns $\PP( \Omega_k)_{(2)}$ we remark that $I(\xi^{k-1})\geq \frac{1}{\sqrt{\phi}}$ implies that $(V_{k-1}-V_{k-2})\geq \frac{1}{2K\sqrt{\phi}}$ (see Lemma \ref{lem:geometriclemB_k}). Since $(V_{k-1}-V_{k-2})=(V_{k-1}(\xi^{k-2},\b v^{k-1},\b x^{k-1})-V_{k-2}(\xi^{k-3},\b v^{k-2},\b x^{k-2}))=\sum_{l=1}^{n_{k-1}}v_l^{(k-1)}$ we can rewrite the condition before as $\sum_{l=1}^{n_{k-1}}v_l^{(k-1)}\geq \frac{1}{2K\sqrt{\phi}}$. Moreover, due to the fact that we are considering $G(v)$ with compact support, $\sum_{l=1}^{n_{k-1}}v_l^{(k-1)} \leq v_{*}n_{k-1}$. Therefore we obtain $n_{k-1}\geq \frac{1}{2Kv_{*}\sqrt{\phi}}$.
Thus, from the $n_{k-1}$ and $n_k$ contributions in $\PP( \Omega_k)_{(2)}$ can be estimated by
\begin{equation}
\begin{split}
&\sum_{n_{k-1}=1}^{\infty}\sum_{n_k=1}^{\infty}\frac{\phi^{n_{k-1} +n_k}}{n_{k-1}! n_k!} \int_{[0,\infty)^{n_{k-1}}}\hspace{-0.5cm}G(\b v^{(k-1)})d\b v^{(k-1)}
\int_{B_{k-1}(\b v^{(k-1)},\xi^{(k-2)})} \hspace{-0.9cm}d\b{x}^{(k-1)}
e^{- \phi \int_{0}^{\infty} |B_{k-1}(v,\xi^{(k-2)}))|G(v)dv}\\&
\quad e^{ -\phi\int_{0}^{\infty} |B_{k}(v,{\xi}^{(k-1)})|G(v)dv}\left(\int_{0}^{\infty}|B_k(v,\xi^{(k-1)}))|G(v)dv\right)^{n_k}\chi(\{\xi\, :\,(V_{k-1}-V_{k-2})\geq \frac{1}{2K\sqrt{\phi}} \})\\&
\leq \sum_{n_{k-1}=1}^{\infty}\frac{\phi^{n_{k-1}}}{n_{k-1}!} \int_{[0,\infty)^{n_{k-1}}}\hspace{-0.5cm}G(\b v^{(k-1)})d\b v^{(k-1)}
 \int_{B_{k-1}(\b v^{(k-1)},\xi^{(k-2)})} \hspace{-0.6cm}d\b{x}^{(k-1)}
e^{- \phi \int_{0}^{\infty} |B_{k-1}(v,\xi^{(k-2)}))|G(v)dv}\\&
\quad \chi(\{\xi\, :\,n_{k-1}\geq \frac{1}{2Kv_{*}\sqrt{\phi}} \})
\end{split}
\end{equation}
and we end up with
\begin{equation}\label{eq:PAk2char_2}
\begin{split}
\PP( \Omega_k)_{(2)}&\leq \sum_{n_1=1}^{\infty}\sum_{n_2=1}^{\infty}\dots \sum_{n_{k-2}=1}^{\infty}\frac{\phi^{n_1+n_2+\dots +n_{k-2}}}{n_1!\, n_2! \dots n_{k-2}!} \int_{[0,\infty)^{n_1}}\hspace{-0.4cm}G(\b v^{(1)})d\b v^{(1)}\dots  \int_{[0,\infty)^{n_{k-2}}}\hspace{-0.6cm}G(\b v^{(k-2)})d\b v^{(k-2)}\\& 
\quad\int_{B_1(\b v^{(1)},\xi^{(0)})}d\b{x}^{(1)}\dots \int_{B_{k-2}(\b v^{(k-2)},\xi^{(k-3)})}\hspace{-0.6cm}d\b{x}^{(k-2)} 
\quad e^{-\phi \int_{0}^{\infty} |B_1(v,\xi^{(0)})|G(v)dv-\dots}\\&
\quad e^{-\phi\int_{0}^{\infty} |B_{k-2}(v,{\xi}^{(k-3)})|G(v)dv}\sum_{n_{k-1}=\frac{1}{2Kv_{*}\sqrt{\phi}}}^{\infty}\frac{\phi^{n_{k-1}}}{n_{k-1}!} \left(\int_{0}^{\infty}|B_{k-1}(v,\xi^{(k-2)})|G(v)dv \right)^{n_{k-1}}\\&
\quad e^{- \phi \int_{0}^{\infty} |B_{k-1}(v,\xi^{(k-2)}))|G(v)dv}.
%\quad \int_{B_{k-1}(\xi^{(k-2)},\b v^{(k-1)})} d\b{x}^{(k-1)}
\end{split}
\end{equation}
We make a further decomposition distinguishing the case $I(\xi^{k-2})\leq \frac{\xi_{*}}{2Kv_{*}\phi^{3/2}}$ from the case $I(\xi^{k-2})>\frac{\xi_{*}}{2Kv_{*}\phi_{*}^{3/2}}$ where $\xi_{*}$ is as in Lemma \ref{lem:Poisson}. We assume that $0<\phi\leq \phi_{*}$ so that $\frac{\xi_{*}}{2Kv_{*}\phi^{3/2}}>2K$. We insert the corresponding characteristic functions in \eqref{eq:PAk2char_2} and we obtain
\begin{equation}\label{eq:PAk2char_3}
\begin{split}
\PP( \Omega_k)_{(2)}&\leq \sum_{n_1=1}^{\infty}\sum_{n_2=1}^{\infty}\dots \sum_{n_{k-2}=1}^{\infty}\frac{\phi^{n_1+n_2+\dots +n_{k-2}}}{n_1!\, n_2! \dots n_{k-2}!} \int_{[0,\infty)^{n_1}}\hspace{-0.4cm}G(\b v^{(1)})d\b v^{(1)}\dots  \int_{[0,\infty)^{n_{k-2}}}\hspace{-0.6cm}G(\b v^{(k-2)})d\b v^{(k-2)}\\& 
\quad\int_{B_1(\b v^{(1)},\xi^{(0)})}d\b{x}^{(1)}\dots \int_{B_{k-2}(\b v^{(k-2)},\xi^{(k-3)})} d\b{x}^{(k-2)} 
\quad e^{-\phi \int_{0}^{\infty} |B_1(v,\xi^{(0)})|G(v)dv-\dots}\\&
\quad e^{-\phi\int_{0}^{\infty} |B_{k-2}({\xi}^{(k-3)},v)|G(v)dv}\sum_{n_{k-1}\geq\frac{1}{2Kv_{*}\sqrt{\phi}}}\frac{\phi^{n_{k-1}}}{n_{k-1}!} \left(\int_{0}^{\infty}|B_{k-1}(\xi^{(k-2)},v)|G(v)dv \right)^{n_{k-1}}\\& 
\quad e^{- \phi \int_{0}^{\infty} |B_{k-1}(\xi^{(k-2)},v))|G(v)dv}\chi(\{\xi: I(\xi^{k-2})\leq \frac{\xi_{*}}{2Kv_{*}\phi^{3/2}}\})+\chi(\{\xi : I(\xi^{k-2})\geq \frac{\xi_{*}}{2Kv_{*}\phi^{3/2}} \})\\&=:I_1+I_2.
\end{split}
\end{equation}
To control $I_1$ we estimate the last sum in the right hand side using Lemma \ref{lem:Poisson} in Appendix \ref{appendix1} with $\zeta=\phi I(\xi^{k-2})$, $a=\frac{|\log(\xi_{*})|}{2}$ %$n=n_{k-1}$ 
and $N=N_k:=\frac{1}{2Kv_{*}\sqrt{\phi}} $. It follows that 
\begin{equation}\label{eq:I_1}
I_1\leq C\,e^{-\frac{a}{2Kv_{*}\phi^{3/2}}}\PP( \Omega_{k-2}), \qquad C=\frac{e}{e-1}. 
\end{equation}
We now consider $I_2$. By applying Lemma \ref{lem:geometriclemB_k} to $I(\xi^{k-2})$ we obtain $(V_{k-2}-V_{k-3})\geq \frac{\xi_{*}}{(2K)^2v_{*}\phi^{3/2}}$. Since $V_{k-2}-V_{k-3}=\sum_{l=1}^{n_{k-2}}v_l^{(k-2)}$, with the same strategy used in the previous step, it follows that $n_{k-2}\geq \frac{\xi_{*}}{(2Kv_{*})^2\phi^{3/2}}$. Therefore
\begin{equation}\label{eq:I_2}
\begin{split}
I_{2}&\leq \sum_{n_1=1}^{\infty}\sum_{n_2=1}^{\infty}\dots \sum_{n_{k-3}=1}^{\infty}\frac{\phi^{n_1+n_2+\dots +n_{k-3}}}{n_1!\, n_2! \dots n_{k-3}!} \int_{[0,\infty)^{n_1}}G(\b v^{(1)})d\b v^{(1)}\dots  \int_{[0,\infty)^{n_{k-2}}}G(\b v^{(k-3)})d\b v^{(k-3)}\\& 
\quad\int_{B_1(\b v^{(1)},\xi^{(0)})} \hspace{-0.4cm}d\b{x}^{(1)}\dots \int_{B_{k-3}(\b v^{(k-3)},\xi^{(k-4)})} \hspace{-0.9cm}d\b{x}^{(k-3)} 
e^{-\phi \int_{0}^{\infty} |B_1( v,\xi^{(0)})|G(v)dv-\dots-\phi\int_{0}^{\infty} |B_{k-3}(v,{\xi}^{(k-4)})|G(v)dv}\\&
\sum_{n_{k-2}\geq \frac{\xi_{*}}{(2Kv_{*})^2\phi^{3/2}}}\frac{\phi^{n_{k-2}}}{n_{k-2}!} \left(\int_{0}^{\infty}|B_{k-2}(v,\xi^{(k-3)})|G(v)dv \right)^{n_{k-2}}e^{- \phi \int_{0}^{\infty} |B_{k-2}(v,\xi^{(k-3)}))|G(v)dv}.
%\quad  \chi(\{\xi\, :\,(V_{k-2}-V_{k-3})\geq \frac{\xi_{*}}{(2K)^2v_{*}\phi^{3/2}} \}).
\end{split}
\end{equation}
We now set iteratively $N_{l-1}=\frac{\xi_{*}N_l}{2Kv_{*}\phi}$ for $l\leq k-1$ and $N_{k-1}=\frac{1}{2Kv_{*}\sqrt{\phi}}$. Notice that the numbers $N_l$ depend on $k$ although we do not write this dependence explicitly.
By iterating the formula which defines the sequence $N_l$ we obtain
\begin{equation}\label{eq:N_l}
N_l=\left(\frac{\xi_{*}}{2Kv_{*}\phi}\right)^{k-1-l}\frac{1}{2Kv_{*}\sqrt{\phi}}.
\end{equation}
On the other hand, by iterating the argument yielding \eqref{eq:POmegak}, \eqref{eq:PAk2char_3}, \eqref{eq:I_1} and \eqref{eq:I_2} we get
\begin{equation}
\begin{split}\label{eq:pOkiter}
\PP( \Omega_k)&\leq \sqrt{\phi}\,\PP(\Omega_{k-1})+Ce^{-aN_{k-1}}\PP(\Omega_{k-2})+C e^{-aN_{k-1}}\PP(\Omega_{k-3})+\dots+Ce^{-aN_{2}}\PP(\Omega_{1})+J_{N_1}\\&
\leq \sqrt{\phi}\,\PP(\Omega_{k-1})+Ce^{-aN_{k-1}}\PP(\Omega_{k-2})+C e^{-aN_{k-2}}\PP(\Omega_{k-3})+\dots+Ce^{-aN_{2}}\PP(\Omega_{1})+J_{N_1} %Ce^{-aN_{1}}.
\end{split}
\end{equation}
where 
\begin{equation}
J_{N_1}=\sum_{n_1\geq N_1}\frac{\phi^{n_1}}{n_1!} \int_{[0,\infty)^{n_1}}G(\b v^{(1)})d\b v^{(1)}\int_{B_1(\b v^{(1)},\xi^{(0)})}d\b{x}^{(1)}e^{-\phi \int_{0}^{\infty} |B_1(v,\xi^{(0)})|G(v)dv}.
\end{equation}

We now claim that for any $\ep\in (0,1]$ there exists a $\phi_{*}=\phi_{*}(\ep)>0$ sufficiently small, such that for $\phi<\phi_{*}$ the following inequality holds:
\begin{equation}\label{claimProp}
C\,e^{-a\left(\frac{\xi_{*}}{2Kv_{*}\phi}\right)^{k-1-l}\frac{1}{2Kv_{*}\sqrt{\phi}}}\leq \ep^{k-l}
\end{equation}
where $C=\frac{e}{e-1}$ and $a$ is as in Lemma \ref{lem:Poisson}.
Indeed, taking the logarithm on both sides of \eqref{claimProp} we obtain that this inequality is equivalent to
   \begin{equation*}
  \frac{a}{(k-l)|\log\ep| +\log C} \geq  \sqrt{\phi}\left(\frac{2Kv_{*}\phi}{\xi_{*}}\right)^{k-l-1}.
   \end{equation*}
We set $m=k-l-1\geq 0$ so that 
$$\sqrt{\phi}\left(\frac{2Kv_{*}\phi}{\xi_{*}}\right)^{m}\leq  \frac{a}{(m+2)|\log\ep| +\log C}\qquad \forall m\geq 0.$$
Choosing $\phi_{*}$ sufficiently small such that $\frac{2Kv_{*}\phi_*}{\xi_{*}}<1$ and $\sqrt{\phi_*}< \frac{a}{(m+2)|\log\ep| +\log C}$ we have that \eqref{claimProp} follows.%  this proves \eqref{eq:boundexp1}

Finally we remark that for any $V_0$ and $0<\phi<\phi_{*}$ there exists $k_{*}=k_{*}(V_0)$ such that $\int_{0}^{\infty} |B_1(v,\xi^{(0)})|G(v)dv\leq C(V_0+v_{*})\leq \xi_{*}N_1$ for $k\geq k_{*}$, where $N_1$ is determined by means of \eqref{eq:N_l}.
This allows us to use Lemma \ref{lem:Poisson} to obtain that $J_{N_1}\leq Ce^{-aN_{1}}$ with $C=\frac{e}{e-1}$.
Thus, from \eqref{eq:pOkiter} we get
\begin{equation}\label{eq:pOkep}
\PP( \Omega_k)\leq \ep^{k}+\sum_{m=1}^{k-1} \ep^{k-m}\PP(\Omega_{m})\qquad \forall \; k\geq k_{*}(V_0)
\end{equation}
with $\ep<\frac 1 4$.
Using Lemma \ref{lem:boundprob} above we have that $\PP( \Omega_k)\leq (4\varepsilon)^{k-k_{*}}$ for any $k\geq k_{*}$. This guarantees that $\sum_{k=1}^{\infty}\PP( \Omega_k)<\infty$ and, by using Borel-Cantelli Lemma, concludes the proof.
%%%%%
\bigskip
\end{proofof}

We will present in Appendix \ref{appendix2} the proof of the Proposition above when $G(v)=\delta(v-1)$ because the geometric ideas behind the proof are easier to grasp. 
     
\subsection{The total length of the free flights is infinite with probability one}

The result in the previous section shows that in the sequence \eqref{def:dyn_sequence} all the coagulation events have a finite number of steps with probability one (see Proposition \ref{prop:wellpos1}). 
This does not ensure yet that the CTP model is globally well defined, with probability one, for a small but finite volume fraction $\phi$. Indeed, note that blow up might take place in finite time if the sequences of free flight times $\{\tau_j\}$ between coalescence events, satisfy $\sum_{j\geq 1}\tau_j<\infty.$ 
Indeed the following example provides a configuration of obstacles such that blow up in finite time takes place.
\begin{example}\label{ex:blowup}
We consider any sequence of positive numbers such that $\sum_{j\geq 1}l_j<\infty$. We will assume that the tagged particle starts its motion at $X_0=0$ with volume $V_0=1$. We place obstacles at positions $x_k\vec e$ where $\vec e$ is a unit vector in the direction of the motion of the tagged particle. Assume that the speed of the motion of the tagged particle is one and the values of $x_k$ are given by the sequences
\begin{align}
&x_{k+1}=x_{k}-\Big(\frac{3}{4\pi}\Big)^{\frac{1}{3}}\frac{(1+k^\frac{1}{3})k}{k+1}+\Big(\frac{3}{4\pi}\Big)^{\frac{1}{3}}(k+1)^{\frac{1}{3}}+l_{k+1}+\Big(\frac{3}{4\pi}\Big)^{\frac{1}{3}}\qquad k=0,1,\dots\\&
x_0=0. %,\quad x_1=2\Big(\frac{3}{4\pi}\Big)^{\frac{1}{3}}+l_1.
\end{align}
We are assuming also that all the obstacles are identical and have volume one. Then at the collision times between the tagged particle and the obstacles $\tau_j=\sum_{k=1}^{j}l_k$ the volume of the tagged particle becomes $V_j=j+1$, it radius is given by $R_j=\Big(\frac{3}{4\pi}(j+1)\Big)^{\frac{1}{3}}$ and the position of its center becomes $X_{j}=x_{j}-\frac{V_{j-1}}{V_j}\Big(\Big(\frac{3}{4\pi}(j+1)\Big)^{\frac{1}{3}}+R_j\Big)$. It follows that at the time $T=\sum_{j=1}^{\infty}l_j<\infty$ the volume of the tagged particle becomes infinity. 
\end{example}

\bigskip
The main result of this section is the following.
\begin{proposition}\label{prop:totlenflights}
Let $G(v)\in \mathcal{M}_{+}([0,\infty))$ be compactly supported with $\supp G(v)\in [0,v_{*}]$. Then there exists a $\phi_{*}=\phi_{*}(v_{*})>0$ such that for any $\phi\leq \phi_{*}$ and any $(Y_0,V_0)\in \R^3\times [0,\infty)$ there exists $\tilde\Omega^{*}\subset \Omega, \; \tilde\Omega^{*}\in \Sigma$ where $\Sigma$ is the $\sigma-$algebra defined in Section \ref{ssec:PM} such that $\PP(\tilde\Omega^{*})=1$ and such that for any $\omega \in \tilde\Omega^{*}$ the sequence \eqref{def:dyn_sequence} has the property that $m_j<\infty$ for any $j\in \N$ and $\sum_{j=1}^{\infty}l_j=\infty$.
\end{proposition}
In order to prove this Proposition we first prove that the probability of having coalescence events which incorporate a large number of particles decreases exponentially. 
\begin{lemma}\label{lem:WPnoblowup1}
Given a sequence with the form \eqref{def:dyn_sequence}  we define $N_{k}$ as 
\begin{equation}\label{def:Nk}
N_k:=\sum_{l=1}^{m_k}n_{k,l}
\end{equation}
(i.e. the number of obstacles involved in the $k-$th coalescence event). Let $\Theta>0$. Then 
$$ \PP(N_{k}\geq \Theta)\leq C(v_{*},V_0) \phi \,e^{-a\Theta}\quad a>0.$$
\end{lemma}
\begin{proof}
Using the notation introduced in Section \ref{ssec:CharCTP} we obtain
\begin{equation}\label{eq:pNk}
\PP(N_{k}\geq \Theta)\leq C_1F_1C_2F_2\dots F_{k-1}\left( \sum_{m\geq \Theta} \sum_{\;\;n_1+n_2+\dots+n_m\geq \Theta}A_{k,1}(n_1)\dots A_{k,m}(n_m)A_{k,m+1}(0)\right).
\end{equation}
We set 
\begin{equation}\label{def:fmN}
f_m(N)=\sum_{\underset{n_1\geq 1,\dots, n_{m}\geq 1}{\;\;n_1+n_2+\dots+n_m=N}}A_{k,1}(n_1)\dots A_{k,m}(n_m)A_{k,m+1}(0)
\end{equation}
and introduce the generating function %$N\geq m$
\begin{equation}\label{def:genfc}
F_m(\lambda)=\sum_{N=m}^{\infty}\lambda^N f_{m}(N)=\sum_{n_1=1}^{\infty}\sum_{n_2=1}^{\infty}\dots \sum_{n_{m}=1}^{\infty}\lambda^{n_{1}+n_{2}+\dots+n_{m}}A_{k,1}(n_1)\dots A_{k,m}(n_m)A_{k,m+1}(0),
\end{equation}
for $\lambda\leq 3$.
We use now the explicit formula for the operators $A_{k,j}(n_j)$ given by \eqref{def:Akj}, \eqref{def:opC} and the fact that Lemma \ref{lem:geometriclemB_k} implies
\begin{equation}\label{eq:geombd2}
%I(\xi^{(k)}):=
\int_{0}^{\infty}G(v)|B_{k}(v,{\xi}^{(k-1)})| dv\leq 2\,K \,n_{k-1}v_{*}\qquad k\geq 2,
\end{equation}
for $n_{k-1}\geq 1$ and 
\begin{equation}
%I(\xi^{(k)}):=
\int_{0}^{\infty}G(v)|B_{1}(v,{\xi}^{(0)})| dv\leq K \,(\delta_{k,1}V_0+1+v_{*})
\end{equation}
where $\delta_{k,1}V_0$ is due to the fact that in the first coagulation event (namely when $k=1$) we must include the effect of the initial volume $V_0$. While for later coagulation events, when $k>1$, we use the fact that the difference of volumes appearing in \eqref{eq:geometriclemB_k} is due to the change of volumes after the flights, therefore is bounded by $v_*$ with probability one.
Then
\begin{equation}
\begin{split}
F_m(\lambda)&\leq  \sum_{n_1=1}^{\infty}\sum_{n_2=1}^{\infty}\dots \sum_{n_{m}=1}^{\infty}\frac{(\lambda\phi)^{n_1}}{n_1!}\frac{(\lambda\phi)^{n_2}}{n_2!}\dots \frac{(\lambda\phi)^{n_{m}}}{n_{m}!} (K \,(\delta_{k,1}V_0+1+v_{*}))^{n_1}\dots (2\,K \,n_{m-1}v_{*})^{n_m}.
\end{split}
\end{equation}
We set $\gamma=2K\lambda\phi v_{*}$, then the equation above becomes
\begin{equation}
\begin{split}
F_m(\lambda)&\leq  \sum_{n_1=1}^{\infty}\sum_{n_2=1}^{\infty}\dots \sum_{n_{m}=1}^{\infty}\frac{(\gamma(\delta_{k,1}V_0+1+v_{*}))^{n_1}}{n_1!}\frac{(\gamma n_1)^{n_2}}{n_2!}\dots \frac{(\gamma n_{m-1})^{n_{m}}}{n_{m}!} \\&
\leq  \sum_{n_1=1}^{\infty}\sum_{n_2=0}^{\infty}\dots \sum_{n_{m-1}=0}^{\infty}\frac{(\gamma(\delta_{k,1}V_0+1+v_{*}))^{n_1}}{n_1!}\frac{(\gamma n_1)^{n_2}}{n_2!}\dots \frac{(\gamma n_{m-2})^{n_{m-1}}}{n_{m-1}!}e^{\gamma n_{m-1}}.
\end{split}
\end{equation}
We observe that there exists $\delta>0$ such that if $0\leq x\leq \delta$ then $e^{x}\leq 1+2x$. Thus, if $\gamma$ is such that $\frac{\gamma}{1-2\gamma}\leq \delta$ then $\gamma e^{\gamma}\leq \gamma(1+2\gamma)$. Then, summing the series in $n_{m-1}$ we get $e^{\gamma(1+2\gamma)n_{m-2}}$. Using that $\gamma+2\gamma^2\leq \frac{\gamma}{1-2\gamma}$ we can iterate this procedure and we get
 \begin{equation}
F_m(\lambda)\leq  \sum_{n_1=1}^{\infty}\frac{(\frac{\gamma}{1-2\gamma}(\delta_{k,1}V_0+1+v_{*}))^{n_1}}{n_1!}\leq (e^{\frac{\gamma}{1-2\gamma}(\delta_{k,1}V_0+1+v_{*})}-1)\leq \bar C(V_0,v_{*}) \phi.
\end{equation}
Now, using the definition of the generating function we have
 \begin{equation}
\lambda^m \sum_{N=m}^{\infty} f_{m}(N)\leq \sum_{N=m}^{\infty} \lambda^N  f_{m}(N) \leq \bar C(V_0,v_{*})\phi
 \end{equation}
 which implies, choosing $\lambda\in [2,3]$, $\sum_{N=m}^{\infty} f_{m}(N)\leq  \bar C(V_0,v_{*})\phi e^{-bm} $ for some $b>0$. 
 Thus we can estimate the sum in the right hand side of \eqref{eq:pNk} as
\begin{equation}
\sum_{m\geq \Theta} \sum_{N=m}^{\infty} f_{m}(N)\leq  \bar C(V_0,v_{*})\sum_{m\geq \Theta}\phi \,e^{-bm} \leq C(V_0,v_{*})\phi \,e^{-b\Theta}.
  \end{equation}
\end{proof}

\bigskip
From now on we will denote by $R_{k}=R_{k}(\xi)$ the radius of the tagged particle before the $k-$th flight which depends on the previous history. We recall that $l_k$ is the length of the free flight between two successive coalescing events. Our next Lemma shows that the probability of having too many small free flights is small. 
\begin{lemma}\label{lem:WPnoblowup2}
Let $(Y_0,V_0)\in \R^3\times [0,\infty)$. Suppose that the dynamics of the tagged particle is described by a sequence of the form \eqref{def:dyn_sequence} with a total number of free flights (up to a given time) given by the even number $M\geq 2$. There exists $\delta>0$ and $b>0$ depending on $G$ but independent on $M$ such that 
\begin{equation}\label{eq:WPnoblowup2}
\PP\left(\left\{\omega:\# \left\{l_k\leq \frac{\delta}{(R_{k}^2+1)}\right\}\geq\frac{M}{2}\right\}\right)\leq e^{-bM}.
%\PP\left(\left\{\omega\,:\,\# \left\{l_k> \frac{\delta}{(R_{k}^2+1)}\right\}\geq\frac{M}{2}\right\}\right)\geq 1-e^{-bM}.
\end{equation}
\end{lemma}
\begin{proof}
We notice that using the notation in Section \ref{ssec:CharCTP} we can compute the following probability as 
\begin{equation*}
\begin{split}
&\PP\left(\left\{ k_1 \leq k_2 \leq\dots \leq k_{L} ,\; \frac{M}{2}\leq L \leq M \; :\; l_j\leq \frac{\delta}{(R_{k_j}^2+1)},\;\; j=1,\dots, L \right\} \right)=\\&
=C_1\,F_1\dots C_{k_1} \Big(\int_{\{ 0 \leq l_1 \leq \delta/{(R_{k_1}^2+1)}\}}{ \hspace{-0.8cm}F_{k_1}(l_1) dl_1}\Big) C_{k_1+1}\dots \\& 
\quad\dots C_{k_2}   \Big(\int_{\{ 0 \leq l_2 \leq \delta/{(R_{k_2}^2+1)}\}}{  \hspace{-0.8cm}F_{k_2}(l_2) dl_2}\Big)C_{k_2+1} \dots \Big(\int_{\{ 0 \leq l_j \leq \delta/{(R_{k_j}^2+1)}\}}{  \hspace{-0.8cm}F_{k_j}(l_j) dl_j}\Big)\dots
\end{split}
\end{equation*}
Thus
\begin{equation}
\begin{split}\label{eq:estimateWPnoblowup2}
\PP\left(\left\{\omega:\# \left\{l_k\leq \frac{\delta}{(R_{k}^2+1)}\right\}\geq\frac{M}{2}\right\}\right)&\leq \hspace{-0.2cm}\sum_{\{k_1 \leq k_2 \leq\dots \leq k_{L}, \; \frac{M}{2}\leq L \leq M  \}}
 \hspace{-1.4cm}C_1\,F_1\dots C_{k_1} \Big(\int_{\{ 0 \leq l_1 \leq \delta/{(R_{k_1}^2+1)}\}}{  \hspace{-1cm}F_{k_1}(l_1) dl_1}\Big) \\& 
\quad \; C_{k_1+1}\dots C_{k_2}   \Big(\int_{\{ 0 \leq l_2 \leq \delta/{(R_{k_2}^2+1)}\}}{\hspace{-1cm} F_{k_2}(l_2) dl_2}\Big)C_{k_2+1} \dots \\&
\quad \; \dots \Big(\int_{\{ 0 \leq l_j \leq \delta/{(R_{k_j}^2+1)}\}}{ \hspace{-1cm} F_{k_j}(l_j) dl_j}\Big)\dots
\end{split}
\end{equation}
Thanks to \eqref{eq:volQ} it follows that $|Q(X_{\star},V_{\star};v,l,A,\xi^{(\star-1)})|\leq C\,(R_{k}^2+1)\,l$, for some $C=C(v_*)>0$. Furthermore, using \eqref{eq:expvolQ} and \eqref{def:opFk}   
we obtain
$$F_{k}\Big(\Big[0, \frac{\delta}{{(R_{k}^2+1)}}\Big]\Big)\leq 
B\,\delta, \qquad B=B(v_{*})>0.$$ 
By substituting this bound in \eqref{eq:estimateWPnoblowup2} we get
\begin{equation}
\begin{split}\label{eq:estimateWPnoblowup2_2}
\PP\left(\left\{\omega:\# \left\{l_k\leq \frac{\delta}{(R_{k}^2+1)}\right\}\geq\frac{M}{2}\right\}\right)&\leq 
 \sum_{\{k_1 \leq k_2 \leq\dots \leq k_{L}, \; \frac{M}{2}\leq L \leq M  \}} \hspace{-1.4cm}C_1\,F_1\dots C_{k_1} (B\,\delta) \dots C_{k_2+1} \dots (B\,\delta) \\& %we estimate all the C_k with 1
\leq \sum_{\{k_1 \leq k_2 \leq\dots \leq k_{L}, \; \frac{M}{2}\leq L \leq M  \}} (B\,\delta)^L=\sum_{L=\frac{M}{2}}^{M }\binom{M}{L} (B\,\delta)^L\\&
\leq C\sqrt{M}\sum_{L=\frac{M}{2}}^{M }\frac{M^{M}}{L^{L}(M-L)^{M-L}} (B\,\delta)^L,
\end{split}
\end{equation}
where we used Stirling's formula in the last inequality and $C>0$ is a numerical constant. Moreover, the last formula can be written as
\begin{equation}
\begin{split}\label{eq:sumexp}
&\sum_{L=\frac{M}{2}}^{M } e^{M\log{M}-L\log{L}-(M-L)\log{M-L}+L\log{(B\,\delta)}+\log{(C\sqrt{M})}}.\\&
\end{split}
\end{equation}
Using the change of variables $L=Mx$ where $\frac 1 2 \leq x\leq 1$ we can write the exponential factor in the equation above as 
$$M(-x\log{x}-(1-x)\log{(1-x)}+x\log{(B\,\delta)})+\log{(C\sqrt{M})}$$
and this term can be estimated for $\frac 1 2 \leq x\leq 1$ and $M\geq 2$ from above as $-bM-\log\big(\frac{M}{2}+1\big)$ choosing $\delta>0$ small (depending on $B$ and $C$). Using the fact that in \eqref{eq:sumexp} there are $\frac{M}{2}+1$ terms we can estimate the whole sum by $e^{-bM}$. This concludes the proof.
\end{proof}

\begin{proofof}[Proof of Proposition \ref{prop:totlenflights}]
We can assume without loss of generality that the initial position is $Y_0=0$. Given $ V_0\geq 0$ we now construct a family of domains as follows. 
For any $A >0$ we define $\Theta_k=\Theta_k (A)=\frac{2}{a}\log k+A $.
We then define the domains 
\begin{equation}\label{eq:domUA}
\mathcal{U}_A:=\{\omega\in \Omega^{*}\,: N_k\leq \Theta_k  (A)\quad \forall k\}
\end{equation}
where $\Omega^{*}$ is as in Proposition \ref{prop:wellpos1}.
For any $M$ even we define 
\begin{equation}\label{eq:domEM}
E_{M}:=\left\{\omega\in\Omega^{*}\,:\,k\in\{1,\dots,M\}\;\;\# \left\{l_k> \frac{\delta}{(R_{k}^2+1)}\right\}\geq\frac{M}{2}\right\} %Good sets 
\end{equation}
where we compute the length in the first $M$ flights. Moreover we set 
\begin{equation}\label{eq:domEinf}
E_{\infty}:=\bigcup_{n=1}^{\infty} \bigcap_{j\geq n}E_{2j}
\end{equation}
and then we define 
\begin{equation}\label{eq:domOmegastar}
\tilde\Omega^{*}:=\bigcup_{A=1}^{\infty}\mathcal{U}_A \cap E_{\infty}.
\end{equation}
We now claim that the set $\tilde\Omega^{*}$ has all the properties stated in Proposition \ref{prop:totlenflights}.
The fact that $m_j<\infty$ follows from the fact that $\tilde\Omega^{*}\subset \Omega^{*}$ and Proposition \ref{prop:wellpos1}. We now have to prove that $\forall\omega\in\tilde\Omega^{*}$ the following two properties hold: $\sum_{j}l_j=\infty$ and $\PP(\tilde\Omega^{*})=1$. 
We begin proving the first. Notice that for any $\omega\in\tilde\Omega^{*}$ there exists an integer $A=A(\omega)>0$ such that 
$$N_k\leq \Theta_k  (A) \quad \forall\; k\geq 1.$$
On the other hand, the volume of the tagged particle at the end of the $k-$th flight is bounded by $V_{k}\leq V_{k-1}+(N_{k}+1)v_{*}$. Then $V_{k}\leq V_{k-1}+(\Theta_k(A)+1)v_{*}$. Iterating we get
\begin{equation}
\begin{split}
V_{k}&\leq V_0+v_{*}{k}+v_{*}\sum_{l=1}^{k} \Theta_l(A) \leq  V_0+v_{*}{k}+A v_{*}k+\frac{2}{a}(k+1)\log(k+1) %V_0+2Av_{*}{k}+\frac{2}{a}\sum_{l=1}^{k}\log (l+1)\leq V_0+2Av_{*}k+\frac{4}{a}(k+1)\log (k+1)
\end{split}
\end{equation} 
which implies that the radius of the tagged particle after the $k-$th flight can be estimated by
\begin{equation}\label{eq:boundRk}
R_{k}\leq C(V_0,v_{*})[A+(k+1)\log (k+1)]^{\frac{1}{3}}.
\end{equation}
On the other hand, $\forall\omega\in\tilde\Omega^{*}$, since $\omega\in E_{\infty}$, there exists $n_{*}=n_{*}(\omega)$ such that $\omega\in  \bigcap_{j\geq n_{*}}E_{2j}$. Then $\forall j\geq  n_{*}$, for the first $2j$ flights $k=1,\dots,2j$, $\# \left\{l_k> \frac{\delta}{(R_{k}^2+1)}\right\}\geq\frac{2j}{2}=j$. Using \eqref{eq:boundRk} it follows that, $\forall j\geq  n_{*}$ among the first $2j$ flights, there are at least $j$ for which 
$$l_k> \frac{\bar\delta(V_0,v_{*})}{[A+(k+1)\log (k+1)]^{\frac{2}{3}}}$$
for some $ \bar\delta=\bar\delta(V_0,v_{*})>0.$
We will denote for each $j\geq n_{*}$ the set of indices $k$ for which the inequality above holds as $I_j$. Therefore, 
\begin{equation}
\begin{split}
\sum_{k=1}^{\infty}l_k&\geq \sum_{k=1}^{2j}l_k\geq \sum_{k\in I_j} l_k
\geq \bar\delta \sum_{k\in I_j}  \frac{1}{[A+(k+1)\log (k+1)]^{\frac{2}{3}}}\\&\geq 
\bar\delta\sum_{k=j+1}^{2j}  \frac{1}{[A+(k+1)\log (k+1)]^{\frac{2}{3}}}
\geq C_A \bar\delta  \frac{(2j)^\frac{1}{3}}{(\log (2j))^\frac{2}{3}}
\end{split}
\end{equation} 
for some $C_A>0$.
Taking the limit $j\to\infty$ we obtain $\sum_{k=1}^{\infty}l_k=\infty.$

We now prove that $\PP(\tilde\Omega^{*})=1$. In order to do this we notice that
\begin{equation}\label{eq:Omstarcomp}
(\tilde\Omega^{*})^c:=(\bigcap_{A=1}^{\infty}\mathcal{U}_A^c) \cup E_{\infty}^c
\end{equation}
and 
\begin{equation}\label{eq:prOmstarcomp}
\PP((\tilde\Omega^{*})^c)\leq \PP((\bigcap_{A=1}^{\infty}\mathcal{U}_A^c))+\PP(E_{\infty}^c).
\end{equation}
We first prove that $ \PP((\bigcap_{A=1}^{\infty}\mathcal{U}_A^c))=\PP(E_{\infty}^c)=0$. For any $A>0$ we have 
$\mathcal{U}_A^c=\{\omega\in\tilde\Omega^{*}\,:\,\exists \, k\;\text{s.t.}\; N_k> \Theta_k(A)\}$
then
\begin{equation*}
\begin{split}
\PP(\mathcal{U}_A^c)& \leq \sum_{k=1}^{\infty}\PP(N_{k}> \Theta_k(A))\leq C(v_{*},V_0) \phi  \sum_{k=1}^{\infty}\phi \, e^{-a\Theta_k(A)}\\&
 \leq C(v_{*},V_0) \phi \sum_{k=1}^{\infty} \frac{e^{-aA}}{k^2}\leq \bar{C}(v_{*},V_0) \,\phi \,e^{-a A}
\end{split}
\end{equation*}
where we used Lemma \ref{lem:WPnoblowup1} in the second inequality. Then 
$\PP((\bigcap_{A=1}^{\infty}\mathcal{U}_A^c)) \leq \bar{C}(v_{*},V_0) \,\phi \,e^{-a \bar A}$ for any $\bar A>0$ and taking the limit $\bar A\to\infty$ we get 
\begin{equation}\label{eq:prUAc}
\PP((\bigcap_{A=1}^{\infty}\mathcal{U}_A^c))=0.
\end{equation}

Using the definition of $E_{\infty}$ in \eqref{eq:domEinf} we have
\begin{equation}\label{eq:domEinfcomp}
E_{\infty}^c:=\bigcap_{n=1}^{\infty} \bigcup_{j\geq n}E_{2j}^c.
\end{equation}
Borel-Cantelli Lemma will imply that 
$\PP(E_{\infty}^c)=0$ if $\sum_{j=1}^{\infty}\PP(E_{2j}^c)<\infty.$ Definition \eqref{eq:domEM} implies 
\begin{equation}\label{eq:domEMcomp}
E_{M}^c:=\left\{\omega\in\Omega^{*}\,:\,k\in\{1,\dots,M\}\;\;\# \left\{l_k\leq \frac{\delta}{(R_{k}^2+1)}\right\}\geq\frac{M}{2}\right\}. %Good sets 
\end{equation}
Now we use \eqref{eq:WPnoblowup2} to obtain $\sum_{j=1}^{\infty}\PP(E_{2j}^c)\leq \sum_{j=1}^{\infty} e^{-2bj}<\infty$. Then $\PP(E_{\infty}^c)=0$.
Combining this with \eqref{eq:prOmstarcomp} and \eqref{eq:prUAc} we obtain $\PP((\tilde\Omega^{*})^c)=0$. Hence $\PP(\tilde\Omega^{*})=1$ and the result follows.
\end{proofof}

\subsection{End of the Proof of Theorem \ref{th:theorem2}}
\begin{proofof}[Proof of Theorem \ref{th:theorem2}]
It follows from Proposition \ref{prop:wellpos1} and Proposition \ref{prop:totlenflights}. Taking into account that the tagged particle moves at constant velocity the divergence of the series $\sum_{j}l_j$ implies that the motion is defined for arbitrary long times.
\end{proofof} 

\appendix
\section{Technical results} \label{appendix1}
In this Appendix we collect some technical results which are used in the proof of the main results of the paper. The first one is a simple analysis result which allows us to control probability of tail events.  
 \begin{lemma}\label{lem:Poisson}
 Let be $\Psi_N(\zeta):=\sum_{n=N}^{\infty} \frac{\zeta^n}{n!}e^{-\zeta}$. For any $0<\xi_{*}<e^{-2}$ we have
 \begin{equation}\label{eq:boundPoisson_lem}
 \Psi_N(\zeta)\leq \frac{e}{e-1} \,e^{-aN} \quad \text{for} \quad \zeta\leq \xi_{*} N \quad \text{and} \quad N\geq 1
 \end{equation}
 where $a=\frac{|\log(\xi_{*})|}{2}$.
  \end{lemma}
  \begin{proof} 
  We consider the sequence $a_n:=\frac{\zeta^n}{n!}e^{-\zeta}$ which results to be increasing in $\zeta$ and, for $n\geq N$ and $\zeta\leq \frac{N}{2}$, decreasing in $n$ since
  $$\frac{a_{n+1}}{a_n}=%\frac{\zeta^{n+1}}{(n+1)!}e^{-\zeta}\frac{n!}{\zeta^n}e^{\zeta}=
  \frac{\zeta}{(n+1)}\leq \frac{N}{2(N+1)}\leq \frac{1}{2}< 1.$$
  We choose $0<\xi_{*}<\frac 1 2$ and consider $\zeta\leq \xi_{*} N$. Since the function $\zeta\to \zeta^ne^{-\zeta}$ is increasing in $\zeta$ for $n\geq \zeta$, it follows that
  \begin{equation*}
 \begin{split}
 \Psi_N(\zeta)&\leq \Psi_N(\xi_{*} N)=\sum_{n=N}^{\infty} \frac{{(\xi_{*} N)}^n}{n!}e^{-\xi_{*} N}=(\xi_{*} N)^Ne^{-\xi_{*} N} \sum_{m=0}^{\infty} \frac{{(\xi_{*} N)}^m}{(m+N)!}\\&
 \leq (\xi_{*} N)^Ne^{-\xi_{*} N} \sum_{m=0}^{\infty} (\xi_{*} N)^m\frac{e^{(m+N)}}{(m+N)^{m+N}}\\& 
  = \sum_{m=0}^{\infty} e^{\Theta(M,N)},
 \end{split}
 \end{equation*}
 where in the first inequality we used that $k!\geq k^ke^{-k}$ for $k\geq 1$ (see equation 6.1.38 in \cite{AS}) and $\Theta(M,N)=\Theta_1(M,N)+\Theta_2(M,N)$ with 
 \begin{align}\label{eq:Theta1/2}
&\Theta_1(M,N)=N\log\xi_{*}+(1-\xi_{*})N\\\vspace{.8mm}&
\Theta_2(M,N)=N\log N-(m+N)\log(m+N)+m+m\log\xi_{*}+m\log N.
 \end{align}
 Moreover, 
$\Theta_1(M,N)\leq -a N$ with $a=\frac{|\log\xi_{*}|}{2}$ if $\xi_{*}<e^{-2}$. 
  This implies that from this contribution we obtain the right decay. We control the second contribution.
 \begin{equation*}
 \begin{split}
 \Theta_2(M,N)&=N\log N-(m+N)\log N-(m+N)\log(1+\frac{m}{N})+m\log\xi_{*}+m\log N+m\\&
 =-(m+N)\log(1+\frac{m}{N})+m (\log\xi_{*}+1)\\&
 \leq -m \log(1+\frac{m}{N})- m\frac{|\log\xi_{*}|}{2}
 \leq -m,
 \end{split}
 \end{equation*}
if $0<\xi_{*}<e^{-2}$. 
Therefore we get $\Psi_N(\xi_{*} N)\leq e^{-a N}\sum_{m=0}^{\infty} e^{ -m }$ 
and we have
$$\Psi_N(\zeta)\leq\Psi_N(\xi_{*} N)\leq \frac{e}{e-1}\, e^{-a N}.$$
 \end{proof}

The following Lemma allows us to control the size of the displacement of the position of the tagged particle after a coalescence event.
\begin{lemma}\label{app:claim}
Suppose that $X_0=0$ and let us consider the CTP dynamics in the reference frame in which the obstacles move towards the tagged particle. Assume that the tagged particle evolves only by means of binary collisions between the tagged particle $(X_k,V_k)$ and the obstacle $(x_{k+1},v_{k+1})$.  
Let $d_{k+1}$ be the maximal distance between particles colliding by means of binary collisions, i.e. $d_{k+1}\leq \left(\frac{3}{4\pi}\right)^{\frac{1}{3}}(V_{k}^{\frac{1}{3}}+v_{k+1}^{\frac{1}{3}})$. Then the particle position and volume are given by $X_{k+1}=X_{k}+\frac{v_{k+1}}{V_{k}+v_{k+1}}(x_{k+1}-X_{k})$, $V_{k+1}=V_{k}+v_{k+1}$ and the following estimate holds 
\begin{equation*}\label{eq:boundXk}
|X_{k}| \leq \frac{9}{2\pi} \phi^{\frac{1}{3}} (V_{k}^{\frac{1}{3}}-V_{0}^{\frac{1}{3}}).
\end{equation*}
\end{lemma}
\begin{proof}
The iterative formula giving $(X_{k+1},V_{k+1})$ is just a reformulation of the CTP dynamics in this particular setting (see \eqref{def:merging_op}).
We note that $|X_{k+1}|\leq |X_{k}|+\frac{v_{k+1}}{V_k+v_{k+1}}\phi^{\frac{1}{3}}d_k$ and we rescale $X_{k}$ as $X_{k}=\frac{3}{4\pi}^{\frac{1}{3}}\phi^{\frac{1}{3}}\xi_{k}$. Then we have
 \begin{equation*}
  \left\{
 \begin{array}{l}\vspace{0.2cm}
|\xi_{k+1}|\leq |\xi_{k}|+\frac{v_{k+1}}{V_k+v_{k+1}}(V_k^{\frac{1}{3}}+v_{k+1}^{\frac{1}{3}})
\\
|\xi_0|=0,\vspace{0.2cm}\\
V_{k+1}=V_{k}+v_{k+1}.
\end{array} \right.
\end{equation*}
Iterating we then obtain
 \begin{equation}\label{app:eq:xik}
  \left\{
 \begin{array}{l}\vspace{0.2cm}
|\xi_{k}|\leq \sum_{j=1}^{k}\frac{v_{j}}{V_{j}}(V_{j-1}^{\frac{1}{3}}+v_{j }^{\frac{1}{3}})=\sum_{j=1}^{k}\frac{v_{j}}{V_{j}}V_{j-1}^{\frac{1}{3}}+\sum_{j=1}^{k}\frac{v_{j}^{\frac{4}{3}}}{V_{j}}
\vspace{0.2cm}\\
V_{k}=\sum_{j=1}^{k}v_j\quad \text{with}\quad v_0=V_0.
\end{array} \right.
\end{equation}
We estimate the first contribution in the right hand side of \eqref{app:eq:xik} as
\begin{equation*}
\sum_{j=1}^{k}\frac{v_{j}}{V_{j}}V_{j-1}^{\frac{1}{3}}\leq \sum_{j=1}^{k}\frac{v_{j}}{V_{j}^{\frac{2}{3}}}=: J_k
\end{equation*}
since $V_{j-1}\leq V_j$. Moreover, we note that $v_j=V_{j}-V_{j-1}$ thus
\begin{equation*}
\begin{split}
J_k&=\sum_{j=1}^{k}\frac{(V_{j}-V_{j-1})}{V_{j}^{\frac{2}{3}}}=\sum_{j=1}^{k}{(V_{j}^{\frac{1}{3}}-V_{j-1}^{\frac{1}{3}}})+\sum_{j=1}^{k}V_{j-1}^{\frac{1}{3}}\big(1-\frac{V_{j-1}^{\frac{2}{3}}}{V_{j}^{\frac{2}{3}}}\big)
\\& 
\leq (V_{k}^{\frac{1}{3}}-V_{0}^{\frac{1}{3}})+\sum_{j=1}^{k}V_{j-1}\big(\frac{V_{j}^{\frac{2}{3}}-V_{j-1}^{\frac{2}{3}}}{V_{j}^{\frac{2}{3}}}\big).
\end{split}
\end{equation*}
To control the second contribution in the sum above we use that $(1+x)^n\leq 1+nx$ for $x\geq 0$ and $0\leq n\leq 1$. Hence
\begin{equation*}
%\begin{split}
\sum_{j=1}^{k}V_{j-1}\big(\frac{V_{j}^{\frac{2}{3}}-V_{j-1}^{\frac{2}{3}}}{V_{j}^{\frac{2}{3}}}\big)\leq \sum_{j=1}^{k}V_{j-1}\big((\frac{V_{j-1}+v_j)^{\frac{2}{3}}-V_{j-1}^{\frac{2}{3}}}{V_{j}^{\frac{2}{3}}}\big)\leq \frac{2}{3} \sum_{j=1}^{k} V_{j-1} \frac{v_{j}}{V_{j}^{\frac{2}{3}}}=\frac{2}{3} J_k.
%\end{split}
\end{equation*}
It follows that 
\begin{equation}
J_k\leq (V_{k}^{\frac{1}{3}}-V_{0}^{\frac{1}{3}})+\frac{2}{3} J_k
\end{equation}
and 
\begin{equation}\label{app:eq:boundJk}
J_k\leq 3 (V_{k}^{\frac{1}{3}}-V_{0}^{\frac{1}{3}}).
\end{equation}
We now consider the second term in the right hand side of \eqref{app:eq:xik} and we obtain 
\begin{equation}\label{app:eq:boundIk}
\begin{split}
&I_k:=\sum_{j=1}^{k}\frac{v_{j}^{\frac{4}{3}}}{V_{j}}=\sum_{j=1}^{k} V_{j}^{\frac{1}{3}}\big(\frac{v_{j}}{V_{j}}\big)^{\frac{4}{3}}
= \sum_{j=1}^{k} V_{j}^{\frac{1}{3}}\big(\frac{(V_{j}-V_{j-1})}{V_{j}}\big)^{\frac{4}{3}}
\\& \leq \quad  \sum_{j=1}^{k} V_{j}^{\frac{1}{3}}\frac{(V_{j}-V_{j-1})}{V_{j}}=\sum_{j=1}^{k} \frac{(V_{j}-V_{j-1})}{V_{j}^{\frac{2}{3}}}=J_k.
\end{split}
\end{equation}
Therefore, using \eqref{app:eq:boundJk}, we get
\begin{equation}
|\xi_{k}|\leq I_k+J_k\leq 2 J_k \leq 6 (V_{k}^{\frac{1}{3}}-V_{0}^{\frac{1}{3}})
\end{equation}
and, in terms of $X_k$, 
\begin{equation}\label{eq:boundXk}
|X_{k}|\leq \frac{3}{4\pi}\phi^{\frac{1}{3}} (I_k+J_k)\leq  \frac{9}{2\pi} \phi^{\frac{1}{3}} (V_{k}^{\frac{1}{3}}-V_{0}^{\frac{1}{3}}). 
\end{equation}
This concludes the proof.
\end{proof}

\section{Proof of Proposition \ref{prop:wellpos1} when the obstacles have the same size } 
\label{appendix2}

In order to clarify the argument used to prove Proposition \ref{prop:wellpos1} we present in this appendix a simpler case. More precisely here we assume that the volumes distribution function is given by $G(v)=\delta (v-1)$. This means that all the obstacles have the same volume and for simplicity we assumed that $v=1$.

\begin{proposition}\label{prop:wellpos2}
Let $G(v)\in \mathcal{M}_{+}([0,\infty))$ be such that $G(v)=\delta (v-1)$. Then there exists a $\phi_{*}=\phi_{*}(v_{*})>0$ such that for any $\phi\leq \phi_{*}$ and any $(Y_0,V_0)\in \R^3\times [0,\infty)$ there exists $\Omega^{*}\subset \Omega, \; \Omega^{*}\in \Sigma$ where $\Sigma$ is the $\sigma-$algebra defined in Section \ref{ssec:PM} such that $\PP(\Omega^{*})=1$ and such that for any $\omega \in \Omega^{*}$ the sequence \eqref{def:dyn_sequence} has the property that $m_j<\infty$ for any $j\in \N$.
\end{proposition}

The main simplification in the proof is due to the fact that the dependence on $v$ of the domains $W_{(\star)}$, $F_{(\star)}$ and $B_{(\star)}$ can be removed. This makes easier to understand the geometry behind the argument. Indeed in this case \eqref{eq:defW_coag}, \eqref{def:forbreg_F}, \eqref{def:setBscalar} and \eqref{def:setBvector} become
\begin{equation}\label{eq:defW_coag_2}
W_{(\star)}(\omega, \xi^{(\star-1)}):=B_{\frac{4}{3\pi}(V_{{\star}-1}^{\frac{1}{3}}+1^{\frac{1}{3}})}(X_{{\star}-1}),
\end{equation}
\begin{equation}\label{def:forbreg_F_2}
F_{\star}(\xi^{(\star-1)})=F_{\star-1}(\xi^{(\star-2)})\cup W_{(\star)}(\omega, \xi^{(\star-1)}),
\end{equation}
\begin{equation}\label{def:setBscalar_2}
B_{\star}(\xi^{(\star-1)}):=W_{(\star)}(\omega, \xi^{(\star-1)})\setminus F_{\star-1}(\xi^{(\star-2)}) \subset \R^3.
\end{equation}

We now  introduce the following sequence of events $\{\Omega_k\}$, the analogous of \eqref{def:A1},  \eqref{def:A2_1},  \eqref{def:Ak} in this setting. More precisely
\begin{equation}\label{def:A1_app}
 \Omega_1:=\{\omega\,:\, \exists\, x_k\in\omega\; \text{s.t.} \; x_k\in B_1(\xi^{(\star-1)})\},
\end{equation}
\begin{equation}\label{def:A2_app}
 \Omega_2:=\{\omega\in \Omega_1\,:\, \exists\, x_k\in\omega\; \text{s.t.} \; x_k\in B_2(\xi^{(1)}) \}.
\end{equation}
and by iterating
\begin{equation}\label{def:Ak_app}
 \Omega_k:=\{\omega\in \Omega_{k-1}\,:\,  \exists\, x_j\in\omega\; \text{s.t.} \; x_j\in B_k(\xi^{(k-1)})\}.
\end{equation}
We observe that the sequence above is such that $ \Omega_{k+1}\subset  \Omega_k$ $\forall k$.  
Our strategy to prove Proposition \ref{prop:wellpos1} will be to show that $\sum_{n=1}^{\infty}\PP( \Omega_n)<\infty$ and to apply then Borel-Cantelli Lemma.

We now compute $\PP(\Omega_n)$. We start from 
$$\PP(\Omega_1)=\sum_{n=1}^{\infty}e^{-\phi |B_1(\xi^{(0)})|}\frac{\phi^n}{n!} (|B_1(\xi^{(k-1)})|)^n.$$
By iterating, at level $k\geq 1$ we have
\begin{equation}\label{eq:prOmk-app}
\begin{split}
\PP(\Omega_k)= &\sum_{n_1=1}^{\infty}\sum_{n_2=1}^{\infty}\dots \sum_{n_k=1}^{\infty}\frac{\phi^{n_1+n_2+\dots +n_k}}{n_1!\, n_2! \dots n_k!}\int_{B_1(\xi^{(0)}))^{n_1}}d\b{x}^{(1)}\int_{B_2(\xi^{(1)}))^{n_2}}\dots \int_{B_k(\xi^{(k-1)}))^{n_k}} d\b{x}^{(k)}\\&
\quad e^{-\phi |B_1(\xi^{(0)})|-\phi|B_2({\xi}^{(1)})| -\dots-\phi |B_{k}({\xi}^{(k-1)})|}.
\end{split}
\end{equation}  
In order to estimate $\PP(\Omega_k)$ we control $ |B_{k}({\xi}^{(k-1)})|$.
We consider the displacement of the position and the volume of the tagged particle in the coalescence step.
More precisely we have 
\begin{align}\label{eq:bdYk-app}
&|Y_{l+1}-Y_l |\leq \frac{n_{l+1}(R_l+1)}{V_l+n_{l+1}},\quad R_l=\frac{3}{4\pi}V_l^{\frac 1 3},\\&\label{eq:bdVk-app}
V_{l+1}=V_l+n_{l+1}.
\end{align}
Moreover, since $B_{l+2}({\xi}^{(l+1)})\subset W_{l+2}( \omega, {\xi}^{(l+1)})\setminus W_{l+1}( \omega, {\xi}^{(l)})$ it follows that 
\begin{equation}\label{eq:boundB_l+2}
|B_{l+2}({\xi}^{(l+1)})|\leq | W_{l+2}( \omega, {\xi}^{(l+1)})\setminus W_{l+1}( \omega, {\xi}^{(l)})| .
\end{equation}
We now claim that 
\begin{equation}\label{app:eqclaim}
|B_{l+2}({\xi}^{(l+1)})|\leq  K n_{l+1}\quad \forall \; l\geq 0
\end{equation}
with $K>0$ a geometrical constant. We remark that \eqref{app:eqclaim} is the analogous of Lemma \ref{lem:geometriclemB_k}.
We set $\Delta:=W_{l+2}( \omega, {\xi}^{(l+1)})\setminus W_{l+1}( \omega, {\xi}^{(l)})$. The proof of \eqref{app:eqclaim} follows from a simple geometrical argument. See Figure \ref{fig:FbdRg}.  
 \begin{figure}[ht]
\centering
\includegraphics[scale= 0.27]{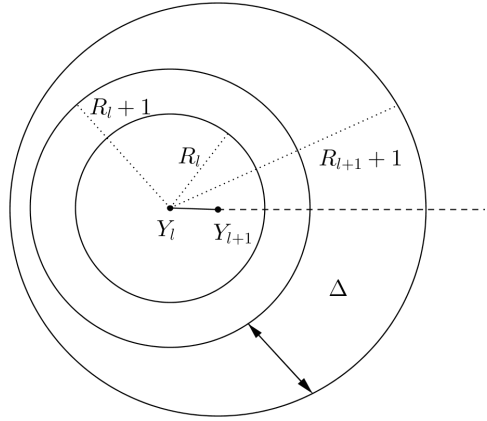}
\caption{The domain $\Delta$.}
\label{fig:FbdRg}
\end{figure}

We consider separately two different cases. 
\begin{itemize}
\item[i)] If $n_{l+1}\geq 2 V_l$ it is straightforward to see that $|\Delta|\leq K n_{l+1}$, where $K>0$ is a geometrical constant. 
\item[ii)] If $n_{l+1}\leq 2 V_l$ then \eqref{eq:bdYk-app} implies $|Y_{l+1}-Y_{l}|\leq C \, \frac{n_{l+1}}{R_l^2}$ and using Taylor approximation formula we have $(V_l+n_{l+1})^{\frac 1 3 }-V_l^{\frac 1 3 }\leq C\, \frac{n_{l+1}}{R_l^2}$. Simple geometry then shows that also in this case $|\Delta|\leq K n_{l+1}$, where $K>0$ is a geometrical constant. 
%$(V_l+n_{l+1})^{\frac 1 3 }-V_l^{\frac 1 3 }=R_l(1+\frac{n_{l+1}}{R_l^3})^{\frac 1 3 }\simeq R_{l}+\frac{n_{l+1}}{R_l^2}$ which imply that $\Delta\leq C n_{l+1}$.
\end{itemize}
We have then proved \eqref{app:eqclaim}. We now come back to the estimate of $\PP(\Omega_k)$. We firstly consider the main contribution for $\phi$ small in \eqref{eq:prOmk-app}, which is given by the terms with $n_1=\dots=n_k=1$. We have
$$\phi|B_1(\xi^{(0)})| e^{-\phi |B_1(\xi^{(0)})|}\phi |B_2(\xi^{(1)})| e^{-\phi|B_2({\xi}^{(1)})|} \dots \phi |B_k(\xi^{(k-1)})| e^{-\phi |B_{k}({\xi}^{(k-1)})|}\leq C(V_0+1)(K\phi)^k$$
where we used \eqref{app:eqclaim} in the last inequality.
We consider now the $k$-th contribution in $\PP(\Omega_k)$. We distinguish between the cases  $|B_k(\xi^{(k-1)})|\leq \frac{1}{\sqrt{\phi}}$ and $|B_k(\xi^{(k-1)})|> \frac{1}{\sqrt{\phi}}$. We obtain
\begin{equation}\label{eq:PAkchar_app}
\begin{split}
\PP(\Omega_k)= &\sum_{n_1=1}^{\infty}\sum_{n_2=1}^{\infty}\dots \sum_{n_k=1}^{\infty}\frac{\phi^{n_1+n_2+\dots +n_k}}{n_1!\, n_2! \dots n_k!}\int_{B_1(\xi^{(0)})^{n_1}}d\b{x}^{(1)}\int_{B_2(\xi^{(1)}))^{n_2}}\dots \int_{B_k(\xi^{(k-1)}))^{n_k}} d\b{x}^{(k)}\\&
\quad e^{-\phi |B_1(\xi^{(0)})|-\phi|B_2({\xi}^{(1)})| -\dots-\phi |B_{k}({\xi}^{(k-1)})|}\\&
\quad \chi(\{\xi\, :\,|B_k(\xi^{(k-1)})|\leq \frac{1}{\sqrt{\phi}} \})+\chi(\{\xi\, :\,|B_k(\xi^{(k-1)})|>\frac{1}{\sqrt{\phi}} \})\\&
=:\PP(\Omega_k)_{(1)}+\PP(\Omega_k)_{(2)}. 
\end{split}
\end{equation}
From the $n_k$-th contribution in $\PP(\Omega_k)_{(1)}$ we get
$$\sum_{n_k=1}^{\infty} \frac{\phi^{n_k}}{n_k!} e^{-\phi |B_k(\xi^{(k-1)})|}|B_k(\xi^{(k-1)})|^{n_k}\leq (1-e^{-\phi |B_k(\xi^{(k-1)}|})\leq \sqrt{\phi}.$$
Therefore we obtain
$$\PP(\Omega_k)_{(1)}\leq \sqrt{\phi}\,\PP(\Omega_{k-1}).$$
For what concerns $\PP(\Omega_k)_{(2)}$ we note that since $\frac{1}{\sqrt{\phi}}<|B_{k}({\xi}^{(k-1)})|\leq K n_{k-1}$ it follows that $n_{k-1} \geq \frac{1}{K\sqrt{\phi}}$. Hence  %estimating \sum_{n_k}\leq 1
\begin{equation}
\begin{split}
\PP(\Omega_k)_{(2)}&
\leq\sum_{n_1=1}^{\infty}\sum_{n_2=1}^{\infty}\dots \sum_{n_{k-1}\geq \frac{1}{K\sqrt{\phi}}} \frac{\phi^{n_1+n_2+\dots +n_{k-1}}}{n_1!\, n_2! \dots n_{k-1}!}\int_{B_1(\xi^{(0)})^{n_1}} \hspace{-0.6cm}d\b{x}^{(1)}\int_{B_2(\xi^{(1)}))^{n_2}} \hspace{-0.6cm}\dots \int_{B_{k-2}(\xi^{(k-3)}))^{n_{k-2}}}  \hspace{-0.6cm}d\b{x}^{(k-2)} \\&
\quad\;  |B_{k-1}({\xi}^{(k-2)}|e^{-\phi |B_1(\xi^{(0)})|-\phi|B_2({\xi}^{(1)})| -\dots-\phi |B_{k-1}({\xi}^{(k-2)})|}.
\end{split}
\end{equation}
We now distinguish between the cases $|B_{k-1}({\xi}^{(k-2)}|\leq\frac{\xi_*}{2K\phi^{\frac 3 2}}$ and $|B_{k-1}({\xi}^{(k-2)}|>\frac{\xi_*}{2K\phi^{\frac 3 2}}$ where $K$ is as in \eqref{app:eqclaim} and $\xi_*$ is as in Lemma \ref{lem:Poisson} and iterate. 
We then obtain, as in the proof of Proposition \ref{prop:wellpos1}, the inequality \eqref{eq:pOkiter}
where now $J_{N_1}=\sum_{n_1\geq N_1}\frac{\phi^{n_1}|B_1(\xi^{(0)})^{n_1}|}{{n_1}!}e^{-\phi |B_1(\xi^{(0)})|}$.
Therefore as in the proof of Proposition \ref{prop:wellpos1} we get \eqref{eq:pOkep} with $\ep<\frac 1 4$.

Using Lemma \ref{lem:boundprob} we conclude that $\PP( \Omega_k)\leq (4\varepsilon)^{k-k_{*}}$ for any $k\geq k_{*}$. Then $\sum_{k=1}^{\infty}\PP( \Omega_k)<\infty$ and, using Borel-Cantelli Lemma, the proof follows.

\bigskip
\textbf{Acknowledgment.}
We thank Barbara Niethammer, who strongly motivated us to study this problem, for useful discussions and suggestions on the topic. The authors acknowledge support through the CRC 1060
\textit{The mathematics of emergent effects}
at the University of Bonn that is funded through the German Science
Foundation (DFG).
%%%%%


\begin{thebibliography}{777777}  \label{bib}

\bibitem{AS}%\bibitem[AS]{AS}
M. Abramowitz and I.A. Stegun,  \emph{Handbook of Mathematical functions}. Dover editions, 1972.

\bibitem{BNP}%\bibitem[BNP]{BNP}
G. Basile, A. Nota and M. Pulvirenti, A Diffusion Limit for a Test Particle in a Random Distribution of Scatterers. \emph{J. Stat. Phys.}, \textbf{155} (6), 1087-1111, 2014.

\bibitem{BNPP} %\bibitem[BNPP]{BNPP} 
G. Basile, A. Nota, F. Pezzotti and M. Pulvirenti, Derivation of the Fick's Law for the Lorentz Model in a Low Density Regime. \emph{Comm. Math. Phys.}, \textbf{336} (3), 1607-1636, 2015.

\bibitem{B}%\bibitem[B]{B}
P. Billingsley, Probability and Measure. Third Edition. \emph{Wiley Series in Probability and Statistics}, \textbf{245}, John Wiley $\&$ Sons, Inc., New York, 1995.


\bibitem{BBS}%\bibitem[BBS]{BBS}
C. Boldrighini, L.A. Bunimovich, and Y.G. Sinai,
On the Boltzmann equation for the Lorentz gas.
 \emph{J. Stat. Phys.} \textbf{ 32}, 477-501, 1983.

\bibitem{DF} %\bibitem[DF]{DF} 
M. Deaconu and N. Fournier, Probabilistic approach of some discrete and continuous coagulation equations with diffusion. \emph{Stoch. Proc. Appl.} \textbf{101}, 83-111, 2002.

\bibitem{DP}%\bibitem[DP]{DP}
L. Desvillettes and M. Pulvirenti,
The linear Boltzmann equation for long-range forces: a derivation from particle systems. 
\emph{Math. Models Methods Appl. Sci. } {\textbf 9}, 1123-1145, 1999.

\bibitem{DR}%\bibitem[DR]{DR}
L. Desvillettes and V. Ricci, 
A rigorous derivation of a linear kinetic equation of Fokker-Planck type in the limit of grazing collisions.
\emph{J. Stat. Phys.} \textbf {104}, 1173-1189, 2001.

\bibitem{EP}%\bibitem[EP]{EP}
M. Escobedo and F. Pezzotti, Propagation of chaos in a coagulation model. \textit{Math. Models Methods Appl. Sci.} \textbf{23}, 1143-1176, 2013.

\bibitem{FG} %\bibitem[FG]{FG} 
N. Fournier and J.S. Giet, Convergence of the Marcus-Lushnikov process. \emph{Methodol. Comput. Appl. Probab.} \textbf{6}, 219-231, 2004.

\bibitem{F} %\bibitem[F]{F} 
S. Friedlander, \emph{Smoke, Dust and Haze: Fundamentals of Aerosol Dynamics. Topics in Chemical Engineering.} Oxford University Press, second edition, 2000.

\bibitem{G}%\bibitem[G]{G}
G. Gallavotti,
Grad-Boltzmann limit and Lorentz's Gas.  \emph{Statistical Mechanics. A short treatise.}
Appendix 1.A2. Springer, Berlin, 1999.

\bibitem{GS-RT}%\bibitem[GS-RT]{GS-RT}
I. Gallagher, L. Saint-Raymond and B. Texier, From Newton to Boltzmann : the case of hard-spheres and
short-range potentials. \emph{Z\"urich Lectures in Advanced Mathematics} \textbf{18}, 2014. 

\bibitem{Gr}%\bibitem[Gr]{Gr} 
G. Grimmett, \emph{Percolation}. Grundlehren
der Math. Wissenschaften, \textbf{321}, Second
edition, Springer-Verlag, Berlin, 1999,

\bibitem{H}%\bibitem[H]{H}
P. Hall, On continuum percolation. \emph{Ann. Probab.} \textbf{13} (4), 1250-1266, 1985.

\bibitem{HR1}%\bibitem[HR1]{HR1}
A. Hammond and F. Rezakhanlou, Kinetic limit for a system of coagulating planar Brownian particles. \emph{J. Stat. Phys.}, \textbf{124}, 997-1040, 2006.

\bibitem{HR2}%\bibitem[HR2]{HR2}
A. Hammond and F. Rezakhanlou, The kinetic limit of a system of coagulating Brownian particles. \emph{Arch. Rat. Mech. Anal.} \textbf {185} (1),1-67, 2007.
%

\bibitem{La}%\bibitem[La]{La}
O. E. Lanford, Time evolution of large classical systems.
\emph{Dynamical systems,
theory  and  applications}, Lecture  Notes  in  Physics,  ed.  J.  Moser, \textbf{38}, 1-111, Springer-Verlag, Berlin, 1975.

\bibitem{LN} %\bibitem[LN]{LN} 
R. Lang and X.-X. Nguyen, Smoluchowski's theory of coagulation in colloids holds rigorously in the Boltzmann-Grad-limit. \emph{Z. Wahrsch. Verw. Gebiete} \textbf {54} (3), 227-280, 1980.

\bibitem{Li}%\bibitem[Li]{Li} 
T.M. Liggett, Interacting Particle Systems, Springer, 1985.

\bibitem{Lo}%\bibitem[Lo]{Lo}
H.A. Lorentz, The motion of electrons in metallic bodies. \emph{Proc. Acad. Amst.} \textbf{7}, 438-453, 1905.

\bibitem{Lu1}%\bibitem[Lu1]{Lu1}
A. Lushnikov, Coagulation  in finite  systems.  \emph{J. Coll. Int. Sc.} \textbf{65}, 276-285, 1978.

\bibitem{Lu2}%\bibitem[Lu2]{Lu2}
A. Lushnikov, Some new aspects of coagulation theory. \emph{Izv. Akad. Nauk SSSR Ser. Fiz. Atmosfer. i Okeana} \textbf{14}, 738-743, 1978.

\bibitem{M}%\bibitem[M]{M}
A. H. Marcus, Stochastic coalescence. \emph{Technometrics} \textbf{10}, 133-143, 1968.

\bibitem{MR}%\bibitem[MR]{MR}
R. Meester, R. Roy,
%\emph{Continuum Percolation}, Cambridge University Press (1996).
Continuum Percolation. Cambridge University Press, 1996.

\bibitem{N} %\bibitem[N]{N} 
J. Norris, Brownian coagulation. \emph{Comm. Math. Sci.}, \textbf{2} (1), 93-101, 2004.

\bibitem{No}
%\bibitem[N]{N} 
A. Nota, Diffusive limit for the random Lorentz gas. \emph{From Particle Systems to Partial Differential Equations II}, Springer Proceedings in Mathematics $\&$ Statistics, \textbf{129}, 273-292,  2015.

\bibitem{PSS}%\bibitem[PSS]{PSS}
M. Pulvirenti, C. Saffirio and S. Simonella, On the validity of the Boltzmann equation for short range potentials. \emph{Rev. Math. Phys.} \textbf {26} (02), 1450001, 2014.


\bibitem{Sm} %\bibitem[Sm]{Sm} 
M. Smoluchowski, Drei Vortr\"age \"uber Diffusion, Brownsche Molekularbewegung und
Koagulation von Kolloidteilchen. 
\emph{Physik. Zeitschrift}, \textbf {17}, 557-599, 1916.

\bibitem{Sp}%\bibitem[Sp]{Sp}
H. Spohn,
The Lorentz flight process converges to a random flight process.
\emph{Comm. Math. Phys.} \textbf {60}, 277-290, 1978.


\bibitem{TH}%\bibitem[TH]{TH}
E.Trizac and J.P. Hansen, Dynamic Scaling Behavior of Ballistic Coalescence. \emph{Phys. Rev. Lett.},\textbf{74} (21), 4114-4117, 1995.

\bibitem{TH}%\bibitem[TH]{TH}
E.Trizac and J.P. Hansen, Dynamics and Growth of Particles Undergoing Ballistic Coalescence. \emph{J. Stat. Phys.}, \textbf{82} (5), 1345-1370, 1996. 

\bibitem{YRH} %\bibitem[YRH]{YRH} 
M.R. Yaghouti,  F. Rezakhanlou and A. Hammond, Coagulation, diffusion and the continuous Smoluchowski equation. \emph{Stoch. Proc. Appl.}, \textbf{119} (9), 3042-3080, 2009.


%\bibitem[W1]{W1} W. Wagner, Explosion phenomena in stochastic coagulation-fragmentation models, \emph{Ann. Appl. Probab.} \textbf{15}, 2081-2112, 2005.

%\bibitem[W2]{W2} W. Wagner, Post-gelation behavior of a spatial coagulation model, \emph{Electronic Journal of Probability}, \textbf{11}, 893-933, 2006.

\end{thebibliography}
\end{document}